\documentclass{article}

\usepackage[english]{babel}

\usepackage[letterpaper,top=2cm,bottom=2cm,left=3cm,right=3cm,marginparwidth=1.75cm]{geometry}
\usepackage{subcaption}
\usepackage{amsmath}
\usepackage{graphicx}

\usepackage{amsthm,amsmath,amssymb,amsfonts,amssymb,mathtools}
\usepackage{color}
\usepackage{xcolor}
\usepackage{empheq}
\usepackage{dsfont}
\usepackage{mathrsfs}
\usepackage{graphicx}
\usepackage{enumitem}
\usepackage[T1]{fontenc}
\usepackage{microtype}
\usepackage{xspace}  
\usepackage[most]{tcolorbox}
\usepackage{bm}
\usepackage{tikz}
\usetikzlibrary{arrows.meta,calc,shapes.geometric}

\usepackage{ragged2e}

\usepackage{soul}

\usepackage{tikz}
\usetikzlibrary{patterns}
\usetikzlibrary{arrows,shapes,automata,backgrounds,petri,positioning}
\usetikzlibrary{shadows}
\usetikzlibrary{calc}
\usetikzlibrary{spy}
\usepackage{pgf,pgfplots}
\usetikzlibrary{angles, quotes}
\usepackage{pgfmath,pgffor}

\usepackage{framed}
\usepackage{multirow}

\usepackage{algorithm}

\usepackage[noend]{algpseudocode}

\algnewcommand{\LineComment}[1]{\Statex $\triangleright$ #1}

\usepackage[colorlinks,citecolor=blue,linkcolor=blue,bookmarks=true]{hyperref}

\usepackage[nameinlink]{cleveref}
\crefname{ineq}{Inequality}{Inequality}
\creflabelformat{ineq}{#2{\upshape(#1)}#3}
\crefname{sub}{Subsection}{Subsection}
\creflabelformat{Subsection}{#2{\upshape(#1)}#3}
\crefname{sdp}{SDP}{SDP}
\creflabelformat{sdp}{#2{\upshape(#1)}#3}
\crefname{lp}{LP}{LP}
\creflabelformat{lp}{#2{\upshape(#1)}#3}

\def\colorful{0}

\newcommand\blfootnotea[1]{\begingroup
  \renewcommand\thefootnote{}\footnote{#1}\endgroup
}

\ifnum\colorful=1
\newcommand{\new}[1]{{\color{red} #1}}

\else
\newcommand{\new}[1]{{#1}}

\fi

\usepackage{thm-restate}

\newtheorem{theorem}{Theorem}[section]
\newtheorem{assumption}{Assumption}

\newtheorem{lemma}[theorem]{Lemma}
\newtheorem{informal theorem}[theorem]{Theorem (informal statement)}

\newtheorem{proposition}[theorem]{Proposition}

\newtheorem{claim}[theorem]{Claim}
\newtheorem{fact}[theorem]{Fact}

\newtheorem{remark}[theorem]{Remark}

\newtheorem{definition}[theorem]{Definition}
\newcommand{\eqdef}{\coloneqq}
\crefname{question}{question}{questions}

\newcommand{\Ey}{\E_{(\x,y)\sim \D}}

\newcommand{\CS}{Cauchy--Schwarz\xspace}

\newcommand{\spanned}{\mathrm{Span}}

\newcommand\norm[1]{\left\| #1 \right\|}

\renewcommand\vec[1]{\mathbf{#1}}
\DeclareMathOperator*{\pr}{\mathbf{Pr}}
\DeclareMathOperator*{\E}{\mathbf{E}}
\newcommand{\proj}{\mathrm{proj}}

\newcommand{\normal}{\mathcal{N}}

\newcommand{\spd}{{\mathbb{S}^{d-1}}}
\newcommand{\ot}{\otimes}

\DeclareMathOperator*{\argmin}{argmin}

\newcommand{\bx}{\mathbf{x}}
\newcommand{\by}{\mathbf{y}}
\newcommand{\bv}{\mathbf{v}}
\newcommand{\bu}{\mathbf{u}}
\newcommand{\bz}{\mathbf{z}}
\newcommand{\bs}{{\mathbf{s}}}
\newcommand{\bw}{\mathbf{w}}
\newcommand{\br}{\mathbf{r}}
\newcommand{\bp}{\mathbf{p}}

\newcommand{\e}{\mathbf{e}}
\newcommand{\Sp}{\mathbb{S}}
\newcommand{\bh}{\mathbf{h}}
\newcommand{\bt}{\mathbf{t}}

\newcommand{\x}{\mathbf{x}}
\newcommand{\w}{\mathbf{w}}
\newcommand{\g}{\mathbf{g}}

\newcommand{\evt}{\mathcal{E}}
 
\newcommand{\he}{\mathbf{He}}
\newcommand{\hep}{\mathrm{he}}

\newcommand{\mM}{\mathbf{M}}
\newcommand{\mP}{\mathbf{P}}
\newcommand{\mQ}{\mathbf{Q}}
\newcommand{\mA}{\mathbf{A}}
\newcommand{\mB}{\mathbf{B}}

\newcommand{\mH}{\mathbf{H}}

\newcommand{\mI}{\mathbf{I}}
\newcommand{\mR}{\mathbf{R}}
\newcommand{\mF}{\mathbf{F}}
\newcommand{\mW}{\mathbf{W}}
\newcommand{\mT}{\mathbf{T}}
\newcommand{\mS}{\mathbf{S}}
\newcommand{\mU}{\mathbf{U}}
\newcommand{\mV}{\mathbf{V}}

\newcommand{\Real}{\mathrm{Re}}

\newcommand{\R}{\mathbb{R}}

\newcommand{\Z}{\mathbb{Z}}

\newcommand{\eps}{\epsilon}

\newcommand{\poly}{\mathrm{poly}}

\newcommand{\calN}{{\cal N}}

\newcommand{\la}{\langle}
\newcommand{\ra}{\rangle}

\newcommand{\D}{\mathcal{D}}

\newcommand{\Ind}{\mathds{1}}
\newcommand{\1}{\Ind}

\newcommand{\littlesum}{\mathop{\textstyle \sum}}

\newcommand{\wt}{\widetilde}
\newcommand{\wh}{\widehat}
\newcommand{\bq}{\mathbf{q}}
\newcommand{\bg}{\mathbf{g}}

\newcommand{\Exy}{\E_{(\x,y)\sim \D}}

\newcommand{\inner}[2]{\la#1,\, #2\ra}

\newcommand*\diff{\mathop{}\!\mathrm{d}}

\def\vzero{{\bm{0}}}

\def\bu{{\mathbf{u}}}

\def\vx{{\mathbf{x}}}

\def\cA{{\mathcal{A}}}

\def\cC{{\mathcal{C}}}
\def\cD{{\mathcal{D}}}
\def\cE{{\mathcal{E}}}

\def\cI{{\mathcal{I}}}

\def\cL{{\mathcal{L}}}

\def\cN{{\mathcal{N}}}

\def\cU{{\mathcal{U}}}
\def\cV{{\mathcal{V}}}

\def\sF{{\mathbb{F}}}

\def\sU{{\mathbb{U}}}
\def\sV{{\mathbb{V}}}
\def\sW{{\mathbb{W}}}

\def\sZ{{\mathbb{Z}}}

\def\sfF{{\mathsf{F}}}

\def\sfS{{\mathsf{S}}}

\newcommand{\sat}{{\mathsf{Test}}}

\newcommand{\init}{\mathsf{CDR}}
\newcommand{\test}{\mathsf{Test}}

\newcommand{\tx}{{\wt{\x}}}
\newcommand{\im}{{\mathrm{i}}}

\title{Learning Positive Linear Combinations of ReLUs\\ 
in Fully Polynomial Time\blfootnotea{Authors are listed in alphabetical order.}}
\author{
Ilias Diakonikolas\thanks{Supported by NSF Medium Award CCF-2107079, 
ONR Award number N00014-25-1-2268, an H.I. Romnes Faculty Fellowship 
and a Guggenheim Fellowship.}\\
UW Madison\\
{\tt ilias@cs.wisc.edu}\\
\and Jelena Diakonikolas\thanks{Supported in part by the AFOSR YIP award FA9550-24-1-0076, by the NSF MFAI Award DMS-2502282, and by the NSF CAREER Award CCF-2440563. Any opinions, findings and conclusions or recommendations expressed in this material are those of the author(s) and do not necessarily reflect the views of the U.S. Department of Defense.}\\
UW Madison\\
{\tt jelena@cs.wisc.edu}\\
\and Puqian Wang\thanks{Supported in part by the AFOSR YIP award FA9550-24-1-0076.} \\
UW Madison\\
{\tt pwang333@wisc.edu}\\
\and
Nikos Zarifis\thanks{Supported by MIT METEOR fellowship. Some of this research was done while the author was at UW Madison, supported by  NSF Medium Award CCF-2107079.} 
\\
CSAIL, MIT\\
{\tt zarifis@mit.edu}
}
\date{}

\begin{document}

\maketitle

\begin{abstract}
We study the fundamental problem of PAC learning positive linear combinations of $k$ 
arbitrary ReLU activations on $\R^d$ with respect to the Gaussian distribution, in $L_2$ 
loss. We give the first fully-polynomial time algorithm for this learning task. 
Specifically, our algorithm draws $\mathrm{poly}(d, k, 1/\epsilon)$ samples, 
runs in $\mathrm{poly}(d, k, 1/\epsilon)$ time, and outputs a hypothesis that is $\epsilon$-close to the 
target concept with high probability. All prior algorithms for this problem require 
super-polynomial time in $k$ or $1/\epsilon$, or leverage strong well-conditioned assumptions on the weight matrix.

On a technical level, our approach develops a novel boosting procedure of 
a ``weak'' learner that recovers the hidden weight vectors to constant error, 
in a list-decodable manner. To develop this ``weak'' learning procedure, we leverage 
implicit tensor computation techniques to construct a low-degree filtering procedure 
based on Hermite tensors of the target function. 
This procedure amplifies separation among 
(clusters of) weight vectors to learn each of them to constant error. 
We then argue that the same procedure can be used for improving the estimation 
of hidden weight vectors, using gradient descent-like updates on the Euclidean sphere. 
A crucial ingredient of our algorithm is a pruning approach that ensures 
only relevant vectors that approximate the hidden weight vectors at target accuracy 
are maintained by the algorithm.   
\end{abstract}

\thispagestyle{empty}

\newpage

\tableofcontents

\thispagestyle{empty}

\newpage

\setcounter{page}{1}

\section{Introduction}\label{sec:intro}

Understanding the computational complexity of learning neural networks
is a central question in the foundations of machine learning.
Over the past decade, a substantial body of work has developed provable
learning algorithms under various structural and distributional
assumptions; see, e.g.,~\cite{Janz15, ZhangLJ16, ZhongS0BD17, GeLM18,
GoelK19, CKM22,  WZDD23, ZWDD25, DIKR25, DLB25}. Despite this progress, a basic question has remained open even for
networks with a single hidden layer: Can positive linear combinations
of ReLUs be learned in polynomial time under Gaussian inputs, without strong 
structural assumptions on the hidden weights?
This longstanding question is explicitly highlighted, for example,
by Chen, Klivans, and Meka~\cite{CKM22}
(see also Chen's thesis~\cite{Chen-thesis}):
\begin{quote}
``Ideally, we would like an algorithm with sample complexity and
running time that is polynomial in all the relevant parameters.
Even for learning arbitrary sums of ReLUs, [\ldots],
it remains \textbf{a major open question} to obtain a
polynomial-time algorithm [\ldots].''
\end{quote}
Positive linear combinations of ReLUs are arguably the simplest
extension of learning a single neuron---the setting of generalized
linear models. A fully polynomial-time algorithm for this basic class 
is an important step towards understanding the learnability of more
general networks. Moreover, the techniques developed for this problem 
may also be useful for learning other latent variable models.

In this paper, we study the task of PAC learning positive linear combinations
of at most $k$ ReLUs on $\R^d$ under the standard Gaussian distribution.
Specifically, the learner has access to independent labeled examples from an unknown 
function in the concept class, and its goal is to output a hypothesis that approximates 
the target in $L_2$ norm.
The relevant parameters are the dimension $d$, the number of hidden
units $k$, and the inverse error $1/\eps$.
It is well-known that polynomially many samples 
information-theoretically suffice for this learning task, ignoring computational complexity.
The challenge is to develop a computationally efficient learning algorithm: 
despite a sequence of algorithmic advances, the best previous running-time bounds 
were super-polynomial in at least one of these parameters.

The work of~\cite{DKKZ20} gave the first learner for this problem
with running time polynomial in $d$ and $1/\eps$ for every fixed $k$.
Their algorithm uses $\widetilde{O}(d)\poly(k/\eps)$ samples and runs
in time $\poly(d/\eps)+(k/\eps)^{O(k^2)}$.
The work of~\cite{diakonikolas2020small} improved the running time
using a structural theorem on small covers for near-zero sets of
polynomials. Its sample and computational complexities are $\poly(dk/\eps)+(k/\eps)^{O(\log^2 k)}$, i.e., 
quasi-polynomial in $k$ even for constant error. 
More recently, \cite{diakonikolas2025implicit} developed an algorithm
based on implicit high-order moment estimation with sample and
computational complexity $\poly(d,k)\,2^{\poly(1/\eps)}$.
That is, while this achieves polynomial dependence on $d$ and $k$, 
the dependence on $1/\eps$ is exponential.
In summary, neither of these results gives a fully-polynomial time algorithm.

The above discussion motivates the following question:
\begin{center}
\begin{minipage}{0.92\linewidth}
\centering\em
Can positive linear combinations of $k$ ReLUs on $\R^d$ be learned
under Gaussian inputs\\ with sample and computational complexity 
$\poly(d,k,1/\eps)$?
\end{minipage}
\end{center}

{\em As our main contribution, we answer this question in the affirmative.}
Specifically, we give a learning algorithm with sample and computational complexity 
$\poly(d,k,1/\eps)$ that works for arbitrary positive linear combinations of ReLUs. 
\new{Our algorithm outputs a simple interpretable hypothesis. Moreover, 
the technical tools we develop in the process may be useful for learning other latent variable models.}

\subsection{Problem Formulation and Main Result}
\label{ssec:main-result}

Before stating our main result, we define the underlying concept class 
and the PAC learning task.

\begin{definition}[Positive linear combinations of ReLUs]
\label{def:positive-relu-class}
Let $\sigma(t)\eqdef\max\{0,t\}$ denote the ReLU activation.
For $d,k\in\mathbb{Z}_{+}$, let $\mathcal{C}_{d,k}$ denote the class
of functions $f:\R^d\to\R$ of the form
$f(\x)=\sum_{i=1}^k\beta_i\sigma(\vec w_i\cdot\x)$,
where $\beta_i\in[0,1]$ and $\|\vec w_i\|_2=1$ for all $i\in[k]$.
We call the $\vec w_i$'s the \emph{hidden weight vectors} and the
$\beta_i$'s the \emph{coefficients}.
\end{definition}

By positive homogeneity of the ReLU, the norm of a nonzero hidden
weight vector can be absorbed into its coefficient.
The constraint $\beta_i\leq1$ fixes the scale (without loss of generality); see
Remark~\ref{rem:normalization}.
For a function $g:\R^d\to\R$, we write
$\|g\|_{L_2}\eqdef
(\mathbf{E}_{\x\sim\calN(\vec 0,\mI_d)}[g(\x)^2])^{1/2}$.

\begin{definition}[PAC learning $\mathcal{C}_{d,k}$ under Gaussian marginals]
\label{def:gaussian-pac-learning}
The learner is given accuracy and confidence parameters
$\eps,\delta\in(0,1)$ and access to independent examples
$(\x,y)\sim\D$, where $\D$ is a distribution supported on
$\R^d\times\R$ such that $\x\sim\calN(\vec 0,\mI_d)$ and
$y=f(\x)$ for an unknown $f\in\mathcal{C}_{d,k}$.
The goal of the learner is to output a hypothesis $h:\R^d\to\R$ that, with
probability at least $1-\delta$, satisfies
\[
    \|h-f\|_{L_2}\leq\eps\|f\|_{L_2}.
\]
The hypothesis need not belong to $\mathcal{C}_{d,k}$.
The learner is called \emph{proper} if it outputs a hypothesis
in $\mathcal{C}_{d,k}$.
\end{definition}

Our main result is summarized in the following theorem 
\new{(see \Cref{thm:ultimate-main} for a detailed full version)}. 

\begin{theorem}[Main result, informal]
\label{thm:main-informal}
There is an algorithm that, given $d,k\in\mathbb{Z}_{+}$,
$\eps\in(0,1)$, and access to independent examples
$(\x,f(\x))$, where $\x\sim\calN(\vec 0,\mI_d)$ and
$f\in\mathcal{C}_{d,k}$, uses $\poly(d,k,1/\eps)$ samples and time
and outputs a hypothesis $h$ such that, with probability at least $2/3$,
\[
    \|h-f\|_{L_2}\leq\eps\|f\|_{L_2}.
\]
\end{theorem}

A few remarks are in order.
While our learner is not proper, its hypothesis has a simple form:
a linear function plus a linear combination of at most $k$
absolute-value functions, i.e., 
$h(\x)=\vec g\cdot\x+ \sum_{i=1}^k \alpha_i|\vec v_i\cdot\x|$. 
Equivalently, using $|t|=\sigma(t)+\sigma(-t)$, it is a linear
function plus a linear combination of at most $2k$ ReLUs.
Thus, the hypothesis is itself a small network with $O(k)$ nonlinear
units.
For comparison, the improper learner from~\cite{diakonikolas2020small}
outputs a positive linear combination of 
quasi-polynomially many ReLUs, while the learner 
from~\cite{diakonikolas2025implicit} outputs a polynomial represented
through compressed Hermite coefficients.
For these representations, the bounds on the time needed to even \emph{evaluate} 
the hypothesis on a single new example are also super-polynomial
in $k$ or $1/\eps$.

\begin{remark}[Additive label noise] 
{\em More generally, one may consider examples $(\x,y)$ with
$\x\sim\calN(\vec 0,\mI_d)$ and $y=f(\x)+\xi$, where $\xi$ is
independent of $\x$, has mean zero, and satisfies
$\mathbf{E}[\xi^2]\leq\sigma_\xi^2$.
Following the convention of~\cite{DKKZ20}, the 
objective in this setting is to output a hypothesis $h$ satisfying 
\[
    \|h-f\|_{L_2}^2
    \leq\eps^2\bigl(\|f\|_{L_2}^2+\sigma_\xi^2\bigr)
\]
with high probability.
The realizable setting of \Cref{def:gaussian-pac-learning} 
corresponds to $\sigma_\xi=0$. 
\new{Our algorithm straightforwardly extends
to this setting with a $\poly(\sigma_\xi)$ overhead in its complexity.}}
\end{remark}

\begin{remark}[Normalization]\label{rem:normalization}
{\em Our result applies to arbitrary positive linear combinations
of at most $k$ ReLUs, without bounds on the coefficients or the
norms of the hidden weight vectors.
By positive homogeneity, we can absorb the norm of each 
hidden weight vector into its coefficient.
Writing $b$ for the resulting coefficient sum, we have
$\mathbf{E}[f(\x)]=\frac{b}{\sqrt{2\pi}}$ and 
$\mathbf{E}[f(\x)^2]\leq\frac{b^2}{2}$.  
By Chebyshev's inequality, the empirical mean $\overline y$ of
a sufficiently large constant number of labels gives an 
estimate $B=2\sqrt{2\pi}\,\overline y$ satisfying
$b\leq B\leq3b$ with high constant probability.
For $B>0$, we can apply our learner to fresh examples with labels
divided by $B$ and multiply its output by $B$.
The normalized target has unit hidden weight vectors and coefficients
in $[0,1]$, and rescaling preserves the relative-error guarantee.
If $B=0$, we return the zero hypothesis.
Thus, our normalization is without loss of generality and requires
no prior knowledge of the target's scale.}
\end{remark}

\new{
\noindent {\bf Terminology.} Throughout this paper, 
we will say that a distribution $(\x,y)\sim\D$ supported on 
$\R^d \times \R$ is a {\em Gaussian sum of $k$ ReLUs}  
if the $\x$-marginals are distributed according to the standard normal distribution, 
i.e., $\x\sim\calN(\vec 0,\mI_d)$, and if 
$y(\x) = \sum_{i=1}^k \beta_i\sigma(\vec w_i\cdot\x)$, where $\sigma$ 
is the ReLU function, $\new{\beta}_i\in \new{[} 0,1]$, and 
$\|\vec w_i\|_2 = 1$, for all $i \in [k]$. 
}

\subsection{Technical Overview}\label{sec:tech-overview}

To explain the main ideas behind our results, we begin with the following observations. 
First, low-coefficient ReLU components in the sum have little influence on the predicted label: 
in particular, if the coefficient $\beta_i$ associated with component $i$ is smaller than $\sqrt{\eps}/k,$ 
then zeroing out all such components can contribute at most an additive $\eps$ to the $L_2^2$ loss 
we aim to minimize. Thus, setting a threshold $\eps_0 \approx \sqrt{\eps}/k$ and focusing on 
approximating the ReLU components with coefficients above this threshold, 
comes at no loss of generality. In the remainder of the technical overview, 
we only consider such nontrivial components with coefficients $\beta_i \ge \eps_0$, 
\new{and refer to them as ``$\eps_0$-heavy'' in the rest of the paper}.

Second, we claim that it suffices to recover the target vectors up to their orientation 
(i.e., up to their signs). The reason is the decomposition
\(  \sigma(z)-\frac{z}{2}=\frac{|z|}{2}.\) 
Specifically, if we calculate the first-moment parameter 
$\g=\E[y\x]=\frac12\sum_{i=1}^k\beta_i\w_i$, we get 
\begin{equation}\label{eq:abs-sum}
    y-\g\cdot\x=\frac12\sum_{i=1}^k\beta_i|\w_i\cdot\x| \;.
\end{equation}
Thus, after subtracting the linear part, the remaining function is a sum
of absolute values. Its terms are unchanged when $\w_i$ is replaced by
$-\w_i$, so the directional part of the algorithm only needs to recover
one of these two orientations. We measure the error of the angle using
$\theta_\pm(\bv,\w_i)=\min\{\theta(\bv,\w_i),\pi-\theta(\bv,\w_i)\}$,
rather than by the ordinary angle for this reason. 

The birds-eye view of our algorithm is the following. 
The algorithm operates in rounds. 
Each round $t = \{1, 2, \dots, T\}$ of computation maintains the following invariant: 
\begin{tcolorbox}[colback=gray!10, colframe=white, boxrule=0pt]{\bf Algorithm Invariant.}
    {There is a set of vectors $\cV_t$, $|\cV_t| \leq k,$ maintained by the algorithm 
    such that for every \new{$\eps_0$-heavy} hidden vector $\w_i,$ 
    there exists a vector $\bv_i \in \cV_t$ that is close to it; 
    in particular, $\theta_\pm(\bv_i,\w_i) \leq R_t,$ $R_t > 0.$}
\end{tcolorbox} 
Note that this mapping of vectors need not be one-to-one: 
 one vector $\bv$ can be close to multiple $\w_i$'s. 
 Initially, $R_0 \in (0, 1)$ is a constant, and this is achieved 
 using the Coarse Direction Recovery ($\mathsf{CDR}$) subroutine described 
 in \Cref{sec:tech-overview-CDR} below. We argue that we can maintain the stated invariant with $R_{t+1} \leq \rho R_t$ 
for a constant $\rho \in (0, 1)$ and all $t \geq 0,$ 
so that the algorithm approximates all \new{$\eps_0$-heavy} 
hidden weight vectors $\w_i$ with error $\theta_\pm(\bv_i,\w_i) \leq \sqrt{\eps}/k,$ 
after $T = O(\log_\rho(k/\epsilon))$ rounds. The last step is then to estimate 
the \new{coefficients} for the absolute-sum representation \eqref{eq:abs-sum}. 
It is not hard to see that this part boils down to solving a linear regression problem of the form
\[
\min_{\bm{\alpha} \in [0, 1]^m} \Exy\Big[\Big( y- \g \cdot\x-\frac{1}{2}\sum_{j=1}^m\alpha_j|\bv_j\cdot\x|\Big)^2\Big] \;, 
\]
where $m = |\cV_T| \leq k.$ This problem is \new{convex} 
and it is not hard to argue that it is efficiently solvable 
in terms of both sample and computational complexity.

The technical crux of the approach thus lies in implementing the rounds of the above high-level algorithm, which relies upon multiple core technical components described below.

{\subsubsection{Coarse Direction Recovery}\label{sec:tech-overview-CDR} 
A key technical tool in our algorithmic approach is a procedure that, 
given a failure probability $\delta_\init \in (0, 1)$ and access 
to i.i.d.\ examples $(\vx, y(\vx)),$ where $\vx \sim \cN(\vzero, \mI_d)$ and $y(\vx) = f(\vx)$
for a fixed but unknown $f \in \cC_{d, k}$, outputs a polynomially-sized list 
in $k, 1/\delta_\init$ with the following property: 
with probability at least $1- \delta_\init$, for each $\eps_0$-heavy hidden weight vector $\w_i$ 
defining the labeling function, there exists at least one list element $\br_i$ 
such that $|\br_i \cdot \w_i| \geq 0.99$. Crucially, this procedure, to which we refer 
as the Coarse Direction Recovery ($\init$) can be implemented using polynomially 
many labeled examples (in $\new{d},  k, 1/\eps_0 ,1/\delta_\init$), and 
in time polynomial in $d$ and the number of samples. 
$\init$ can be understood as a particularly ``weak''\footnote{\new{Here, ``weak'' is in quotes to distinguish it from the standard notion in PAC learning theory, where weak learning is understood as slightly better than random, whereas we require constant correlation of vectors.}} learner for our problem, 
where each $\eps_0$-heavy hidden weight vector is learned to \emph{constant} error, 
in a \emph{list-decodable} manner. The latter means that the output provided 
by $\init$ is a polynomially-sized list containing useful, constant-error outputs, 
but without any information on how to separate useful outputs from the rest of the list. 
It is not hard to argue that a stronger learner, which would only output vectors 
$\br_i$ with sufficiently strong correlation $|\br_i \cdot \w_i| \geq 1 - \Theta(\sqrt{\eps}/k)$, 
would solve our learning problem stated in \Cref{def:gaussian-pac-learning}. 
Yet, with the techniques described below, such a stronger requirement would lead only 
to a quasi-polynomial algorithm, at best. }

The first insight for developing the $\init$ subroutine comes from the properties 
of Hermite tensors of positive linear combinations of ReLUs 
(see \Cref{sec:prelims} for relevant definitions); 
namely, that Hermite moments of a linear combination of activations can be expressed
as a sum of rank-one tensors, an observation used extensively
in prior work~\cite{DKKZ20,  ChenDGKM23, diakonikolas2025implicit}. 
More specifically, for a positive linear combination of ReLUs, we have
\begin{equation}
    \new{\vec Q_{2p+2}} := \E_{(\x,y)\sim\D}[y\he_{2 p+2}(\x)]\propto 
    \sum_{i=1}^k\beta_i \vec w_i^{\otimes 2p+2}\;, \label{eq:low-degree-tensor-form}
\end{equation}
where $\he_{2 p+2}(\x)$ is the Hermite tensor of order $2p + 2,$ for $p \geq 0.$ The moment formulation here is what enables 
estimation from samples; the fact that it is proportional to the sum of 
rank-one tensors $\sum_{i=1}^k\beta_i \vec w_i^{\otimes 2p+2}$ 
is what we use to obtain the guarantee of $\init.$

Directly recovering the individual rank-one components of a
high-order tensor is, in general, a computationally challenging problem. 
Instead, a useful approach is to contract the tensor along
all but two indices, obtaining a matrix whose leading eigenvector 
can reveal one of the hidden directions. 
To be precise: taking the inner product of \Cref{eq:low-degree-tensor-form} with a tensor $\vec u^{\otimes 2p}$ and reshaping the result as a matrix, we get that 
\begin{equation}
   \vec Q_2(\vec u)= \left\la \vec Q_{2p+2} ,\vec u^{\otimes 2p} \right\ra\propto \sum_{i=1}^k \beta_i |\w_i\cdot\bu|^{2p}\w_i\w_i^{\top} \;. \label{eq:low-degree-matrix-form}
\end{equation}
The key observation is that the coefficients of the rank-one
matrices in \Cref{eq:low-degree-matrix-form} depend on 
\emph{powers of the correlations}, $|\w_i\cdot\bu|^{2p}$.
Consequently, even a small multiplicative advantage in the
correlation with one hidden weight vector can be amplified by
choosing a sufficiently large degree $p$. 
To illustrate this idea, suppose that $\w_1$ is well-separated 
from the remaining hidden vectors
$\w_2, \dots, \w_k$ 
(e.g., $\max_{i\neq 1}|\vec w_1\cdot \vec w_i|\leq 1/2$). 
Then, we can write $\vec Q_2(\vec u)$ as a sum of 
two matrices $  \vec Q_2(\vec u)\propto \beta_1 |\w_1 \cdot \vec u|^{2p}\w_1\w_1^\top + \vec Q_\eps(\vec u)$.
If a vector $\vec u$ is such that the spectral norm of 
$\beta_1 |\w_1 \cdot \vec u|^{2p}\w_1\w_1^\top$ is much larger than the spectral norm of 
$\vec Q_\eps(\vec u)$, then the top eigenvector of $\vec Q_2(\vec u)$ must have constant 
correlation with $\w_1$. A sufficient condition for this to hold 
is to have $\beta_i |\w_i\cdot\bu|^{2p}\gg \sum_{j\neq i}^k \beta_j |\w_j\cdot\bu|^{2p}$. 
Given that we can choose $p$ to be a (sufficiently) large power, 
it suffices to find a $\vec u$ so that for all $i\neq 1$, we have 
$|\vec u\cdot\w_1|\geq (1+\xi)|\vec u\cdot\w_i|$ for some $\xi>0$ 
(for simplicity of this exposition, assume $\xi=1$). 
Then, by choosing $p\gg \log(k \max_i \beta_i/\beta_1)$, we have that the spectral norm of 
$ \beta_1 |\w_1 \cdot \vec u|^{2p}\w_1\w_1^\top$ is at least two times larger  
than the spectral norm of $\vec Q_\eps(\vec u)$. 
Hence, the top eigenvector of $\vec Q_2(\vec u)$ is almost parallel to $\w_1$. 
Importantly, this intuition extends beyond the well-separated case.
Indeed, directions $\w_i$ that are close to $\pm\w_1$, i.e.,
$    |\w_i\cdot\w_1|
    \gg 0,$ can be grouped together with $\w_1$ while retaining the same properties, since
 their contributions align with the desired
rank-one matrix. Therefore, we only need to suppress directions that are
far from $\pm\w_1$, rather than isolate a single component. The hope here is that by sampling $\bu$ from an appropriately selected probability distribution, the above event would happen for all $\w_i,$ with some nontrivial probability.  
This is the main intuition for why the stated guarantee of $\init$ should be possible.

There are at least two main technical challenges with the above argument. 
First, since $\bu$ is chosen from a probability distribution, 
to ensure that we can produce a polynomially-sized list containing 
at least one vector that aligns well with some $\w_i$, 
the event ensuring separation in alignment among hidden weight vectors with $\bu$ 
described above must happen with probability no smaller than an inverse polynomial 
in the problem parameters. 
(Ensuring that such an event is satisfied for all $\w_i,$ $i \in [k]$, 
for at least one of the drawn probing vectors $\bu$ follows by the union bound.) 
We argue that this is indeed the case for the simple distribution where $\bu$ 
is chosen uniformly at random from the $k$-dimensional unit Euclidean sphere, 
defined over the space containing the space spanned by the targets $\w_i,$ $i \in [k]$. 
The same can also be argued for other related distributions\new{, such as standard Gaussians}.

Second, as hinted before \new{and more importantly}, we need to argue that the matrix in \eqref{eq:low-degree-matrix-form} 
can be estimated efficiently from labeled examples with bounded error in the operator norm, 
so the top eigenvector is preserved (up to constant error). For example, when $\bu$ is sampled uniformly from a
$k$-dimensional sphere, its typical correlation with a
fixed unit vector is of the order $1/\sqrt{k}$. 
Raising this correlation to degree $2p$ produces a factor
of order $k^{-p}$, which is quasipolynomially small
when $p=\Theta(\log(k/\eps))$.
Estimating the resulting matrix may therefore require
quasipolynomially small accuracy, and this, as a result, leads to quasipolynomial sample complexity. \new{We provide full technical details for this approach in \Cref{sec:quasi-poly}.}

To obtain a fully polynomial $\init$ implementation, the main challenge is to ensure that the 
estimation task can be implemented efficiently, in terms of both samples and time. 
A major obstacle when sampling $\bu$ uniformly from the sphere 
or from the standard Gaussian distribution is that the variance of the empirical estimate 
of $\mQ_2 = \left\la \vec Q_{2p+2} ,\vec u^{\otimes 2p} \right\ra$ is large 
because the variance of the tensor $\vec u^{\otimes 2p}$ is large for large degree $p$, 
which leads to error amplification. 
\new{
To illustrate this point, consider $p=1$ and decompose the estimation tensor 
into $\mQ_{4}\in (\R^d)^{\ot 4}$ and an error tensor 
$\zeta \mH,$ $\wh{\mQ}_{4} = \mQ_{4} + \zeta \mH$. 
For the sake of discussion, suppose $\mH = \mI_d\ot\mI_d$.
Then the error term after tensor contraction with the filter $\bu^{\ot 2}$ 
becomes $\la\mH,\bu^{\ot 2}\ra = (\bu^\top\bu)\mI_d$. 
If $\bu$ is selected from a real Gaussian distribution, then the second moment of the spectral norm of this error matrix is $\E[\|\la\mH, \bu^{\ot 2}\ra\|_2^2] = \E[\|\bu\|_2^4] = O(d^2)$. 
On the other hand, if $\bu = \bp + \im \bq$ is a complex Gaussian vector 
where $\bp, \bq$ are i.i.d.\ standard real Gaussians, then the second order moment becomes $\E[\|\la\mH, \bu^{\ot 2}\ra\|_2^2] = O(d).$ 
}
This is the main intuition for why complex vectors are more helpful for keeping 
the estimation error small. The advantage of complex vectors leading to lower growth 
of the error also persists for higher moments; this is formally established 
for real and complex Gaussian vectors in \Cref{lem:moment-bound-query-tensor-after-projection} 
and \Cref{lem:real-complex-gaussian-comparison}, in \Cref{sec:PCA}.  
The choice of Gaussian vectors in our analysis is for technical convenience.
Full technical details for this polynomial sample and time implementation based 
on complex Gaussians are  included in \Cref{sec:poly-time-initialization}.

\subsubsection{Boosting Coarse Direction Recovery}

The $\init$ procedure described in \Cref{sec:tech-overview-CDR} gives a concrete approach for constructing initial, constant-correlation candidate vectors. However, several questions remain on how to make this procedure useful for constructing a hypothesis with arbitrarily small target error $\eps > 0.$ Most importantly, we need a way to reduce the error of candidate vectors and isolate useful candidate vectors from the rest of the produced list. We begin with providing the intuition for how to improve the error. 

\paragraph{Ideal (and nearly-ideal) update directions.} 
Fix one candidate vector $\bv \in \spd$ and one target $\w_i \in \spd.$ 
If we knew $\w_i^{\perp \bv},$ \new{the component of $\w_i$ orthogonal to $\bv$,} 
then we would immediately know $\w_i$, since 
$\w_i = \sqrt{1 - \norm{\w_i^{\perp \bv}}_2^2}\bv + \w_i^{\perp \bv},$ 
so one update of $\bv$ would suffice. 
Notice that $\w_i^{\perp \bv}$ corresponds to the ``shortest path'' direction 
for rotating $\bv$ towards $\w_i$: rotating $\bv$ in the direction of $\w_i^{\perp \bv}$ 
leads it to $\w_i,$ following the geodesic that connects $\bv$ and $\w_i.$ 
Of course, such an ideal update direction is not available to the algorithm. 
However, it is also unnecessary for contracting the distance between $\bv$ and $\w_i.$ 
If, instead, we had access to a unit vector $\br$ that correlates well with $\w_i^{\perp \bv}$, 
in the sense that $\br \cdot \frac{\w_i^{\perp \bv}}{\norm{\w_i^{\perp \bv}}_2} \geq 0.9,$ 
then $\br$ would be a ``nearly-ideal'' update direction: 
although not moving $\bv$ on the geodesic connecting it to $\w_i$, 
it would allow for a large rotation that contracts the angle between $\bv$ and $\w_i$ 
by a constant factor; this is the primary idea underlying efficient optimization on the Euclidean sphere \cite{wang2026efficient}. Of course, the question remains of how to construct such update directions.

\new{Observe that the ``constant correlation'' requirement $\br \cdot \frac{\w_i^{\perp \bv}}{\norm{\w_i^{\perp \bv}}_2} \geq 0.9$ is similar to what we had stated for $\init.$ In particular, if we had access to samples from a problem instance with $\wh{\w_i^{\perp \bv}} := \frac{\w_i^{\perp \bv}}{\norm{\w_i^{\perp \bv}}_2}$ as the hidden weight vectors, $\init$ would apply directly. We argue that this is indeed possible; the small loss in correlation (0.99 vs 0.9) is a small error slack needed for the technical argument. Methodologically, the use of moments beyond the first to construct update directions that enable contracting the  distance to target vectors is related to \cite{wang2025robustly}; in the present work, the use of higher moments in \eqref{eq:low-degree-tensor-form} is crucial to updating a set of $k \geq 1$ vectors \emph{simultaneously} (in \cite{wang2025robustly}, $k = 1$, and $p = 0$ in \eqref{eq:low-degree-tensor-form} suffices).}

To argue that $\init$ could likewise apply to estimating $\wh{\w_i^{\perp \bv}}$, in \Cref{subsubsec:boosting-the-directions},
we show that a simple rejection sampling-based procedure applied 
to labeled examples drawn from the input distribution leads to a new sampling distribution 
that itself is a Gaussian sum of ReLUs with hidden weight vectors 
$\widetilde\w_i= \overline\w_i/\|\overline\w_i\|_2,$ 
where $\overline\w_i = \w_i^{\perp\bv} +\gamma(\bv\cdot\w_i)\bv$ and $\gamma \in (0, 1)$ 
is a parameter of the rejection sampling procedure that we can choose. 
Then, choosing $\gamma$ sufficiently small 
\new{(relative to the relevant accuracy/heaviness threshold)}, 
the contribution of $\gamma(\bv\cdot\w_i)\bv$ in $\widetilde\w_i$ becomes negligible, 
and we can ensure that constant-alignment estimates for $\widetilde\w_i$ 
are also constant-alignment estimates for $\w_i^{\perp\bv}$, 
with just a slightly smaller alignment constant (this is where the loss from 0.99 to 0.9 occurs). Thus, to obtain vectors 
with constant-factor alignment with $\w_i^{\perp\bv}$s, it suffices to apply $\init$ to the samples 
from $\cD$ retained after the rejection sampling. 

Now, having access to a procedure that can produce near-ideal update directions, we discuss how the algorithm can maintain the invariant stated at the beginning of technical overview. The process of generating candidate updates and pruning the generated list is illustrated in \Cref{fig:one_step-intro}.

\paragraph{Constructing a large list of candidate updates.}

The process of constructing new candidate vectors is now as follows. 
In each round of computation $t,$ we maintain a list $\cV_t$, which contains vectors 
$\bv_i$ that approximate targets $\w_i.$ At the beginning of each round $t,$ 
for each $\bv \in \cV_t,$ we call the $\textsf{Refine}$ procedure, which performs the following operations. We generate a polynomially-sized list $\cL(\bv)$ 
that contains near-ideal update directions for all $\w_i$, $i\in [k]$ 
(and possibly many other vectors with no specific guarantee). 
For each vector in $\br \in \cL(\bv),$ we perform candidate projected gradient descent-style  
updates $\proj_{\spd}(\bu + \eta_1 \br),\, \proj_{\spd}(\bu - \eta_1 \br),$ 
and add both obtained vectors to $\cL(\bv),$ 
where $\eta_1$ is an appropriately chosen step size. 
For technical reasons (namely, to ensure sufficient contraction),
we repeat this procedure three times.
As a result, we obtain a large but still polynomially sized list $\cV_t^{\rm new}$ of vectors, 
with the following property: for each hidden weight vector $\w_i,$ $i \in [k],$ 
the list $\cV_t^{\rm new}$ contains at least one vector $\bv_i$ such that 
$\theta_\pm(\bv_i,\w_i) \leq c R_{t+1} = c \rho R_t,$ $\rho \in (0, 1),$ and $c \in (0, 1)$ 
is a constant that provides slack for introducing small errors when pruning the list. That is, there are new approximate vectors $\bv_i$ within this large list $\cV_t^{\rm new}$ 
that ensure constant contraction $c\rho$ of distance to the target vectors $\w_i.$  
However, as should be clear, at this point these ``useful'' vectors $\bv_i$ 
are like a few needles in a large haystack. The next step thus involves identifying 
the useful vectors and pruning the rest. 

\paragraph{Testing and pruning.} 
As a first step in removing irrelevant vectors from $\cV_t^{\rm new},$ in \Cref{subsubsec:tester}, 
we develop a testing procedure that can be used to remove all irrelevant vectors 
in $\cV_t^{\rm new};$ namely, all vectors with alignment to all $\w_i$ larger 
than a small constant $1+\tau$, $\tau > 0,$ times $c R_{t+1}.$ 
This testing procedure is based on using properties of Hermite coefficients of ReLUs, with similar intuition as in construction of $\init$ discussed in \Cref{sec:tech-overview-CDR}. 
In particular, it is possible to argue that for \new{even} integer $p$,  
we have the following identity:
\begin{equation}\label{eq:tester-population-moment-intro}
 a_p(\bv):=\frac1{\wh{\sigma}_p}\E[y\hep_p(\bv\cdot\x)]
     =\sum_i\beta_i\cos^p\theta_\pm(\bv,\w_i),
\end{equation}
where $\hep_p(\bv\cdot\x)$ is the normalized probabilist's Hermite polynomial 
of order $p$ (see \Cref{sec:prelims}). 
The first term $(1/{\wh{\sigma}_p})\E[y\hep_p(\bv\cdot\x)]$ tells us that $a_p$ 
can be estimated from samples; the second term $\sum_i\beta_i\cos^p\theta_\pm(\bv,\w_i)$ 
that for sufficiently large degree $p$
($p \approx \log(k/\eps)/(cR_{t+1})^2$), we can determine if $\bv$ correlates nontrivially 
with any target direction $\w_i$ (as measured by $\theta_\pm(\bv,\w_i)$) 
or if it can be eliminated. More specifically, \new{(with high probability)},
our tester returns $\mathsf{FALSE}$ 
for all $\bv \in \cL_t$ such that  $\theta_\pm(\bv,\w_i) > (1+\tau)c R_{t+1}$ 
and $\mathsf{TRUE}$ for all $\bv \in \cL_t$ such that $\theta_\pm(\bv,\w_i) \leq c R_{t+1}$. 
As a result, the tester can be used as a rough pruning method that 
removes all $\bv \in \cL_t$ that do not approximate any of the target directions within angle $(1+\tau)c R_{t+1}$. 

Ideally, the tester alone would suffice to move on to the next round of computation. 
However, it is possible that one $\w_i$ is approximated by many vectors $\bv_i$ in the list, 
and if all were retained in $\cV_{t+1}$, the size of $\cV_T$ would not be guaranteed 
to be polynomial, as the algorithm takes $T = \Theta(\log(k/\eps))$ rounds and each call to $\mathsf{Refine}$ (for each vector in the list) produces a polynomially-sized list. 
Thus, the final step, in \Cref{subsec:cover-refinement-pruning}, is to prune the list $\cV_t^{\rm new}$ 
 down to at most $k$ vectors (and assign it to $\cV_{t+1}$). 
The reason we can do this is geometric: since the previous steps ensured sufficient 
contraction with leeway for additional error, we make the following observation. 
If multiple vectors $\bv$ are close to the same target $\w_i$, 
then they must be close to each other. The final pruning step is thus the following: pick any 
candidate vector $\bv$ in the pruned list and add it to $\cV_{t+1}$. The testing procedure ensures 
that this vector is within the target distance of at least one hidden weight vector. Now, iterate 
through the remaining vectors $\bv$ in the pruned list $\cV_t^{\rm new}$ and add $\bv$ to $\cV_{t+1}$ 
if and only if it is at least $2(1+\new{2}\tau)c R_{t+1}$-far from all the vectors already in the 
list. We claim that for each $\w_i$, there is exactly one $\bv_i \in \cV_{t+1}$ such that 
$\theta_{\pm}(\bv_i, \w_i) \leq R_{t+1}$ (but the same $\bv_i$ can satisfy 
$\theta_{\pm}(\bv_i, \w_i) \leq R_{t+1}$ for multiple $\w_i$), thus $|\cV_{t+1}| \leq k.$  Fix one 
$\w_i$ and consider all $\bv$ such that $\theta_{\pm}(\bv, \w_i) \leq (1+\tau)c R_{t+1};$ these are 
the candidates $\bv$ that survive the testing prune, \new{and among them, there is at least one 
vector $\bv_i^*$ such that $\theta_{\pm}(\bv_i^*,\w_i)\leq c R_{t+1}$, as guaranteed by the 
refinement procedure}. Because, by the triangle inequality, all such $\bv$ are within 
$2(1+\tau)c R_{t+1}$ of each other, at most one can get added to $\cV_{t+1}.$ 
If exactly \new{$\bv_i^*$} is added, we are done, so suppose \new{$\bv_i^*$}
is \new{not} added to $\cV_{t+1}.$ This can only happen if 
there exists $\bv' \in \cV_{t+1}$ such that 
$\theta_{\pm}(\new{\bv_i^*}, \bv') \leq 2(1+\new{2}\tau)c R_{t+1}.$ 
Thus, again by the triangle inequality, 
$\theta_{\pm}(\w_i, \bv') \leq \new{(2(1 + \new{2}\tau) + 1)}c R_{t+1}$. 
Parameters $c, \tau$ are chosen so that $(2(1+\new{2}\tau)+1) c \leq 1,$ 
and thus we have that each $\w_i$ can be mapped 
to one $\bv_i \in \cV_{t+1}$ such that 
$\theta_{\pm}(\bv_i, \w_i) \leq R_{t+1}.$

\begin{figure}[H]
    \centering
    \begin{subfigure}{\textwidth}
        \centering
        \includegraphics[width=0.7\linewidth]{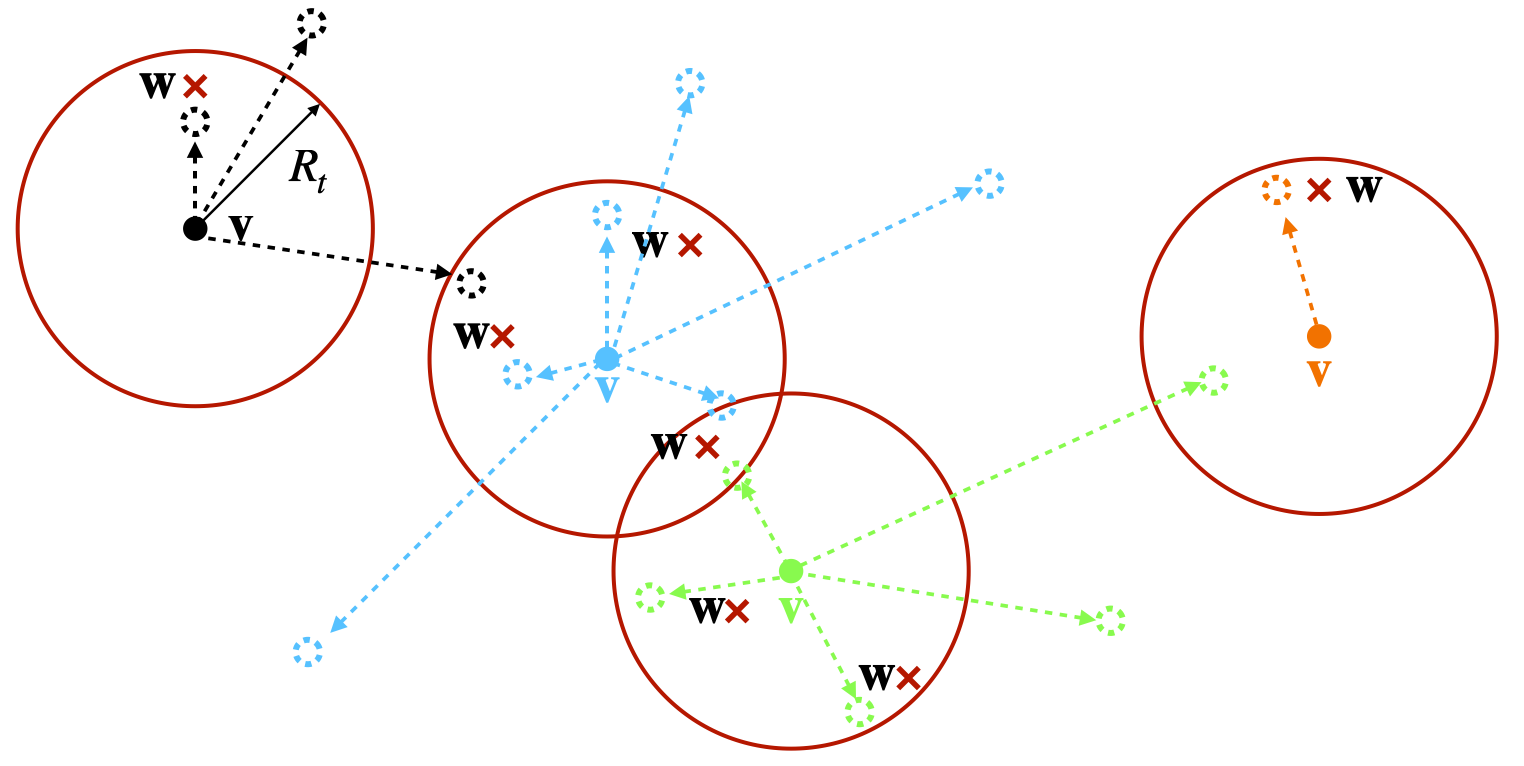}
        \caption{After the $\mathsf{Refine}$ procedure, each candidate produces a list of vectors (marked in a dotted circle), of which some are closer to the hidden weight vectors. But there are also potentially many other redundant vectors.}
\end{subfigure}

    \begin{subfigure}{\textwidth}
        \centering
        \includegraphics[width=0.7\linewidth]{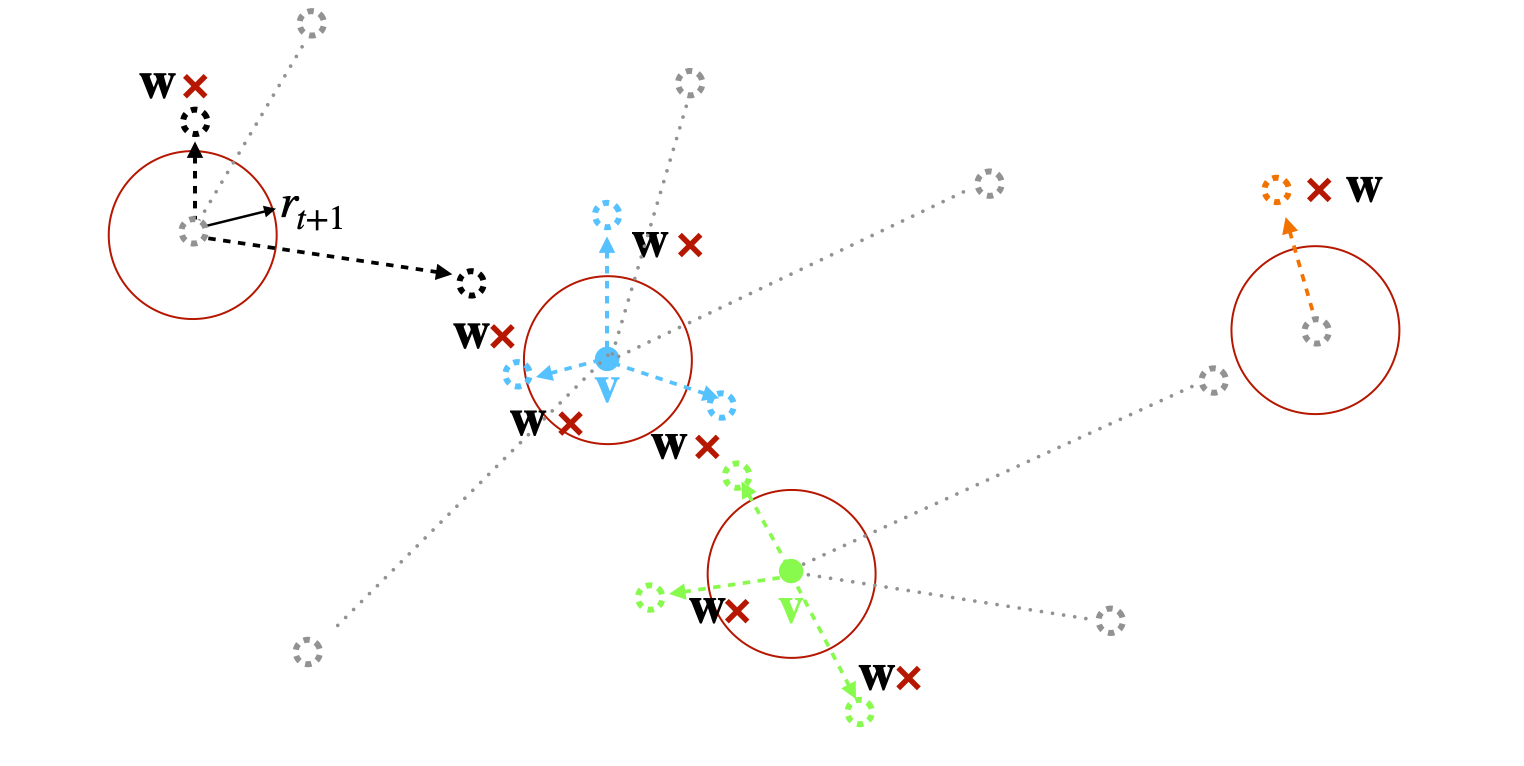}
        \caption{After the $\sat$ procedure, those vectors that do not ``cover'' hidden weight vectors at a target contracted distance are discarded (marked in gray).}
\end{subfigure}

    \begin{subfigure}{\textwidth}
        \centering
        \includegraphics[width=0.7\linewidth]{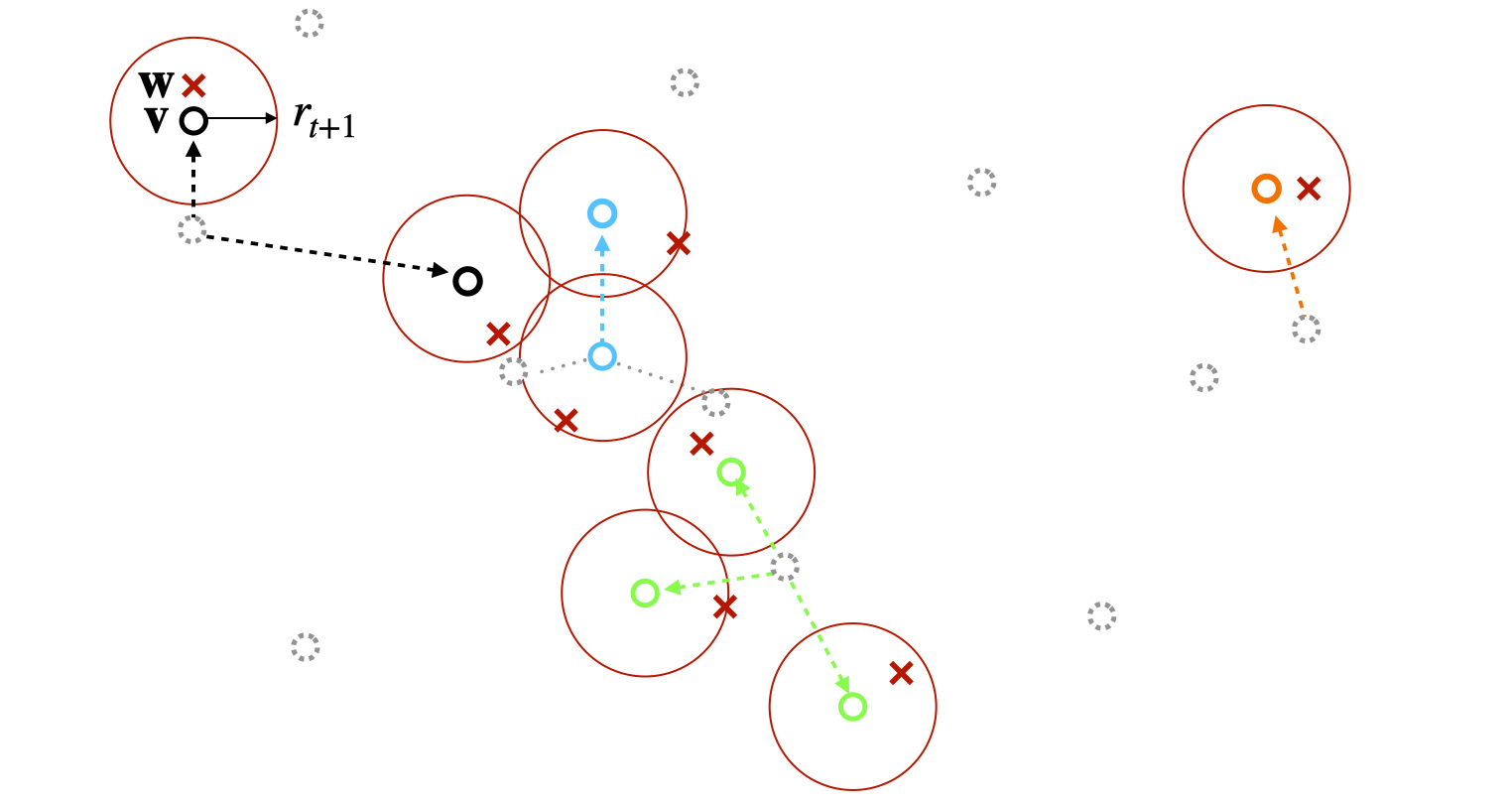}
        \caption{There might be many candidate vectors that are close to each other, covering the same hidden weight vector. After the $\mathsf{Prune}$ procedure, the candidate vectors are separated (the pruned vector is also marked in gray). Yet, we guarantee that each hidden weight vector is still covered by the candidate vectors after pruning.}
\end{subfigure}

    \caption{The update procedure of the \nameref{alg:cover-and-prune} algorithm. $\times$ denotes the hidden weight, $\bullet$ denotes the candidate vector at step $t$, $\circ$ denotes the vectors generated after step $t$. Only the hidden weight vectors with not-too-small coefficients are displayed.}
    \label{fig:one_step-intro}
\end{figure}

\subsection{Related Work}\label{ssec:related}

\new{Our work belongs to the literature on learning structured multi-index models, and we review related work within this context. }

\paragraph{Learning positive linear combinations of ReLUs.}
The most directly related works are
\cite{DKKZ20, diakonikolas2020small, diakonikolas2025implicit},
whose complexity guarantees were discussed earlier in the introduction.
The algorithm of~\cite{DKKZ20} uses the degree-$2$ Chow parameters
of the target to find a low-dimensional subspace sufficient for
prediction, and then searches over candidate networks in that subspace.
The work of~\cite{diakonikolas2020small} uses higher-degree moments
to find a vector space of polynomials that nearly vanish on the
hidden weight vectors.
Its structural result gives a small cover for the near-zero set
of these polynomials, which is then used to construct a hypothesis.
The work of~\cite{diakonikolas2025implicit} estimates high-order
moment tensors implicitly, without writing down their entries,
and uses these estimates to approximate the target's Hermite expansion.
This generalizes the implicit moment-estimation technique of
Li and Liu~\cite{LL22}, developed for clustering separated mixtures
of spherical Gaussians.

\paragraph{Learning general linear combinations of ReLUs.}
PAC learning general linear combinations of ReLUs, with both positive
and negative coefficients, has also been studied.
Without structural assumptions on the hidden weight vectors,
\cite{ChenDGKM23} gave the first algorithm with running time
polynomial in $d$ and $1/\eps$ for every fixed $k$.
Subsequent work improved the running time to
$(dk/\eps)^{O(k^2)}$~\cite{ChenN24} and
$(dk/\eps)^{O(k)}$~\cite{DK23-schur}.
These algorithms handle a broader class of functions than the one studied here,
but their complexity scales exponentially in the number of hidden units.\footnote{
Here we focus on learning from random labeled examples.  
If query access to the function is allowed, \cite{ChenKM21} gives a learner 
for one-hidden-layer ReLU networks with general coefficients and biases 
whose query complexity and running time are polynomial in dimension, the number of hidden
units, inverse accuracy, and bounds on the weights and biases.}

\paragraph{Learning multi-index models and deeper networks.}
The work of~\cite{CKM22} gives fixed-parameter learning algorithms
for homogeneous ReLU networks of arbitrary depth under Gaussian
inputs, with polynomial dependence on the ambient dimension but
exponential dependence on network size and inverse accuracy.
Recent work on multi-index models gives learning guarantees for
more general classes of functions that depend only on an unknown
low-dimensional subspace~\cite{DIKZ25, DIKR25, DLB25}. 
Similarly, the latter algorithms have polynomial dependence on the ambient dimension 
but exponential dependence on the other relevant parameters.

\paragraph{Parameter learning under nondegeneracy assumptions.}
A number of earlier works focus on parameter learning, namely recovering
the parameters of the data-generating network.
Examples include tensor-based methods~\cite{Janz15} and algorithms
combining initialization with local optimization~\cite{ZhongS0BD17}.
Related tensor methods were developed for mixtures of generalized
linear models in~\cite{SedghiJA16}.
For positive linear combinations of ReLUs under Gaussian inputs,
\cite{GeLM18} gives a parameter-learning algorithm assuming a
full-rank hidden-weight matrix, with complexity depending polynomially
on its condition number.
\cite{BakshiJW19, GeKLW19} give parameter-recovery
algorithms for vector-valued two-layer networks under rank and
related nondegeneracy assumptions.
The work of~\cite{ATV21} develops tensor-based algorithms for
scalar-output ReLU networks with biases and general coefficients.
In contrast, we focus on small prediction error rather than parameter
recovery, with complexity independent of separation, rank, 
condition-number\new{, or any other parameters requiring strong assumptions about the configuration of the hidden weight vectors}.

\paragraph{Computational lower bounds.}
For general linear combinations of ReLUs,
\cite{DKKZ20, GoelGJKK20} establish super-polynomial lower bounds
in the correlational statistical-query (CSQ) model~\cite{BshoutyFeldman:02,Feldman08}.
This is a restriction of the Statistical Query (SQ) model~\cite{Kearns93}, 
in which the learner accesses the target through
estimates of correlations with query functions; see
also~\cite{BFJ+:94, VempalaW19}.
For networks with two hidden layers, \cite{ChenGKM22} establishes
general SQ and cryptographic hardness
even in the noiseless setting.
More recently, \cite{LZZ24} proves cryptographic hardness for
one-hidden-layer ReLU networks with both positive and negative
coefficients and biases, in the presence of small independent
Gaussian label noise.
Whether general linear combinations of ReLUs without biases can be
learned in fully-polynomial time under Gaussian inputs remains open.
Since our algorithm can be implemented in the CSQ model, 
the CSQ lower bounds above rule out an immediate 
extension to general coefficients.

\subsection{Concurrent Work and AI Use}
All ideas and technical arguments in this paper are the intellectual ownership of the authors, who have worked on developing them over the past 2+ years. In the process of writing the paper, OpenAI GPT models were used to improve the presentation (including generating Figures 1 and 2), audit the proofs, and simplify some of the statements. We note that 
a very recent paper \cite{zheng2026polynomial} appeared on the arXiv, 
claiming essentially the same main result. 
Rather curiously, that work is quoting some of the exact same core techniques as our work, 
but is lacking in most of the proof details. The paper \cite{zheng2026polynomial} further 
uses some of the same non-standard terminology appearing in our work 
(such as ``refine,'' ``prune,'' ``cover,'' etc.) and even some of the same universal 
constants chosen lazily by the authors to be smaller than needed. 
As a result, we were forced to rush posting of this work, before we were entirely 
satisfied with its presentation, but with complete technical details included. 

\subsection{Paper Organization}
The paper is structured as follows:
In \Cref{sec:prelims}, we record the notation and background required in our technical sections.
\Cref{sec:main-alg,sec:opt,sec:PCA} provide a detailed description and analysis of our algorithm.

\section{Notation and Preliminaries}\label{sec:prelims}

\paragraph{Basic Notation.} 
For $n \in \Z_+$, we denote $[n] \eqdef \{1, \ldots, n\}$. We use bold lowercase letters for vectors
and bold uppercase letters for matrices and tensors. For $\bx \in \R^d$ and $i \in [d]$, $(\x)_i$ denotes the
$i^{\text{th}}$ coordinate of $\bx$, and $\|\bx\|_2 \eqdef (\littlesum_{i=1}^d (\x)_i^2)^{1/2}$ denotes the
$\ell_2$ norm of $\bx$.
We write $\bx \cdot \by$ for the standard inner product of $\bx,\by\in\R^d$ 
and $\theta(\bx,\by):= \arccos(\bx \cdot \by/(\norm{\bx}_2\norm{\by}_2))$ 
when $\bx \neq \vzero$ and $\by \neq \vzero$.
We denote $\theta_\pm(\bx,\by)\eqdef \min\{\theta(\bx,\by), \pi - \theta(\bx,\by)\}$. 

For unit vectors $\bu,\bv$, we use $\bu^{\perp \bv}$ to denote the component of $\bu$ 
orthogonal to $\bv$, i.e., 
$\bu^{\perp \bv} = (\vec I - \bv\bv^\top)\bu$. For any vector $\bu\neq \vec 0$, 
we use $\wh{\bu}$ to denote its normalization to unit $\ell_2$ norm. 
Given $\bv\in\R^d$, we use $\bv^{\perp}$ to denote the linear subspace of vectors 
orthogonal to $\bv$, i.e.,  $\bv^{\perp} = \{\bu\in\R^d:\bu\perp\bv\}$. 
We denote the projection from $\R^d$ to $\bv^{\perp}$ by $\mP_{\perp\bv} = \mI - \bv\bv^\top$.
More generally, for a linear subspace $\sF$ of $\R^d$, 
the projection from $\R^d$ to $\sF$ is denoted by $\mP_{\sF}$. 

Throughout the paper, $\Sp^{d-1}$ denotes the unit sphere in $\R^d$. 
We use $\1\{\cE\}$ to denote the indicator of an event $\cE$.
$\sigma(z) \eqdef \max\{0,z\}$.  
\new{
We write $E\gtrsim F$
for two nonnegative expressions $E$ and $F$ to denote that
there exists some universal constant $c>0$ (independent
of the variables or parameters on which $E$ and $F$ depend)
such that $E \geq c F$.
The expression $E\lesssim F$ is similarly defined.
We use the notation $g(t) \leq \poly(t)$ for a quantity $t \geq 1$
to indicate that there exist constants $c,C>0$ such that
$g(t) \leq Ct^c$.
Similarly, we use $g(t) \geq \poly(t)$ for a quantity $t<1$
to denote that there exist constants $c,C>0$ such that
$g(t) \geq Ct^c$.
} Finally, we use standard asymptotic notation $O(\cdot)$, $\Theta(\cdot)$, $\Omega(\cdot)$, 
and $\widetilde O(\cdot)$, $\widetilde \Omega(\cdot)$ 
(the latter two suppress polylogarithmic  factors in the argument).

\paragraph{Complex Space.} 
We use $\mathbb{C}^d$ to denote the $d$-dimensional complex vector space. 
For $\bz\in\mathbb{C}^d$, $\bz = \x + \im\by$ where $\im$ is the imaginary unit, 
and $\bar{\bz}= \x - \im\by$ is the complex conjugate of $\bz$. We use $\Real(\bz)$ to denote the real part of $\bz$ and $\mathrm{Im}(\bz)$ to denote its imaginary part.
For a complex number $z = a+\im b$, recall that $|z|^2 = a^2 + b^2 = z\bar{z}$.

\paragraph{Matrices and Tensors.}
For a matrix $\mM\in\R^{d\times d}$, we use $\|\mM\|_2$ to denote its 
spectral/operator norm. For a symmetric matrix $\mM$, $\lambda_i(\mM)$ 
denotes its $i^{\text{th}}$ largest eigenvalue and $\br_i(\mM)$ denotes the corresponding 
(unit) eigenvector, $i\in[d]$. \new{$\mathrm{tr}(\mM)$} denotes the trace of a matrix. Our analysis makes use of the following standard fact.
\begin{fact}[Wedin's Theorem]\label{fct:wedin}
    For any pair of symmetric matrices $\mM_1,\,\mM_2 \in \R^{d\times d}$, if $\lambda_1(\mM_1) - \lambda_2(\mM_1) - \|\mM_1 - \mM_2\|_2> 0$, then
    \begin{align*}
        \sin(\theta(\br_1(\mM_1), \br_1(\mM_2)))\leq \frac{\|\mM_1 - \mM_2\|_2}{\lambda_1(\mM_1) - \lambda_2(\mM_1) - \|\mM_1 - \mM_2\|_2}.
    \end{align*}
\end{fact}

A complex tensor $\mW\in(\mathbb{C}^{d})^{\ot p}$ of dimension $d$ and order $p$ 
is a multilinear map defined by a \new{$p$-index} array with complex entries $(\mW)_\alpha$,
where $\alpha = (i_1,\ldots,i_p)$ with $i_1,\dots,i_p \in [d]$.
A tensor $\mW$ is symmetric if $(\mW)_\alpha = (\mW)_{\alpha_\pi}$,
where $\alpha_\pi = (i_{\pi(1)},\ldots,i_{\pi(p)})$,
for any permutation $\pi:[p]\to[p]$.
For complex tensors $\mW,\mT$ of dimension $d$ and order $p$,
we denote by $\langle \mW,\mT\rangle$ their entrywise inner product:
\begin{align*}
    \la \mW,\mT\ra = \sum_{i_1,\dots,i_p} (\overline{\mW})_{i_1,\dots,i_p}(\mT)_{i_1,\dots,i_p} \;.
\end{align*}
For a complex tensor $\mW$, let $\|\mW\|_F^2 = \la \overline{\mW},\mW\ra = \sum_{i_1,\dots,i_p}|\mW_{i_1,\dots,i_p}|^2$
denote its Frobenius norm.
\new{Let $\sU,\sV$ be finite-dimensional complex inner-product
spaces. The Hermitian adjoint of a complex-linear map
$\mW:\sU\to\sV$ is the unique complex-linear map
$\mW^*:\sV\to\sU$ satisfying $\la\mW(\mA),\mB\ra  = \la\mA,\mW^*(\mB)\ra$ for any $\mA\in\sU$ and $\mB\in\sV$.}

We overload the notation $\la\cdot,\cdot\ra$ with partial tensor contraction: let $\mW\in(\mathbb{C}^{d})^{\ot(p+m)}$, $\mT\in(\mathbb{C}^{d})^{\ot p}$, $\la\mW,\mT\ra$ is a tensor in $(\mathbb{C}^{d})^{\ot m}$ defined by
\begin{align*}
    (\la\overline{\mW},\mT\ra)_{j_1,\dots,j_m} = \sum_{i_1,\dots,i_p} (\overline{\mW})_{i_1,\dots,i_p,j_1,\dots,j_m}(\mT)_{i_1,\dots,i_p}.
\end{align*}
Let $\mathrm{Sym}(\mW)$ be the symmetrization operator: \begin{align*}
    (\mathrm{Sym}(\mW))_{i_1,\dots,i_p} = \frac{1}{p!} \sum_{\text{permutation $\pi$ of $[p]$}}(\mW)_{i_{\pi(1)},\dots,i_{\pi(p)}}.
\end{align*}

\new{
Identify an order-$p$ tensor $\mT$ with the vector
$\bt=\operatorname{vec}(\mT)$ of its entries.
Its covariance is the Hermitian positive semidefinite matrix
    $\mathbf{Cov}(\mT)
    =
    \E[
        (\bt-\E[\bt])(\bt-\E[\bt])^*],$
and its uncentered second moment is $\E[\bt\bt^*]$.
We say that its covariance is bounded above by $V$ if $\mathbf{Cov}(\mT)\preceq V\mI$, which is equivalent to $\sup_{\|\mW\|_F = 1}\E[|\la\mW,\mT - \E[\mT]\ra|^2]\leq V$.
For real-valued tensors, these definitions reduce to the usual
real covariance and second-moment matrices after flattening.
}

\paragraph{Gaussian Space and Hermite Analysis.} 
Let $\normal(\boldsymbol\mu,\vec \Sigma)$ denote the $d$-dimensional Gaussian distribution with mean
$\boldsymbol\mu\in \R^d$ and covariance $\vec\Sigma\in\R^{d\times d}$.
We will primarily use the standard normal distribution $\normal(\vec 0,\mI_d)$, denoted simply by $\normal$.
The standard Gaussian distribution on a linear subspace $\sF$ of $\R^d$ is denoted by $\tx\sim\calN(\vec 0,\mI_{\sF})$ and  defined via $\tx = \mP_\sF\x$, $\x\sim\calN(\vec 0, \mI_d)$.
The Gaussian $L_2$ inner product is $\inner{f}{g} \eqdef \E_{\x\sim\normal}[f(\x)g(\x)]$.
For $f\in L_2(\normal)$, define $\|f\|_{L_2}\eqdef (\E_{\x\sim\normal}[f(\x)^2])^{1/2}$.

\begin{fact}[Gaussian Hypercontractivity; see, e.g.,~\cite{ODonnell:BFA}]\label{fct:gaussian-hypercontractivity}
    For any real polynomial $f(z)$ of degree at most $p$ and \new{real $q\geq 2$}, it holds that $(\E_{z\sim\calN}[|f(z)|^q])^{1/q}\leq (q-1)^{p/2} (\E_{z\sim\calN}[f(z)^2])^{1/2}$.
\end{fact}

We now proceed to define Hermite polynomials and the associated orthonormal basis they define.
In one dimension, define the normalized probabilists' Hermite polynomial
\begin{equation} \label{eqn:hermite}
\hep_p(z)\;\eqdef\;\frac{1}{\sqrt{p!}}(-1)^p \exp(z^2/2)\,\frac{\diff^p}{\diff z^p}\exp(-z^2/2).
\end{equation}
The family $\{\hep_p\}_{p\ge 0}$ is an orthonormal basis for $L_2(\normal(0,1))$.
Thus, any $f:\R\to\R$ with $f\in L_2(\normal(0,1))$ has a unique Hermite expansion
\[
f(z)= \sum_{\new{p} \ge 0} \widehat f_p\,\hep_p(z),\qquad
\widehat f_p\eqdef \E_{z\sim\normal(0,1)}[f(z)\hep_p(z)].
\]

\new{We make use of a number of basic properties of Hermite polynomials summarized below.}

\begin{fact}[Basic Hermite Facts, see, e.g.,~\cite{Bog:98,ODonnell:BFA}]\label{fact:hermite-poly}
For normalized Hermite polynomials $\hep_p$, as defined in (\ref{eqn:hermite}), 
the following hold: 
\begin{enumerate}[leftmargin=*]
\item (Differentiation) $(\hep_p(z))' = \sqrt{p} \, \hep_{p-1}(z)$, for $p \geq 1$.
\item (Degree Identity) $z \hep_p(z)=\sqrt{p+1} \, \hep_{p+1}(z)+\sqrt{p} \, \hep_{p-1}(z)$, 
for $p\geq 2$.
\item (Smoothing) Let $f\in L_2(\calN(0,1))$, $\rho\in[0,1]$ and denote 
$\mathrm{T}_\rho(f)(x) \eqdef \E_{z\sim\calN(0,1)}[f(\rho x + \sqrt{1 - \rho^2}z)]$. Then the smoothed function $\mathrm{T}_\rho(f)(x)$ has Hermite expansion:
\begin{align*}
    \mathrm{T}_\rho(f)(x) = \sum_{p\geq 0} \rho^p \wh{f}_p \hep_p(x).
\end{align*}
\item (Hermite expansion in one specific \new{direction}) Let $f\in L_2(\calN)$
and 
$\bu,\bv$ be unit vectors in $\R^d$. Then:
\begin{align*}
    \E_{\x\sim\calN(\vec 0,\mI_d)}[f(\bv\cdot\x)\hep_p(\bu\cdot\x)] = \wh{f}_p \cos^p(\theta(\bv,\bu)).
\end{align*}
\end{enumerate}
\end{fact}

Throughout the paper, we denote by $\wh{\sigma}_i$ 
the $i^{\text{th}}$ Hermite coefficient of the ReLU function $\sigma$. 

\smallskip 

These Hermite coefficients admit the following characterization.

\begin{fact}[Hermite expansion of ReLU; see, e.g.,~\cite{GGK20, DK23-schur}]\label{fct:hermite-relu}
For $\sigma(z) \eqdef \max\{0, z\},$ 
we have $\sigma(z)=\sum_{p=0}^{\infty} \wh{\sigma}_p\hep_p(z)$, 
where $\wh{\sigma}_0 = 1/\sqrt{2\pi}$, $\wh{\sigma}_1 = 1/2$, and
\[
\wh{\sigma}_{2p-1}=0\quad \textrm{for all }  p\ge 2,\qquad
\wh{\sigma}_{2p}=\frac{\hep_{2p-2}(0)}{\sqrt{2\pi(2p)(2p-1)}}\quad \textrm{for all } p \ge 1.
\]
In particular, $|\wh{\sigma}_{2p}|=\Theta(p^{-5/4})$ and $\wh{\sigma}_{2p}^2=\Theta(p^{-5/2})$ for $p\ge 1$.
Moreover, if we formally differentiate the Hermite series, we obtain
$\sigma'(z)=\sum_{p\ge 0} \wh{\sigma}'_{p}\hep_p(z)$ with $\wh{\sigma}'_p=\wh{\sigma}_{p+1}\sqrt{p+1}$.
\end{fact}

Hermite polynomials can be extended to Hermite tensors in $\R^d$. 
Throughout this paper, we will use the following normalized Hermite tensor.

\begin{definition}[Normalized Probabilist's Hermite Tensor] \label{def:herm-tensor}
We define the normalized order-$p$ Hermite tensor in $d$ dimensions as follows: for any $i_1,\dots,i_p\in[d]$,
\begin{equation} \label{eqn:tensor-herm}
    (\he_p(\x))_{i_1,\dots,i_p} \eqdef \frac{1}{\sqrt{p!}}\sum_{\substack{\text{Partitions $P$ of $[p]$} \\ \text{into sets of size 1 and 2}}}\prod_{\substack{\text{all size-2 sets $\{a,b\}$} \\ \text{ in partition $P$}}} -\new{(\mI_d)_{i_a,i_b} }\prod_{\substack{\text{all size-1 sets $\{c\}$} \\ \text{ in partition $P$}}}(\x)_{i_c}\; . \end{equation}
\end{definition}

We gather here important facts about Hermite tensors that are used in our analysis.

\begin{fact}[\new{Properties of Normalized} Hermite Tensor]\label{fct:hermite-tensor}
    Let $\he_p$ be the $p^{\text{th}}$-order \new{normalized} Hermite tensor. 
    The following holds: \begin{enumerate}[leftmargin=*]
\item $\he_p(\x) = \E_{\bz\sim\calN(\vec 0,\mI_d)}[\Real((\x+\im\bz)^{\otimes p})/\sqrt{p!}]$, for all $\x\in\R^d$.
        \item \new{Let $f\in L_2(\calN(0,1))$. For any $\bv \in \spd$, $\E_{\bz\sim\calN(\vec 0,\mI_d)}[f(\bv\cdot\bz)\he_2(\bz)] =\E_{\bz\sim\calN(\vec 0,\mI_d)}[f''(\bv\cdot\bz)]\bv\bv^\top$.}
    \end{enumerate}
\end{fact}

\section{Optimization via Coarse Direction Recovery, Refinement, and Pruning} \label{sec:opt}

In this section, we provide the full technical details of our analysis, assuming access to a coarse 
direction recovery ($\init$) procedure as an oracle; the implementation of $\init$ is presented 
in \Cref{sec:PCA}. Our oracle access to $\init$ is formally stated as the following assumption.

\begin{assumption}[Coarse Direction Recovery]
\label{assum:initialization-algorithm}
Let $d,k\in \Z_+$  and $\lambda,\delta\in(0,1)$.
\new{Given access to a Gaussian sum of $k$ ReLUs distribution on 
$\R^d \times \R$}, 
a Coarse Direction Recovery algorithm $\init(\D,\lambda,\delta)$
draws $N_{\init}(d,k,\lambda,\delta)$ i.i.d.\ samples from $\D$, runs in time $T_{\init}(d,k,\lambda,\delta)$,
and with probability at least $1-\delta$ returns a list
$\cL\subseteq\Sp^{d-1}$ of size at most $L_{\init}(d,k,\lambda,\delta)$ such that: 
for every $\lambda$-heavy component $i \in [k]$, there exists $\br_i \in\cL$ 
satisfying $|\br_i \cdot\vec w_i|\geq 0.99$.
\end{assumption}

At a high level, the $\init$ procedure postulated 
in \Cref{assum:initialization-algorithm} guarantees 
that if we had i.i.d.\ samples from a \new{Gaussian sum of $k$ ReLUs} distribution $\D$, 
then we could efficiently compute 
a not-too-large list $\mathcal L$ of ``candidate'' unit vectors 
such that for each $\lambda$-heavy hidden weight vector $\w_i$ there exists a candidate $\br \in \mathcal{L}$ 
that is ``close'' to $\w_i$ (as measured by the nontrivial correlation 
$|\br \cdot \w_i| \geq 0.99$). The list may also contain  other vectors with no specific guarantee.  
Given the existence of such an algorithm \new{that runs in polynomial time} 
(which we establish in \Cref{sec:PCA}), 
our main result in this section is a polynomial-time learner of an unknown sum of ReLUs, 
as stated in the following theorem.

\begin{theorem}[Direction Covering Algorithm]\label{prop:optimization}
Let $\eps\in(0,1)$, $k, d \in\mathbb \Z_+$, 
and $\D$ be a \new{Gaussian sum of $k$ ReLUs} distribution on $\R^d \times \R$. 
Given access to a \new{procedure} $\init$ 
satisfying \Cref{assum:initialization-algorithm}, 
there exists an algorithm (\nameref{alg:main}) that 
repeatedly calls $\init$ with parameters $\lambda = \poly(\eps/k)$ 
and $\delta=\poly(\eps/k)$ and attains the following performance guarantee: 
it draws $N = \widetilde{O}(d\,\new{\poly(k/\eps)} \, + \,L_\init^3(N_\init + L_\init^3)\poly(k/\eps))$
samples from $\D$, runs in \new{$\poly(d, N, T_\init)$} time and, with probability at least $2/3$,
outputs a hypothesis 
$\widehat y(\x)$ such that 
$ \E_{(\x,y)\sim\D}[(y-\widehat y(\x))^2] \leq\eps.$
\end{theorem}

The pseudocode for our main algorithm is displayed below.

\begin{algorithm}[h]
\caption{$\mathsf{LearnSumReLUs}(d,k,\eps;\D)$}\label{alg:main}
    \begin{algorithmic}[1]
        \State \textbf{Input: } $k,\eps>0$, and access to \new{a Gaussian sum of $k$ ReLUs distribution} $\D$.
        \State \textbf{Output:} a hypothesis $\wh{y}$ such that $\Exy[(\wh{y}(\x)-y)^2]\leq \eps$ with probability at least $2/3$.
        \medskip
\State  $\eps_0  \gets \sqrt{\eps}/(4k)$ , $\delta\gets 1/100$ 
\Comment{$\eps_0$: error threshold of Parameters; $\delta$: failure probability} 
        \State $\wh{\cV} \gets$ \nameref{alg:cover-and-prune} \Comment{Returns a candidate list by repeatedly calling the $\init$ Subroutine. }\State Let $\widehat\cV=\{\bv_1,\ldots,\bv_m\}$ where $m=|\widehat\cV|$. \Comment{Enumerate the elements}
        \State Draw $\Theta(dk^2/\eps^2)$ samples from $\D$ to construct an estimator $\widehat{\bg}$ of $ \Exy[y\x]$.
        \State Draw $\Theta(k^{\new{10}}/\eps^2)$ samples from $\D$ to formulate the empirical distribution $\widehat{\D}$.
        \State $\wh{\bm{\alpha}}\gets\argmin_{\bm{\alpha}\in\{\alpha_i\geq0: \sum \alpha_i\leq k\}}\E_{(\x,y)\sim \widehat{\D}}[(y - \widehat{\bg}\cdot\x - \sum_{i=1}^m \frac{\alpha_i}{2}|\bv_i\cdot\x|)^2]$.  
        \Statex \Comment{Convex Quadratic Program}
\State\textbf{Return: }$\wh{y}(\x) = \widehat{\bg}\cdot\x +\sum_{i\in[m]} (\wh{\alpha}_i/2)|\bv_i\cdot\x|$.
    \end{algorithmic}
\end{algorithm}

The body of this section is devoted to the proof of \Cref{prop:optimization}.
In \Cref{sec:direction-covering-algorithm}, we show that the main subroutine 
of our algorithm, \nameref{alg:cover-and-prune}, successfully covers 
all hidden weight vectors $\w_i$ with non-trivial coefficients.
Then in \Cref{subsec:proof-for-optimization}, we present the proof of \Cref{prop:optimization}.

\subsection{Learning $\lambda$-Heavy Hidden Weight Vectors}\label{sec:direction-covering-algorithm}

Learning $\lambda$-heavy (i.e., nontrivial-coefficient) hidden weight vector described in this 
subsection is the technical core of the paper. We begin this section by first arguing that the $\init$ 
procedure can be used to construct updates to candidate vectors maintained as estimates of the hidden 
vectors $\w_i.$ In particular, in \Cref{subsubsec:boosting-the-directions}, we argue that, given a 
fixed vector $\bv \in \spd$ and a simple to implement rejection sampling procedure (used for selecting 
samples from $\D$), $\init$ procedure can be used to construct a polynomially-sized list $\cL(\bv)$ 
that contains constant-correlation estimates of $\w_i^{\perp \bv}$ for $\lambda$-heavy ReLU 
components. On input $\bv,$ our \nameref{alg:refine} procedure uses $\init$ to construct a list $\cV$ 
of gradient descent-style updates to $\bv$ such that for each $\lambda$-heavy $\w_i$, there exists at 
least one updated vector that contracts the distance between $\bv$ and $\w_i$ by a small constant 
factor. At the end of \nameref{alg:refine}, only a small subset of updated candidate vectors in the 
output list $\cV$ is useful to algorithm progress, that is, to improving the estimation of hidden 
$\lambda$-heavy vectors. The next step in the algorithm is thus to prune the list, after 
\nameref{alg:refine} is applied to all candidate vectors maintained by the algorithm. As a first 
coarse pruning step, \Cref{subsubsec:tester} introduces a testing procedure that can be used to 
identify elements of $\cV$ farther from all $\lambda$-heavy $\w_i$ than a prespecified threshold. As 
a final step, \Cref{subsec:cover-refinement-pruning} introduces \nameref{alg:global-prune} procedure 
that prunes the list further so that only up to $k$ candidate vectors are retained by the algorithm. A 
complete algorithm statement for learning $\lambda$-heavy hidden vectors is then provided in 
\nameref{alg:cover-and-prune} and the main result of the subsection (\Cref{prop:cover-and-prune}) is 
proved.

\subsubsection{Boosting Hidden Vector Estimates}\label{subsubsec:boosting-the-directions}

We begin by arguing that, given any $\bv \in \spd,$ the $\init$ procedure can be called to obtain 
constant-correlation estimates of $\w_i^{\perp \bv}$, for any $\lambda$-heavy component $i \in [k].$ 
These estimates are then used to boost the estimation of $\w_i$, by treating estimates of 
$\w_i^{\perp}$ as gradient descent directions from $\bv.$ Of course, it is not immediately clear that 
$\init$ can be applied here, as \Cref{assum:initialization-algorithm}   requires samples from a 
\new{Gaussian sum of $k$ ReLUs}
distribution with 
$\widehat{\w_i^{\perp\bv}}:= \w_i^{\perp \bv}/\|\w_i^{\perp \bv}\|_2$ as the hidden directions. 

\new{For intuition on how $\init$ could be applied to estimating $\w_i^{\perp \bv}$ and boosting the estimate of $\w_i$, recall first that a sample $(\x, y)$ drawn from our distribution $\D$ satisfies 
$\x\sim\calN(\vec 0,\mI_d)$ and $y= \sum_{i=1}^k \new{\beta}_i\sigma(\vec w_i\cdot\x)$.}
Fix a unit candidate $\bv$ and a target $\w_1$. 
Given $\vec v$, we can express $\vec w_1$ as  
$\vec w_1=\cos(\theta_1) \, \vec v + \sin(\theta_1) \, \widehat{\vec w_1^{\perp \vec v}}$, 
where $\theta_1 = \theta(\w_1,\bv)$.
Consider the samples $(\vx, y) \sim \D$ that lie in a band of width $\gamma$ around the origin, 
namely satisfy $|\vec v\cdot\x|\leq \gamma$. 
Then, it is possible to argue that $\sigma(\vec w_1\cdot\x)\approx \sigma( \sin(\theta_1) \widehat{\vec w_1^{\perp \vec v}}) + O(\gamma)$. 
Thus, \new{when $\gamma \to 0$,}  for such samples we have that 
$y \approx \sum_{i=1}^k \beta_i\sin(\theta(\w_i,\bv))\sigma(\widehat{\vec w_i^{\perp \vec v}}\cdot\x)$, 
which is another sum of ReLUs with coefficients
$b_i = \beta_i \sin(\theta(\w_i,\bv))$ and with hidden weight vectors 
$\widehat{\vec w_i^{\perp \vec v}}$ lying in the subspace $\vec v^{\perp}$. 
Moreover, the component of $\x$ that lies 
in the subspace $\vec v^{\perp}$ 
follows the standard normal distribution 
(the rest is irrelevant, as $y$ depends only on 
${\vec w_i^{\perp \vec v}}\cdot\x = {\vec w_i^{\perp \vec v}}\cdot\x^{\perp \bv}$). 
This means that we can apply $\init$ to this new instance 
(obtained via conditioning) to obtain a new list of unit vectors $\mathcal L$ 
with the following guarantee: there exists $\vec r\in\mathcal L$ such that 
$|\vec r\cdot \widehat{\vec w_1^{\perp \vec v}}|\geq 0.99$. 
\new{Then, updating $\vec v$ along the direction $\br$ with a carefully chosen step-size,}
we obtain a new vector $\vec v'$ such that 
$|\vec v'\cdot\vec w_1|\gg |\vec v\cdot \vec w_1|$, boosting the estimation of $\w_1$. We now formalize this intuition.

\paragraph{Rejection Sampling.}  
\new{For technical reasons,} instead of using a ``hard'' band, i.e., 
conditioning on $|\vec v\cdot\x|\leq \gamma$ for some $\gamma >0$, 
we will use a soft-conditioning method based on rejection sampling. 
Soft-conditioning is more convenient because the marginal distribution 
remains a standard Gaussian while ensuring that the retained samples 
define an instance of \new{the same learning problem} (with different parameters). 
Formally, we use the following rejection sampling procedure:

\begin{fact}[see, e.g., Lemma 4.7 in \cite{DKS18a}]
\label{fact:centered-gaussian-rejection}
Given parameters $\bv\in\Sp^{d-1}$ and $\gamma\in(0,1)$, 
the $(\gamma, \bv)$-rejection procedure \new{$\mathrm{RS}_{\gamma,\bv}$} is defined as follows: 
A vector \new{$\x \in \R^d$} is accepted by \new{$\mathrm{RS}_{\gamma, \bv}$} 
with probability
\begin{equation*}
a_{\gamma,\bv}(\x) \eqdef \exp\left(-\frac{\gamma^{-2}-1}{2}(\bv\cdot\x)^2 \right).
\end{equation*}
Let $\x_\gamma$ denote the sample $\x$ conditioned on being accepted by $\mathrm{RS}_{\gamma, \bv}$. 
Then, \new{for $\x\sim\calN(\vec 0,\mI_d)$,}  we have that 
$\pr_{\x\sim\calN(\vec 0,\mI_d)}[\x \textrm{ is accepted}]=\gamma$
and
$ \x_\gamma \sim \calN\left(\vec 0,\mA_{\gamma,\bv}\right)$ with 
$\mA_{\gamma,\bv} \eqdef \mI_d-(1-\gamma^2)\bv\bv^\top$.  \new{Furthermore, for any $\delta>0$, 
in order to obtain $N$ i.i.d.\ samples from $\x_\gamma$ with probability at least $1-\delta$, 
we need $O(N\log(N/\delta)/\gamma)$ samples from $\normal(\vec 0,\vec I_d)$. }
\end{fact}
For concreteness of discussion, we introduce the following notation and terminology for the conditional distribution of samples following the rejection sampling procedure from \Cref{fact:centered-gaussian-rejection}.

\begin{definition}[$(\gamma, \bv)$-Rejection Sampling Distribution]
\label{def:retained-samples-rejection}
Given a distribution $\D$ supported on $\R^d \times \R$ and samples $(\x,y)\sim \D$, consider the rejection sampling procedure from 
\Cref{fact:centered-gaussian-rejection} applied to the covariates $\x.$ 
Whenever $\x$ is accepted, accept the whole labeled example $(\x, y)$, \new{and let $\bz = \mA_{\gamma,\bv}^{-1/2}\x$.} 
The distribution of transformed accepted samples \new{$(\bz, y)$} 
is then called the $(\gamma, \bv)$-rejection sampling 
distribution obtained from $\cD$, and it is denoted by $\D^{\gamma,\vec v}$. 
\end{definition}

We now argue that $\D^{\gamma,\vec v}$ itself defines an instance of a sum of $k$ ReLUs under 
Gaussian marginals.  

\begin{fact}[Properties of Rejection Sampling] \label{fact:relu-preserved-under-centered-rejection}
Fix parameters $\bv\in\Sp^{d-1}$ and $\gamma\in(0,1)$. Let $(\x,y)\sim \cD$ be such that
   $ \x\sim\calN(\vec 0,\mI_d)$ and  
   \new{$y(\x) = \sum_{i=1}^k\beta_i\sigma(\w_i\cdot\x)$}
with $\|\w_i\|_2=1$ and $\beta_i\in \new{[} 0,1]$ for $i \in [k].$
   \new{Let $(\x_\gamma,y_\gamma)$ be the sample accepted by $\mathrm{RS}_{\gamma,\bv}$.}
Let
$
    \bz
    \eqdef
    \mA_{\gamma,\bv}^{-1/2}\x_\gamma,$ where $\mA_{\gamma,\bv}
    \eqdef
    \mI_d-(1-\gamma^2)\bv\bv^\top.$ 
Then, the following statements both hold:
\begin{enumerate}[leftmargin=*]
    \item $\bz$ follows the standard normal distribution, 
    i.e., $\bz\sim\calN(\vec 0,\mI_d)$.
    \item $\E[y_\gamma\mid\bz] =
    \sum_{i=1}^k
        \widetilde\beta_i
        \sigma(\widetilde\w_i\cdot\bz)$,  
        where  $\overline\w_i = \w_i^{\perp\bv} +\gamma(\bv\cdot\w_i)\bv$, 
        $\widetilde\w_i= \overline\w_i/\|\overline\w_i\|_2,$ and 
        $ \widetilde\beta_i= \beta_i\|\overline\w_i\|_2 \leq \beta_i.$ 
\end{enumerate} 
\end{fact}

\begin{proof}
The acceptance decision depends only on $\x$. Since
$
    \x_\gamma
    =
    \mA_{\gamma,\bv}^{1/2}\bz$,
we have
\begin{align*}
    \E[y_\gamma\mid\bz]
    =
    \sum_{i=1}^k
        \beta_i
        \sigma\Big(
            \w_i^\top
            \mA_{\gamma,\bv}^{1/2}\bz
        \Big)= \sum_{i=1}^k
        \beta_i
        \sigma\Big(
           \big( \mA_{\gamma,\bv}^{1/2}\w_i
    \big)^\top \bz
        \Big) \;.
\end{align*}
We can write $\vec w_i=  \vec w_i^{\perp \vec v} + (\vec v\cdot\vec w_i)\vec v$. Therefore, $\mA_{\gamma,\bv}^{1/2}\w_i=  \vec w_i^{\perp \vec v} + \gamma(\vec v\cdot\vec w_i)\vec v = \overline\w_i$.
Using positive homogeneity of ReLU, namely 
$
    \sigma(\overline\w_i\cdot\bz)
    =
    \|\overline\w_i\|_2
    \sigma(\widetilde\w_i\cdot\bz),$
we get
\[
    \E[y_\gamma\mid\bz]
    =
    \sum_{i=1}^k
        \beta_i \|\overline\w_i\|_2
        \sigma(\widetilde\w_i\cdot\bz)
    =
    \sum_{i=1}^k
        \widetilde\beta_i
        \sigma(\widetilde\w_i\cdot\bz).
\]
Finally, $\widetilde\beta_i \leq \beta_i$ because $ \|\overline\w_i\|_2 \leq \|\w_i\|_2,$ as $\gamma \in (0, 1).$
\end{proof}

With \Cref{fact:relu-preserved-under-centered-rejection} on hand, we now formally argue that $\init$ can be applied to estimate $\w_i^{\perp \bv}.$

\begin{proposition}[Estimating Update Directions]\label{coro:obtain-hidden-directions-using-init-assum}
    Fix a unit vector $\bv\in\R^d$ and let $\lambda,\delta\in(0,1)$. Let $\cD$ be a sum of \new{$k$} 
    ReLUs under Gaussian marginals and, given $\gamma = \lambda/100$, 
    let $\D^{\gamma,\bv}$ be the $(\gamma, \bv)$-rejection sampling distribution obtained from $\cD.$ Then, procedure $\init(\D^{\gamma,\bv},\lambda,\delta)$ that satisfies 
    \Cref{assum:initialization-algorithm} applied to i.i.d.\ samples from $\D^{\gamma,\vec v}$ returns a list $\mathcal L$ such that 
    with probability at least $1-\delta$, for any $\w_i$ with 
    $\beta_i\sin(\theta_\pm(\w_i,\bv))\geq \lambda$, 
    there exists a unit vector $\br_i\in\cL$ with 
    $|\widehat{\br_i^{\perp \bv}}\cdot\wh{\w_i^{\perp\bv}}|\geq 9/10$. 
\end{proposition}

\begin{proof}
Denote the projective angle
between $\w_i$ and $\bv$ by $\theta_i = \theta_\pm(\w_i,\bv)$. From \Cref{fact:relu-preserved-under-centered-rejection}, we have that 
\[
  \E_{(\x,y)\sim \D^{\gamma,\bv}}[y\mid\bx]=    \sum_{i=1}^k
        \widetilde\beta_i
        \sigma(\widetilde\w_i\cdot\bz)\;,
\]
with $\overline\w_i
    =
    \w_i^{\perp\bv}
    +\gamma(\bv\cdot\w_i)\bv$, $\widetilde\w_i= 
    \overline\w_i/\|\overline\w_i\|_2,$ and $
    \widetilde\beta_i=
    \beta_i\|\overline\w_i\|_2.$  Observe that 
    \begin{align}
           \|\widetilde\w_i-\widehat{\w_i^{\perp \vec v}}\|_2^2 &= \|\widetilde\w_i\|_2^2 + \|\widehat{\w_i^{\perp \vec v}}\|_2^2 - 2 \widetilde\w_i^\top \widehat{\w_i^{\perp \vec v}}\notag\\ 
           &=2\left(1-\frac{\|\w_i^{\perp \vec v}\|_2}{\| \w_i^{\perp\bv}
    +\gamma(\bv\cdot\w_i)\bv\|_2}\right)\notag\\
    &=2\left(1-\frac{1}{( 1
    +\gamma^2(\bv\cdot\w_i)^2/\|\w_i^{\perp\bv}\|_2^2)^{1/2}}\right)
    \leq \frac{\gamma^2(\bv\cdot\w_i)^2}{\|\w_i^{\perp \vec v}\|_2^2}\;, \label{eq:tildew_i-whw_i^perp}
    \end{align}
    where the second equality uses that both $\widetilde\w_i,\widehat{\w_i^{\perp \vec v}}$ are 
    unit vectors and applies their definitions to simplify the inner product, while the last 
    inequality uses $1/(1+z^2)^{1/2}\geq 1-z^2/2$ (this can be seen as  
    $\varphi(x) = 1/(1 + x)^{1/2}$, $x \geq 0,$ is convex, by applying 
    $\varphi(x) \geq \varphi(0) + \varphi'(0)x$ and plugging in $x = z^2$). 
    By using that $\beta_i \sin\theta_i\geq \lambda$ from our assumptions, we have that $1/\sin(\theta_i) \leq \beta_i /\lambda \leq 1/\lambda,$ 
    as $\beta_i \in (0, 1].$ Combining with \eqref{eq:tildew_i-whw_i^perp},  
    we thus get that 
    \[\|
    \widetilde\w_i-\widehat{\w_i^{\perp \vec v}}\|_2 \leq \frac{\gamma |\bv \cdot \w_i|}{\norm{\w_i^{\perp \bv}}_2} = \frac{\gamma |\cos(\theta_i)|}{|\sin(\theta_i)|}\leq \gamma/\lambda\leq 1/100.
    \] 
    Therefore, we have that $\init(\D^{\gamma,\bv},\lambda,\delta)$ 
    returns a list $\cL$ containing $\br_i\in\cL$ with $|\br_i\cdot\widetilde{\w}_i|\geq 9.1/10$ for any $\w_i$ satisfying $\beta_i\sin(\new{\theta_i})\geq \lambda$. Since, by the triangle inequality and Cauchy-Schwarz, 
    \[
    |\br_i\cdot\widetilde{\w}_i| \leq |\br_i\cdot\widehat{\w_i^{\perp \vec v}}| + \|\br_i\|_2 \|\widetilde{\w}_i-\widehat{\w_i^{\perp \vec v}}\| \leq |\br_i\cdot\widehat{\w_i^{\perp \vec v}}| + 1/100,\] 
    we get that $|\widehat{\br_i^{\perp \bv}}\cdot\widehat{\w_i^{\perp \vec v}}| \geq |\br_i\cdot\widehat{\w_i^{\perp \vec v}}| \geq 9.1/10 - 1/100 = 9/10$. 
\end{proof}

\new{
Since the rejection sampling procedure works by (approximately) localizing 
the marginal distribution of $\x$  to a thin band $|\bv\cdot\x|\leq \gamma$, 
it is typically referred to as ``localization''.
}
The localization argument in \Cref{coro:obtain-hidden-directions-using-init-assum} allows us to obtain components 
of the \new{hidden weight} vectors that our algorithm potentially has not yet discovered. 
The subroutine \nameref{alg:refine} formalizes the updating procedure to obtain vectors closer 
to the target ones. 

In Steps \eqref{alg-refine:loop-start} through \eqref{alg-refine:loop-end}, 
we guarantee that we have created a new set that contains some vectors 
that are closer than before to some of the hidden vectors, by a multiplicative factor $3/5$. 
We note that for technical reasons (\Cref{prop:cover-and-prune} requires 
a large decrease between iterations), we require the update to decrease 
the angle by a factor larger than $3/5$; for this purpose, in \nameref{alg:refine}, we boost this procedure by applying three update steps instead of one.

\begin{algorithm}[H]
\caption{$\mathsf{Refine}(\vec v,r,\eps_0,\delta_\init;\D)$}
\label{alg:refine}
\begin{algorithmic}[1]
\State \textbf{Input:} unit vector $\vec v\in\Sp^{d-1}$, radius upper bound $r \in(\eps_0, \pi/2]$,
accuracy $\eps_0>0$, 
failure probability $\delta_\init\in(0,1)$, i.i.d.\ access to $\D$.
\State \textbf{Output:} A list of vectors $\cV$ satisfying the guarantees of \Cref{lem:add-vectors}.

\medskip

\State $\chi\gets3/5$, $\cV \gets \{\vec v\}$, $\lambda_\init\gets\eps_0\sin(\chi^3 \eps_0)$, $\gamma \gets \lambda_\init/100$. 
\new{\Comment{$\lambda_\init$: coefficient threshold; $\gamma$: rejection sampling width.}} 
\For{$j=0,1,2$}
    \State $\eta_j\gets\tan(4\chi^j r /5)$. \label{alg-refine:loop-start}

    \For{{all} $\vec v\in\cV$}
\State Construct $\D^{\gamma,\bv}$ by applying the rejection sampling procedure $\mathrm{RS}_{\gamma,\bv}$ (\Cref{fact:centered-gaussian-rejection}).
        \State $\mathcal L_0(\vec v)\gets
        \init(\D^{\gamma,\bv},\lambda_\init,\delta_\init)$ 
        \Comment{Apply the $\init$ procedure that satisfies \Cref{assum:initialization-algorithm}.} \State $\mathcal L(\vec v)\gets \{\widehat{\vec q^{\perp \vec v}}:\vec q \in \mathcal L_0(\vec v)\}$. \new{\Comment{Project the vectors obtained from the initialization procedure to the tangent space of $\bv$ then normalize to unit norm (\Cref{coro:obtain-hidden-directions-using-init-assum}).}}
        \For{$\vec r\in\mathcal L(\vec v)$}
            \State Add 
            $\proj_{\Sp^{d-1}}(\vec v+\eta_j\vec r)$ and
            $\proj_{\Sp^{d-1}}(\vec v-\eta_j\vec r)$ to $\cV$.\label{alg-refine:loop-end} 
            \Statex\new{\Comment{Gradient descent/Refinement step}}
        \EndFor
    \EndFor
\EndFor
\State \textbf{Return: }$\cV$ 
\end{algorithmic}
\end{algorithm}

\new{The following lemma establishes that our \nameref{alg:refine} subroutine 
contracts the distance between the input candidate vector $\bv$ and the 
$\lambda_\init$-heavy hidden vectors $\w_i$, 
and bounds its sample and computational complexities.}

\begin{lemma}[Properties of $\mathsf{Refine}$]\label{lem:add-vectors}
Fix the input parameters to \nameref{alg:refine}: $\bv \in \spd,$ $ r\in [\eps_0, \pi/2)$, $\eps_0\in(0,1/2), \, \delta_\init\in(0,1)$, and a \new{Gaussian sum of $k$ ReLUs} 
distribution $\D$, and let all other algorithm parameters be specified as in the algorithm statement. 
Suppose that for some $i\in[k]$, 
$\beta_i\geq\eps_0$ and $\theta_{\pm}(\bv,\w_i) \leq r$. 
Then, with probability at least $1-O(L_\init^3 \delta_\init)$, 
the list $\cV$ returned by 
\nameref{alg:refine}
contains a vector $\vec v_i\in\cV$ satisfying
\[ \theta_{\pm}(\vec v_i,\w_i) \leq\chi^3 r.
\]
\new{
Moreover, \nameref{alg:refine} draws 
$N_{\mathsf{Refine}} =\widetilde{O}(L_\init^3 N_\init\log(1/\delta_\init)/\gamma)$ 
samples from $\D$,
} runs in time $d\,\poly(N_{\mathsf{Refine}},T_\init)$ and the size of the output list $\cV$ 
is bounded by $|\cV|= O(L_\init^3)$, where $L_\init, N_\init, T_\init$ are the number of vectors returned, 
the sample complexity, and the runtime of $\init(\D^{\gamma,\bv}; \lambda_\init,\delta_{\init})$, 
respectively. \end{lemma}
\begin{proof}
Recall that $\chi\eqdef{3}/{5}$. 
If $ \theta_{\pm}(\bv,\w_i)\leq\chi^3 r$, then $\bv$ already satisfies the condition. 
Otherwise, since we have that $\beta_i\geq \eps_0$ and $r\geq \eps_0$, 
we get that 
$\beta_i \sin\theta_{\pm}(\bv,\w_i)\geq \beta_i\sin(\chi^3 r)\geq \eps_0\sin(\chi^3\eps_0) = \lambda_\init$.
We argue that after each step $j \in \{0,1,2\}$, 
there exists a vector $\bv_i\in\cV$ such that $\chi^{j+1} r\geq \theta_{\pm}(\vec v_i,\w_i)$.
Before the algorithm starts, the input vector $\bv$ satisfies 
$\theta_\pm(\bv,\w_i)\leq  r$, by our assumption. 
Since $\beta_i \sin\theta_{\pm}(\bv,\w_i)\geq \lambda_\init$, 
\Cref{coro:obtain-hidden-directions-using-init-assum} implies that running 
$\init$ with $N_\init$ samples from $\D^{\gamma,\bv}$ returns a list $\mathcal L$ of unit vectors 
that contains a vector $\vec r$ such that $\vec r\perp\vec v$ and
\[
    |\vec r\cdot
    \widehat{\w_i^{\perp\vec v}}|
    \geq 0.9 \;.
\]
The main insight is that $\br$ defines a direction that can improve 
the estimate of the hidden direction $\w_i$. In particular, as argued in the next claim, 
a projected gradient descent-style update of a candidate vector $\bv$ 
along either $\br$ or $-\br$ leads to a new unit vector that is better aligned 
with (and thus closer to) $\w_i$, unless the vector $\bv$ 
was already well-aligned with (i.e., sufficiently close to) $\w_i.$   
\begin{claim}[One-step improvement in projective angle]
\label{lem:one-step-improvement}
Let
$0< r\leq\pi/2$. Let $\vec u,\vec w\in\Sp^{d-1}$ be such that
$
    \theta_\pm(\vec u,\vec w)\leq r.$
Let $\vec r\in\vec u^\perp$ be a unit vector such that
$
    |
    \vec r\cdot
    \widehat{\vec w^{\perp\vec u}}
    |
    \geq 0.9.$
Let
$
    \phi\eqdef\frac{4 r}{5},
    \eta\eqdef\tan\phi,$
and define
\[
    \vec u_+
    \eqdef
    \proj_{\Sp^{d-1}}(\vec u+\eta\vec r),
    \qquad
    \vec u_-
    \eqdef
    \proj_{\Sp^{d-1}}(\vec u-\eta\vec r).
\]
Then
$
    \min_{\vec z\in\{\vec u,\vec u_+,\vec u_-\}}
    \theta_\pm(\vec z,\vec w)
    \leq
    {3 r}/{5} = \chi r.$
\end{claim}

\begin{proof}[Proof of \Cref{lem:one-step-improvement}]
Let $\theta\eqdef\theta_\pm(\vec u,\vec w).$ 
Assume without loss of generality that $\theta(\vec u,\vec w)\leq \pi/2$.
If $\theta\leq 3 r/5$, 
then the vector $\vec u$ already satisfies the conclusion.
Therefore, we assume throughout the proof that
${3 r}/{5}<\theta\leq r.$
    
Let $\vec q= \widehat{\vec w^{\perp{\vec u}}}$
and note that $\vec q\cdot\vec u=0$ and 
$ \vec w = \cos\theta\,\vec u  + \sin\theta\,\vec q.$
Recall that
$|\vec r\cdot\vec q|\geq 0.9.$ 
Assume without loss of generality that $\vec r\cdot\vec q\geq 0.9$ 
and consider $\bu' = \proj_{\Sp^{d-1}}(\vec u+\eta\vec r)$, 
as the other case is captured by the update rule $\vec u_-$ 
with exactly the same analysis.
Letting $ \delta=\theta(\vec r,\vec q)$,
then $\delta \leq \arccos(0.9)<1/2.$
Consider the ideal update, i.e., the update we could take if we knew exactly 
the unknown part $\vec q$: 
\[
    \vec z
    \eqdef
    \cos\phi\,{\vec u}
    +
    \sin\phi\,\vec q \;.
\]
Note that the angle between $\vec u$ and $\vec z$ is $\phi$, and the
angle between $\vec u$ and $\vec w$ is $\theta$. Therefore
\begin{equation}\label{eq:theta-z-w}
    \theta(\vec z,\vec w)=|\theta-\phi|\leq \left|\theta-{4 r}/{5}\right|
    \leq
    { r}/{5}\;, 
\end{equation}
where we used that $\phi=4 r/5$ and $ r\geq\theta\geq 3 r/5$.
Recall that $\vec u'$, the output of the update, was defined by 
\(
    \vec u'
    =
    \proj_{\Sp^{d-1}}
    (\vec u+\eta\vec r) \;,
\)
and observe that 
\[
    \vec u'
    =
    \frac{{\vec u}+\tan\phi\vec r}
         {\sqrt{1+\tan^2\phi}}
    =
    \cos\phi\,{\vec u}
    +
    \sin\phi\,\vec r \;.
\]
Next we bound
$\theta(\vec u',\vec w)$. For that, we need to bound $\theta(\vec u',\vec z)$.
Note that $\|\vec u'-\vec z\|_2=2\sin(\theta(\vec u',\vec z)/2) $. 
Furthermore, by construction, 
$\|\vec u'-\vec z\|_2 = \sin\phi \|\vec q-\vec r\|_2=2\sin\phi \sin(\delta/2)$. 
Therefore, we have that $\sin(\theta(\vec u',\vec z)/2)\leq  \sin\phi \sin(\delta/2)$. 
{Noting that the sine function is concave on $(0,\pi/2)$, 
from Jensen's inequality we have 
$\sin((\delta/2)\, \phi)\geq (\delta/2)\sin\phi + (1-\delta/2)\sin(0)$. 
Using that $\sin(\delta/2)\leq \delta/2$, we get that 
$ \sin\phi \sin(\delta/2)\leq \delta/2\sin(\phi )\leq\sin(\phi \delta/2)$.} 
Therefore, from monotonicity, we get 
that 
\begin{equation}\label{eq:theta-u'-z}
    \theta(\vec u',\vec z)\leq \delta\phi \leq 2 r/5 \;.
\end{equation} 
Since $\theta(\bu,\bv)$ is a metric on the unit sphere, by the triangle inequality, and using \eqref{eq:theta-z-w} and \eqref{eq:theta-u'-z}, we get
\[
    \theta(\vec u',\vec w)
    \leq
    \theta(\vec u',\vec z)
    +
    \theta(\vec z,\vec w)
    \leq
    \frac{2 r}{5}
    +
    \frac{ r}{5}
    =
    \frac{3 r}{5}\;.
\]
Since $\vec u' \in \{\vec u_+,\vec u_-\}$, this gives
\(
    \min_{\vec z\in\{\vec u,\vec u_+,\vec u_-\}}
    \theta_\pm(\vec z,\vec w)
    \leq
    \frac{3 r}{5} \;,
\)
completing the proof of ~\Cref{lem:one-step-improvement}.
\end{proof}

By \Cref{lem:one-step-improvement}, the list $\cV$ contains
a vector $\vec v_i$ such that $ \chi r\geq \theta_{\pm}(\vec v_i,\w_i)$. 
Hence, after step $j=0$, we have $ \chi r\geq \theta_{\pm}(\vec v_i,\w_i)$ 
for some vector $\bv_i\in\cV$. This is the base case of the induction argument.

Now at step $j = \{1,2\}$, by induction, we know that before updating, 
there exists a vector $\bv_i\in\cV$ such that $ \chi^j r\geq \theta_{\pm}(\vec v_i,\w_i)$. 
Suppose $\theta_{\pm}(\vec v_i,\w_i)\geq \chi^{j+1} r$, 
which implies $\beta_i \sin\theta_{\pm}(\bv_i,\w_i)\geq \lambda_\init$. 
Then, using \Cref{coro:obtain-hidden-directions-using-init-assum} again, 
the \new{$\init$ procedure} returns a list $\mathcal L(\bv_i)$ of unit vectors 
that contains a vector $\vec r$ such that $\vec r\perp\vec v_i$ and 
$|\vec r\cdot \widehat{\w_i^{\perp\vec v_i}}| \geq 0.9.$
Then \Cref{lem:one-step-improvement} implies that after updating $\bv_i$, 
inside the list $\mathcal V$ there is a vector $\vec v_i'$ 
so that $ \chi^{j+1} r\geq \theta_{\pm}(\vec v_i',\w_i)$.
Therefore, at step $j=2$, the algorithm produces a vector $\bv\in\cV$ 
such that $\theta_{\pm}(\vec v,\w_i) \leq \chi^3 r$.

\new{At each call, $\init$ uses $N_\init$ samples from $\D^{\gamma,\bv}$. 
By \Cref{fact:centered-gaussian-rejection}, with probability at least $1 - \delta_\init$, 
these $N_\init$ samples from $\D^{\gamma,\bv}$ can be obtained by drawing 
${O}(N_\init\log(N_\init/\delta_\init)/\gamma)$ samples from the original distribution $\D$, 
and passing through the $(\gamma,\bv)$-rejection procedure $\mathrm{RS}_{\gamma,\bv}$.
Conditioned on successfully obtaining these $N_\init$ from $\D^{\gamma,\bv}$,}
the $\init$ procedure 
succeeds with probability 
at least $1 - \delta_\init$, by \Cref{assum:initialization-algorithm} and \Cref{coro:obtain-hidden-directions-using-init-assum}. \new{Therefore, the total number of samples required by $\init$ is ${O}(N_\init\log(N_\init/\delta_\init)/\gamma)$, and the probability of success for one call of $\init$ procedure is at least $(1 - \delta_\init)^2\geq 1 - 2\delta_\init$.}
Recall that for each query vector $\bv\in\cV$, $\init$  
returns a list of $|\cL(\bv)|\leq L_\init$ vectors. 
Therefore, the $\init$ gets called $T = (2|\cL|)^3 = O(L_\init^3)$ times.

Thus, in total, \nameref{alg:refine} uses 
$T N_\init = O(L_\init^3 N_\init\log(N_\init/\delta_\init)/\gamma)$ samples from $\D$, and succeeds with probability $1 - 2T \delta_\init = 1 - O(L_\init^3\delta_\init)$. 
The \new{size} of $\cV$ is $|\cV| = O(L_\init^3)$. 
This completes the proof of \Cref{lem:add-vectors}.
\end{proof}

\subsubsection{Closeness Tester}\label{subsubsec:tester}
The $\mathsf{Refine}$ procedure (\nameref{alg:refine}) returns a large list $\mathcal V$ 
of vectors with the guarantee that for each $\lambda$-heavy hidden weight vector
$\vec w_i$, 
there is at least one vector $\vec v \in \mathcal V$ that is sufficiently close, 
i.e., $\theta_{\pm}(\vec w_i,\vec v)\leq  r$ for some $ r>0$ (\Cref{lem:add-vectors}). 
However, $\mathsf{Refine}$ branches over all possible updates, 
because it has no information on which hidden weight vector the input $\bv$ approximates nor which update directions from $\init$ are useful for improving the approximation. As a result, the output list may contain many irrelevant vectors that provide no useful approximation of any of the hidden weight vectors. To eliminate irrelevant candidate vectors, we design a tester that distinguishes between $\bv \in \cV$ that are $r$-close to at least one heavy hidden weight vector versus $(1 + \tau)r$-far from all hidden weight vectors, for constant $\tau > 0.$ 

 For an even degree $p$, we use the following moment test (see \Cref{fact:hermite-poly}):  \begin{equation*}
 a_p(\bv) \eqdef \frac1{\wh{\sigma}_p}\E[y\hep_p(\bv\cdot\x)]
     =\sum_i\beta_i|\bv\cdot\w_i|^p.
\end{equation*}
\new{The main intuition behind the tester is that the small value of $a_p(\bv)$ 
identifies vectors $\bv$ that are far from all hidden directions $\w_i$. 
To see this, recall first that $|\bv\cdot\w_i| = \cos(\theta_{\pm}(\w_i, \bv)).$ 
If $\bv$ is $r$-close to at least one $\eps_0$-heavy $\w_i$ (i.e., if $\theta_{\pm}(\w_i, \bv) \leq r$), 
then $a_p(\bv) \geq \eps_0 \cos^p(r)$. On the other hand, if $\theta_{\pm}(\w_i, \bv) > (1+\tau)r$ 
for all $i \in [k],$ then $|\bv\cdot\w_i| = 
\cos(\theta_{\pm}(\w_i, \bv)) \leq (1 - \tau r^2/2)\cos(r)$ for all $i \in [k],$ 
and so $a_p(\bv) \leq k (1 - \tau r^2)^p\cos^p(r)$. Choosing $p$ sufficiently large 
($p \gg \log(k/\eps_0)/(\tau r^2)$) then ensures $a_p(\bv) \ll \eps_0 \cos^p(r),$ 
so the two cases can be separated by the tester.}

\begin{algorithm}[h]
   \caption{$\sat(\bv, r,\eps_0,\tau,\delta_\test;\D)$}\label{alg:sat} \begin{algorithmic}[1]
\State \textbf{Input: }Unit vector $\bv$, $\eps_0,\delta_\test>0$, 
$\tau>0$, upper bound $r$, sample access to $\D$. \State \textbf{Output: }$\mathsf{TRUE}$ if there exists $i\in[k]$ such that $\theta_\pm(\vec w_i,\bv)\leq r$ and $\beta_i\geq \eps_0$; $\mathsf{FALSE}$ if for \new{all} $i\in[k]$: $\theta_\pm(\vec w_i,\bv)> (1+\tau)r$.

\medskip

\State $p \gets 2\lceil(1/(\tau r^2))\log(k/\eps_0)\rceil$.
\State Draw $N = \poly(k,1/\eps_0) \new{(1/\delta_\test)}$ samples from $\D$ to formulate the empirical distribution $\wh{\D}$. 
\State $\wh{a}_p\gets (1/\wh{\sigma}_p)\E_{(\x,y)\sim\wh{\D}}[y\hep_p(\bv\cdot\x)]$.
\If{$\wh{a}_p\geq (7/8)\eps_0\cos^p(r)$}
\State \textbf{Return: }$\mathsf{TRUE}$.
\Else
\State \textbf{Return: }$\mathsf{FALSE}$.
\EndIf
\end{algorithmic}
\end{algorithm}

\begin{lemma}[Smallest Angle Tester]\label{lem:smallest-angle-tester}
    Let $\eps_0 \in (0, 1/2],$ $ r \ge \eps_0/(5k)$, let $\tau>0$ be a small absolute constant, and define $p = 2\lceil(1/(\tau r^2))\log(k/\eps_0)\rceil$. Using $N = \poly(k,1/\eps_0) \new{(1/\delta_\test)}$
samples and $O(d N)$ time
with probability at least $1 - \delta_{\sat}$, 
    {\nameref{alg:sat} has the following performance guarantee:}\begin{enumerate}[leftmargin=*]
        \item (Completeness) If there exists $i\in[k]$ such that $\theta_\pm(\w_i,\bv)\leq  r$ and $\beta_i\geq \eps_0$, then \nameref{alg:sat} returns $\mathsf{TRUE}$. 
\item (Soundness) If $\min_{i\in[k]}\theta_\pm(\w_i,\bv) > (1 + \tau) r$, then \nameref{alg:sat} returns $\mathsf{FALSE}$.
\end{enumerate}
\end{lemma}
\begin{proof}
Denote $\theta_i\eqdef \theta_\pm(\w_i,\bv)$
and suppose without loss of generality that $\theta_1 = \min_{i \in [k]}\theta_i$ 
\new{and $\beta_1\geq \eps_0$}.
For even degree $p$, by \Cref{fact:hermite-poly}, we have that 
\begin{align*}
    a_p\eqdef\frac{1}{\wh{\sigma}_p}\Exy[y\hep_p(\bv\cdot\x)] = \frac{1}{\wh{\sigma}_p}\sum_{i=1}^k \E_{\x\sim\calN}[\beta_i\sigma(\w_i\cdot\x)\hep_p(\bv\cdot\x)] = \sum_{i=1}^k \beta_i\cos^p\theta_i,
\end{align*}
where $\wh{\sigma}_p$ is the $p^{\text{th}}$ Hermite coefficient of the activation $\sigma$. 
We first prove the following simple claim, establishing a bound between the empirical estimate $\wh{a}_p$ and the population mean $a_p$. 
\new{Note that in \Cref{clm:empirical-conc}, the number of samples required for our tester has polynomial dependence on $1/\delta_\test$. This can be replaced by logarithmic dependence 
using median-of-means estimation~\cite{Nemirovsky1983Problem} 
instead of the simple empirical estimator used in \nameref{alg:sat}. 
Yet, for our purposes, a polynomial dependence is sufficient.}

\begin{claim} \label{clm:empirical-conc}
With probability at least $1-\delta_\test$, the empirical estimator $\wh{a}_p$ of $a_p$ that uses $N=\poly(k,1/\eps_0)\new{(1/\delta_\test)}$ samples 
satisfies $|\wh{a}_p - a_p|\leq (1/8)\eps_0\cos^p(r)$.
\end{claim}

\begin{proof}[Proof of \Cref{clm:empirical-conc}]
\new{Note that $\wh{a}_p$ is an unbiased estimator of $a_p$. 
The proof of the claim follows by bounding the variance of the estimate 
and applying Chebyshev's inequality.}

We start by noting that the variance of $y\hep_p(\bv\cdot\x)$ is bounded by:
\begin{align}\label{eq:smallest-angle-tester-1}
    \E[y^2\hep_p(\bv\cdot\x)^2]\leq k\sum_{i=1}^k\E_{\x\sim \normal}[\sigma(\w_i\cdot\x)^2\hep_p(\bv\cdot\x)^2]&\leq k\sum_{i=1}^k\E_{\x\sim \normal}[(\w_i\cdot\x)^2\hep_p(\bv\cdot\x)^2]\;.
    \end{align}
Using the standard Hermite identity $z \hep_p(z)=\sqrt{p+1}\hep_{p+1}(z)+\sqrt{p}\hep_{p-1}(z)$ 
(see \Cref{fact:hermite-poly}) and the orthogonality of Hermite polynomials, we get that
$$\E_{\x\sim\calN}[(\vec v\cdot\x)^2\hep_p(\bv\cdot\x)^2] = (p+1)\E_{\x\sim\calN}[\hep_{p+1}(\bv\cdot\x)^2]+p\E_{\x\sim\calN}[\hep_{p-1}(\bv\cdot\x)^2] = 2p+1.$$
Therefore, by \eqref{eq:smallest-angle-tester-1}, since $p\geq 2$, we have:\new{
\begin{align*}
    \E[y^2\hep_p(\bv\cdot\x)^2] &\leq k\sum_{i=1}^k \E_{\x\sim\calN}[(\cos\theta(\w_i,\bv)\bv\cdot\x + \sin\theta(\w_i,\bv)\wh{\w_i^{\perp\bv}}\cdot\x)^2\hep_p(\bv\cdot\x)^2]\\
    &\leq k\sum_{i=1}^k (2\E_{\x\sim\calN}[(\bv\cdot\x)^2\hep_p(\bv\cdot\x)^2] + 2\E_{\x\sim\calN}[(\wh{\w_i^{\perp\bv}}\cdot\x)^2]\E_{\x\sim\calN}[\hep_p(\bv\cdot\x)^2])\\
    &=2k\sum_{i=1}^k ((2p+1) + 1) \leq 6 k^2 p. \end{align*}}
Since $p$ is an even integer, $\wh{\sigma}_p\neq 0$. From Chebyshev's inequality, we get
\begin{align*}
    \pr[|\wh{a}_p - \Exy[(1/\wh{\sigma}_p)y\hep_p(\bv\cdot\x)]|\geq t]\leq \frac{\Exy[y^2\hep_p(\bv\cdot\x)^2]}{\wh{\sigma}_p^2 Nt^2}\leq \frac{6k^2 p}{N \wh{\sigma}_p^2 t^2}.
\end{align*}
Choosing $t = (1/8)\eps_0\cos^p(r)$ and 
$N =(6k^2p)/(\wh{\sigma}_p^2 t^2 \delta_\test) \lesssim k^2 p/(\wh{\sigma}_p^2 \eps_0^2\cos^{2p}(r) \delta_\test)$ then leads to  
$|\wh{a}_p - a_p|\leq (1/8)\eps_0\cos^p(r)$ with probability at least $1-\delta_\test$.

Recall that $\wh{\sigma}_p^2 \gtrsim p^{-3}$ (see \Cref{fct:hermite-relu}) 
and the elementary identity 
$\cos\theta = \exp(-\int_0^\theta \tan(s)\diff{s})\geq \exp(-\theta^2)$.
Therefore, when $\tau$ is chosen as an absolute positive constant, 
plugging in our choice of $p = 2\lceil(1/(\tau r^2))\log(k/\eps_0)\rceil$, 
and the assumption $ r\geq \eps_0/(5k)$, we get that \begin{align*}
    N &\lesssim \bigg(\frac{1}{\delta_\test}
    \bigg)\bigg(\frac{k}{\eps_0}\bigg)^2 \frac{p}{\wh{\sigma}_p^2} \exp(2p r^2) \lesssim\bigg(\frac{1}{\delta_\test}
    \bigg)p^4\bigg(\frac{k}{\eps_0}\bigg)^{2+4/\tau}\lesssim \new{\bigg(\frac{\log^4(k/\eps_0)}{\delta_\test}
    \bigg)\bigg(\frac{k}{\eps_0}\bigg)^{10+4/\tau}}\;.
\end{align*}
This completes the proof of \Cref{clm:empirical-conc}.
\end{proof}

\new{We henceforth condition on the event of \Cref{clm:empirical-conc} and}
\new{we are ready to prove the two statements of the lemma, 
i.e., the completeness and the soundness of the tester.}
Note that a direct calculation yields for $\theta_i > r$:
    \begin{align*}
        \frac{\cos^p\theta_i}{\cos^p(r)} =\exp\bigg(-p\int_r^{\theta_i} \tan(s)\diff{s}\bigg)  &\leq \exp\bigg(-\frac{p(\theta_i^2 -  r^2)}{2}\bigg).
    \end{align*}

\paragraph{Completeness.} 
    Assume first that $\theta_1\leq  r$ and $\beta_1\geq \eps_0$. Then, we have:
    \begin{align*}
        a_p = \sum_{i=1}^k \beta_i\cos^p\theta_i \geq \beta_1\cos^p\theta_1\geq \eps_0\cos^p(r).
    \end{align*}
    \new{By \Cref{clm:empirical-conc},} this implies $\wh{a}_p\geq (7/8)\eps_0\cos^p(r)$.
    Thus, when $\theta_1\leq  r$ and $\beta_1\geq \eps_0$, \Cref{alg:sat} would return $\mathsf{TRUE}$.   

\paragraph{Soundness. } 
Now assume $(1 + \tau) r < \theta_1\leq\theta_i$ for all $i\in[k]$. Then, we have
    \begin{align*}
        a_p &= \sum_{i=1}^k \beta_i\cos^p\theta_i\leq \sum_{i=1}^k\cos^p\theta_i\leq ke^{-p(\theta_i^2 -  r^2)/2}\cos^p(r) \leq ke^{-p(\tau r^2)}\cos^p(r)\leq \frac{1}{2}\eps_0\cos^p(r),
    \end{align*}
    where the last inequality is due to our choice of $p = 2\lceil(1/(\tau r^2))\log(k/\eps_0)\rceil$.
    \new{Similarly, by \Cref{clm:empirical-conc},} this implies $\wh{a}_p\leq (5/8)\eps_0\cos^p(r)$.
    Thus, when no angle $\theta_i$ is smaller than $(1 + \tau) r$, \Cref{alg:sat} would return $\mathsf{FALSE}$.
\end{proof}

\subsubsection{Covering and Pruning}
\label{subsec:cover-refinement-pruning}

\begingroup
\ifdefined\CoverPruneMonochrome
  \definecolor{cpInk}{gray}{0.14}
  \definecolor{cpActive}{gray}{0.12}
  \definecolor{cpPruned}{gray}{0.37}
  \definecolor{cpMuted}{gray}{0.53}
  \definecolor{cpGuide}{gray}{0.70}
  \definecolor{cpFrame}{gray}{0.83}
\else
  \definecolor{cpInk}{RGB}{39,48,58}
  \definecolor{cpActive}{RGB}{26,92,123}
  \definecolor{cpPruned}{RGB}{166,103,30}
  \definecolor{cpMuted}{RGB}{139,148,158}
  \definecolor{cpGuide}{RGB}{163,177,187}
  \definecolor{cpFrame}{RGB}{212,220,225}
\fi
\tikzset{
 cp/base/.style={font=\fontsize{8.5}{10.5}\selectfont,
                 text=cpInk,inner sep=1.2pt},
 cp/title/.style={cp/base,font=\fontsize{10.4}{12.5}\selectfont\bfseries},
 cp/header/.style={cp/base,font=\fontsize{9.2}{11.2}\selectfont\bfseries},
 cp/small/.style={cp/base,font=\fontsize{7.9}{9.5}\selectfont},
 cp/depthlabel/.style={cp/small,anchor=east,align=right},
 cp/active/.style={circle,draw=cpActive,fill=cpActive,
                   minimum size=6.1pt,inner sep=0pt,line width=.7pt},
 cp/history/.style={circle,draw=cpMuted,fill=cpMuted,
                   minimum size=4.8pt,inner sep=0pt,line width=.6pt},
 cp/candidate/.style={circle,draw=cpActive,fill=white,
                      minimum size=5.7pt,inner sep=0pt,line width=.8pt},
 cp/witness/.style={cp/candidate,line width=1.1pt,minimum size=6.3pt},
 cp/pruned/.style={diamond,draw=cpPruned,fill=white,
                   minimum size=6.1pt,inner sep=0pt,line width=.85pt},
 cp/branch/.style={draw=cpMuted!80,line width=.6pt},
 cp/faded/.style={draw=cpMuted!40,line width=.55pt},
 cp/track/.style={draw=cpActive,line width=1.05pt},
 cp/replacement/.style={draw=cpPruned,line width=.85pt,
              dash pattern=on 2.3pt off 1.8pt,
              -{Stealth[length=1.55mm,width=1.15mm]},
              shorten <=1.5pt,shorten >=1.5pt},
 cp/depth/.style={draw=cpGuide,line width=.5pt,
                 dash pattern=on .45pt off 2.2pt,line cap=round},
 cp/rule/.style={draw=cpFrame,line width=.45pt},
 cp/explain/.style={cp/small,text width=4.16cm,align=center,anchor=north},
}
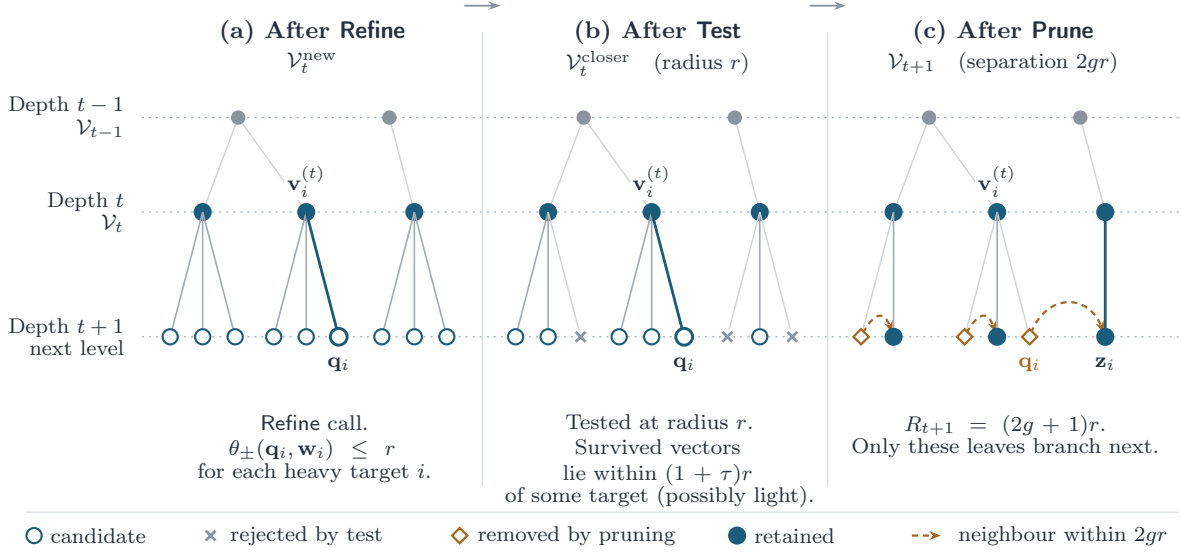
\begin{figure}

\begin{tikzpicture}[x=1cm,y=1cm,cp/base]
\path[use as bounding box] (0,-1.97) rectangle (15.34,6.25);

\foreach \yy in {4.30,3.05,1.40}{
  \draw[cp/depth] (1.58,\yy)--(15.29,\yy);
}
\node[cp/depthlabel] at (1.39,4.30)
  {Depth $t-1$\\[-1pt]$\mathcal V_{t-1}$};
\node[cp/depthlabel] at (1.39,3.05)
  {Depth $t$\\[-1pt]$\mathcal V_t$};
\node[cp/depthlabel] at (1.39,1.40)
  {Depth $t+1$\\[-1pt]next level};

\draw[cp/rule] (6.08,5.43)--(6.08,-.58);
\draw[cp/rule] (10.65,5.43)--(10.65,-.58);
\node[cp/header] at (3.83,5.43) {(a) After \textsf{Refine}};
\node[cp/base] at (3.83,5.02) {$\mathcal V_t^{\mathrm{new}}$};
\node[cp/header] at (8.40,5.43) {(b) After \textsf{Test}};
\node[cp/base] at (8.40,5.02)
  {$\mathcal V_t^{\mathrm{closer}}$\quad(radius $r$)};
\node[cp/header] at (12.97,5.43) {(c) After \textsf{Prune}};
\node[cp/base] at (12.97,5.02)
  {$\mathcal V_{t+1}$\quad(separation $2gr$)};
\draw[cpMuted,-{Stealth[length=1.6mm,width=1.1mm]},line width=.7pt]
  (5.84,5.78)--(6.32,5.78);
\draw[cpMuted,-{Stealth[length=1.6mm,width=1.1mm]},line width=.7pt]
  (10.41,5.78)--(10.89,5.78);

\foreach \panel/\xx in {1/1.70,2/6.27,3/10.84}{
 \begin{scope}[xshift=\xx cm]

  \coordinate (cp-\panel-olda) at (1.15,4.30);
  \coordinate (cp-\panel-oldb) at (3.15,4.30);
  \foreach \xpos/\id in {.68/1,2.05/2,3.48/3}{
   \coordinate (cp-\panel-parent\id) at (\xpos,3.05);
  }
  \draw[cp/faded] (cp-\panel-olda)--(cp-\panel-parent1);
  \draw[cp/faded] (cp-\panel-olda)--(cp-\panel-parent2);
  \draw[cp/faded] (cp-\panel-oldb)--(cp-\panel-parent3);
  \node[cp/history] at (cp-\panel-olda) {};
  \node[cp/history] at (cp-\panel-oldb) {};
  \foreach \id in {1,2,3}{
   \node[cp/active] at (cp-\panel-parent\id) {};
  }
  \node[cp/base,anchor=south,fill=white,inner sep=1pt]
    at (2.05,3.21) {$\mathbf v_i^{(t)}$};

  \foreach \xpos/\leaf in {.25/1,.68/2,1.11/3,1.62/4,2.05/5,2.48/6,3.05/7,3.48/8,3.91/9}{
   \coordinate (cp-\panel-leaf\leaf) at (\xpos,1.40);
  }

  \ifnum\panel=1

   \foreach \parent/\leaf in {1/1,1/2,1/3,2/4,2/5,2/6,3/7,3/8,3/9}{
    \draw[cp/branch] (cp-\panel-parent\parent)--(cp-\panel-leaf\leaf);
   }
   \draw[cp/track] (cp-\panel-parent2)--(cp-\panel-leaf6);
   \foreach \leaf in {1,2,3,4,5,7,8,9}{
    \node[cp/candidate] at (cp-\panel-leaf\leaf) {};
   }
   \node[cp/witness] at (cp-\panel-leaf6) {};
   \node[cp/base,anchor=north] at (2.48,1.18) {$\mathbf q_i$};
  \fi

  \ifnum\panel=2

   \foreach \parent/\leaf in {1/1,1/2,2/4,2/5,2/6,3/8}{
    \draw[cp/branch] (cp-\panel-parent\parent)--(cp-\panel-leaf\leaf);
   }
   \draw[cp/track] (cp-\panel-parent2)--(cp-\panel-leaf6);
   \foreach \parent/\leaf in {1/3,3/7,3/9}{
    \draw[cp/faded] (cp-\panel-parent\parent)--(cp-\panel-leaf\leaf);
    \draw[white,line width=3.0pt]
      ($(cp-\panel-leaf\leaf)+(-.065,-.065)$)--($(cp-\panel-leaf\leaf)+(.065,.065)$)
      ($(cp-\panel-leaf\leaf)+(-.065,.065)$)--($(cp-\panel-leaf\leaf)+(.065,-.065)$);
    \draw[cpMuted,line width=1.05pt]
      ($(cp-\panel-leaf\leaf)+(-.065,-.065)$)--($(cp-\panel-leaf\leaf)+(.065,.065)$)
      ($(cp-\panel-leaf\leaf)+(-.065,.065)$)--($(cp-\panel-leaf\leaf)+(.065,-.065)$);
   }
   \foreach \leaf in {1,2,4,5,8}{
    \node[cp/candidate] at (cp-\panel-leaf\leaf) {};
   }
   \node[cp/witness] at (cp-\panel-leaf6) {};
   \node[cp/base,anchor=north] at (2.48,1.18) {$\mathbf q_i$};
  \fi

  \ifnum\panel=3

   \foreach \parent/\leaf in {1/1,2/4,2/6}{
    \draw[cp/faded] (cp-\panel-parent\parent)--(cp-\panel-leaf\leaf);
   }
   \foreach \parent/\leaf in {1/2,2/5,3/8}{
    \draw[cp/branch] (cp-\panel-parent\parent)--(cp-\panel-leaf\leaf);
   }
   \draw[cp/track] (cp-\panel-parent3)--(cp-\panel-leaf8);
   \draw[cp/replacement] (cp-\panel-leaf1)
      .. controls (.43,1.74) and (.56,1.74) .. (cp-\panel-leaf2);
   \draw[cp/replacement] (cp-\panel-leaf4)
      .. controls (1.80,1.74) and (1.93,1.74) .. (cp-\panel-leaf5);
   \draw[cp/replacement] (cp-\panel-leaf6)
      .. controls (2.70,1.98) and (3.26,1.98) .. (cp-\panel-leaf8);
   \foreach \leaf in {1,4,6}{\node[cp/pruned] at (cp-\panel-leaf\leaf) {};}
   \foreach \leaf in {2,5,8}{\node[cp/active] at (cp-\panel-leaf\leaf) {};}
   \node[cp/base,anchor=north,text=cpPruned] at (2.48,1.18)
     {$\mathbf q_i$};
   \node[cp/base,anchor=north] at (3.48,1.18)
     {$\mathbf z_i$};
  \fi
 \end{scope}
}

\node[cp/explain] at (3.83,.42)
 {\textsf{Refine} call.\\
  $\theta_\pm(\mathbf q_i,\mathbf w_i)\leq r$\\[-1pt]
  for each heavy target $i$.};
\node[cp/explain] at (8.40,.42)
 {Tested at radius $r$.\\
  Survived vectors lie within $(1+\tau)r$\\[-1pt]
  of some target (possibly light).};
\node[cp/explain] at (12.97,.42)
 {
  $R_{t+1}=(2g+1)r$.\\[-1pt]
  Only these leaves branch next.};

\draw[cp/rule] (.04,-.92)--(15.29,-.92);
\node[cp/candidate] at (.15,-1.23) {};
\node[cp/small,anchor=west] at (.32,-1.23) {candidate};
\draw[cpMuted,line width=1.05pt]
 (2.42,-1.295)--(2.55,-1.165) (2.42,-1.165)--(2.55,-1.295);
\node[cp/small,anchor=west] at (2.70,-1.23) {rejected by test};
\node[cp/pruned] at (5.78,-1.23) {};
\node[cp/small,anchor=west] at (5.97,-1.23) {removed by pruning};
\node[cp/active] at (9.45,-1.23) {};
\node[cp/small,anchor=west] at (9.63,-1.23) {retained};
\draw[cp/replacement] (11.71,-1.23)--(12.20,-1.23);
\node[cp/small,anchor=west] at (12.33,-1.23) {neighbour within $2gr$};
\end{tikzpicture}
    \caption{\textbf{List of candidate vectors over the different stages of $\mathsf{CoverAndPrune}$}. \textbf{Left}: after $\mathsf{Refine}$ is called, each $\bv \in \cV_t$ produces a large list of new candidate vectors; all produced lists combined contain, for each $\lambda$-heavy $\w_i,$ at least one candidate $\bv_i$ with $\theta_{\pm}(\bv_i, \w_i) \leq r$. \textbf{Center}: $\mathsf{Test}$ procedure is used to prune all candidates $\bv$ far from all $\w_i,$ that is, all $\bv$ such that $\theta_{\pm}(\bv, \w_i) > (1+\tau)r$, $\forall i \in [k].$ \textbf{Right}: $\mathsf{Prune}$ procedure makes a single iterative pass through the candidate vectors that survive after $\mathsf{Test},$ forming the output list $\cV_{t+1};$ it starts with an empty list and adds a new iterate if and only if it is at least $r$-far from all the candidates already in the list. } \label{fig:one_step}
\end{figure}
\endgroup

\new{The approach described thus far, based on the $\mathsf{Refine}$ subroutine 
and the $\mathsf{Test}$ procedure, provides a concrete way of constructing lists 
of candidate vectors that progressively improve approximations of $\eps_0$-heavy hidden weight vectors. 
In particular, suppose that the algorithm starts with a list of candidates $\cV_t$ 
such that for each $\eps_0$-heavy $\w_i,$ $i \in [k],$ there exists 
at least one $\bv_i \in \cV_t$ such that $\theta_{\pm}(\bv_i, \w_i) \leq R_t.$ 
We say in this case that $\w_i$ is covered by $\bv_i$ at resolution $R_t.$} 

\new{A call to $\mathsf{Refine}$ for each $\bv_i \in \cV_t$ creates a new list of candidate vectors, 
within which only one candidate vector is guaranteed to land significantly closer to $\w_i,$ with 
$\theta_{\pm}(\bv_i, \w_i) \leq \chi^3 R_t$. However, 
it is unclear how to identify such a candidate. 
What is possible to identify are the candidate vectors that are far 
from all hidden weight vectors; this is where our $\mathsf{Test}$ procedure is used. 
Thus, $\mathsf{Refine}$ and $\mathsf{Test}$ can be used jointly to construct 
a list $\cV_t'$ such that all $\eps_0$-heavy $\w_i,$ $i \in [k],$ 
are covered at resolution slightly larger than $\chi^3 R_t$ 
(as testing increases the resolution slightly), but still a constant factor smaller than $R_t.$ 
Since the resolution improves by a constant factor, repeating this combination 
of $\mathsf{Refine}$ and $\mathsf{Test}$ geometrically reduces the resolution 
and leads to covering all $\eps_0$-heavy hidden weight vectors 
at resolution $\eps/k$ within order-$\log(k/\epsilon)$ repetitions. 
Unfortunately, the described approach does not run in polynomial-time, 
since there is no guarantee that the number of candidate vectors in our list 
covering the heavy hidden weight vectors 
is small; as a result, as described, the list $\cV_t$ could grow super-polynomially.} 

\new{To prevent such a size blow up, we additionally introduce a simple pruning procedure 
(\nameref{alg:global-prune}). The procedure takes as input the list $\cV_t'$ produced 
after $\mathsf{Refine}$ and $\mathsf{Test}$ are used. 
It then creates $\cV_{t+1}$, starting from an empty list as follows.  
It traverses $\cV_t'$ element by element, and adds the traversed element to $\cV_{t+1}$ 
if and only if it is sufficiently far from all elements already in $\cV_{t+1}$. 
The intuition behind this method is that if two vectors $\vec v,\vec u$ are very close to each other, 
then both of them have similar correlation with the hidden weight vectors.  
Thus, we can remove one of them without substantially altering 
the coverage of the hidden weight vectors.  The whole sequence of steps following $\mathsf{Refine},$ $\mathsf{Test},$ and $\mathsf{Prune}$ is illustrated in \Cref{fig:one_step}.} 

\new{We argue that the last pruning step (\nameref{alg:global-prune}) ensures that all heavy hidden weight vectors remain covered by vectors in $\cV_{t+1}$ at a slightly larger resolution $R_{t+1} \leq \rho R_t,$ $\rho \in (0, 1)$, but only up to $k$ candidate vectors are added to $\cV_{t+1}.$} \new{More broadly, the main procedure of this section,} 
\nameref{alg:cover-and-prune}, guarantees the following: in each iteration $t \geq 0$ of $\mathsf{CoverAndPrune}$, the following invariant holds: \begin{equation}\tag{Invariant\textsubscript{$t$}}\label{eq:induction-cond}
\begin{gathered}
{\text{for all } } \eps_0\text{-heavy } \w_i,\
\text{there exists}\  \vec v_i^{(t)}\in \cV_t
\ \text{such that}\
\theta_\pm(\vec v_i^{(t)},\bw_i)\le R_t, \text{ and}
\\
|\cV_t|\le k.
\end{gathered}
\end{equation}

\begin{algorithm}[H]
\caption{$\mathsf{Prune}(\cV,r)$}
\label{alg:global-prune}
\begin{algorithmic}[1]
\State \textbf{Input: } A finite set $\cV\subseteq\spd$ and a separation radius $r>0$.
\State \textbf{Output: } A finite set $\cV_{\rm output}\subseteq\spd$ such that every pair of elements is separated by $r$ in $\theta_\pm$.

\medskip

\State $\cV_{\rm output}\gets\emptyset$.
\For{$\bv\in \cV$}
    \If{$\theta_\pm(\bv,\bv')>r$ for every $\bv'\in \cV_{\rm output}$}
        \State $\cV_{\rm output}\gets \cV_{\rm output}\cup\{\bv\}$.
    \EndIf
\EndFor
\State \textbf{Return: } $\cV_{\rm output}$.
\end{algorithmic}
\end{algorithm}

\begin{algorithm}[H]
\caption{$\mathsf{CoverAndPrune}(\eps_0,\delta;\D)$}
\label{alg:cover-and-prune}
\begin{algorithmic}[1]
\State \textbf{Input: } \new{Coefficient} threshold $\eps_0>0$, failure probability $\delta\in(0,1)$, sample access to $\D$.
\State \textbf{Output: }A list $\widehat \cV$ of \new{unit} vectors satisfying 
\Cref{prop:cover-and-prune}.

\medskip

\State Tester accuracy $\tau\gets 1/200$, ratio $g\gets 1 + 2\tau$,
radius shrinking rate \new{$\chi\gets (3/5)$}. 
\State Initial radius $r_{0}\gets\arccos(9/10)$, 
error threshold $\lambda_\init \gets \eps_0\new{\sin(\chi^3\eps_0)}$.
\State Relaxed initial radius $R_0\gets(2g+1)r_{0}$,
relaxed radius shrinking rate $\rho\gets(2g+1)\new{\chi^3}$.
\State 
Number of iterations $
T\gets
\lceil\frac{\log(R_0/\eps_0)}{\log(1/\rho)}\rceil.$
\State Failure probability of immediate steps $\delta'\gets O(\delta/(\new{T} k L_\init^3) )$
\State $\mathcal L\gets
        \init(\D,\lambda_\init,\delta')$. \Comment{Apply the Coarse Direction Recovery procedure.}
\State $\cL'\gets
\{\bv\in \cL:
\sat(\bv,r_{0},\eps_0,\tau,\delta';\D)=\mathsf{TRUE}\}$. \label{Line:cover-prune-test-1}
\State $\cV_0\gets\mathsf{Prune}(\cL',2g r_{0})$.\Comment{Initialized vector set satisfying the induction  \eqref{eq:induction-cond}.}\label{Line:cover-prune-prune-1}
\For{$t=0,1,\ldots,T-1$}
    \State $r_{t+1}\gets \new{\chi^3} R_t$ and $\cV_t^{\rm new}\gets \cV_t$.
    \For{$\vec v\in \cV_t$}
        \State $\cV_t^{\rm new}\gets \cV_t^{\rm new}\cup
        \mathsf{Refine}(\vec v,\new{R_t},\eps_0,\delta';\D)$. 
        \Comment{``Gradient Descent/Refinement'' step; Vectors in $\cV_t^{\rm new}$ 
        have smaller distance to the hidden weight vectors compared to $\cV_t$}\label{Line:cover-prune-refinement-step} \EndFor
    \State $\cV_t^{\rm closer}\gets\{\bv\in \cV_t^{\rm new}:
    \sat(\bv,r_{t+1},\eps_0,\tau,\delta';\D)=\mathsf{TRUE}\}$. 
    \label{Line:cover-prune-test-2}
    \Statex \Comment{Discard candidate vectors that do not cover hidden weight vectors.}
    \State $\cV_{t+1}\gets\mathsf{Prune}(\cV_t^{\rm closer},2g r_{t+1})$. \label{Line:cover-prune-prune-2}\Comment{Discard candidate vectors covering the same hidden weight vectors.}
    \State $R_{t+1}\gets(2g+1)r_{t+1}$. \Comment{$\cV_{t+1}$ satisfies \eqref{eq:induction-cond} at $t+1$ with parameter $R_{t+1}$.}
\EndFor
\State \textbf{Return: } $\widehat \cV\gets \cV_T$.
\end{algorithmic}
\end{algorithm}

\new{We briefly discuss how \nameref{alg:cover-and-prune} preserves the induction invariant at each step.
The refinement subroutine in Line~\ref{Line:cover-prune-refinement-step} generates new vectors 
that are closer to the hidden weights, but also produces many redundant and unwanted candidates.
In order to guarantee that the vectors we keep cover the hidden weight vectors 
within $R_t$, we use our testing routine from the previous subsection (\nameref{alg:sat}). 
This procedure, applied in Line~\ref{Line:cover-prune-test-2}, 
allows us to discard candidate vectors that 
do not have any hidden weight vector within radius $R_t$ (in $\theta_\pm$). 
Then in Line~\ref{Line:cover-prune-prune-2}, the Pruning subroutine removes redundant candidates, 
since two well-separated vectors cannot both be
near the same hidden weight vector.
By the triangle inequality, after the pruning step, we guarantee that each candidate vector covers 
hidden weight vectors within a relaxed radius, $R_{t+1} = (2g+1)r_{t+1} = \rho R_t$, with $\rho<1$.
Thus, as a result, 
at most $k$ candidates remain in $\cV_{t+1}$, and every hidden weight with non-negligible coefficient 
is covered by a candidate vector within radius $R_{t+1}$, 
i.e., the invariant \eqref{eq:induction-cond} is satisfied after each update.
}

\begin{proposition}[Covering and Pruning]
\label{prop:cover-and-prune}
Assume access to an algorithm $\init$ that satisfies \Cref{assum:initialization-algorithm}. 
Given parameters $\eps_0,\delta>0$,
\nameref{alg:cover-and-prune} draws
$$N = \widetilde O(L_\init^3(L_\init^3 + N_\init)) \, \poly(k, 1/\eps_0, 1/\delta)$$
i.i.d.\ samples from $\D$,
runs in $\poly(d,N,T_\init)$ time, and returns a list $\wh{\cV}$ such that with probability 
at least \(1-\delta\), the following holds: 
\begin{enumerate}[leftmargin=*]
    \item  $|\wh{\cV}|\le k$; and 
    \item  For all $\eps_0$-heavy $\w_i,$ there exists \(\widehat\bv\in\widehat \cV\) such
that
$
\theta_\pm(\widehat\bv,\bw_i)\le \eps_0.$
\end{enumerate}
\end{proposition}
\begin{proof}
Let $g,\chi,r_0,R_0$, and $\rho$ be defined as in \nameref{alg:cover-and-prune}. For iteration $t$,  let 
    \[
    R_{t+1}=\rho R_t = (2g+1)\chi^3 R_t = (2g+1)r_{t+1}.
    \]
By \Cref{assum:initialization-algorithm}, we have that running the 
$\init$ routine on $\D$ returns a list of unit 
    vectors $\cL$ 
    such that for each \new{heavy} $\w_i$, $i\in[k],$ there exists $\bv_i\in\cL$ 
    satisfying $\theta_\pm(\bv_i,\w_i)\leq \arccos(9/10) = r_0$.
After \Cref{Line:cover-prune-test-1}, the algorithm \nameref{alg:cover-and-prune} keeps only the vectors that pass 
the smallest-angle test. These are the candidate vectors that 
cover the \new{heavy} hidden weight vectors within radius $r_{0}$, i.e.,
\[
\cL'
\eqdef
\left\{
\bv\in \cL:
\mathsf{Test}(\bv,r_{0},\eps_0,\tau,\delta_\test;\D)=\mathsf{TRUE}
\right\},
\]
and then in \Cref{Line:cover-prune-prune-1} prunes them to obtain 
the initial list $\cV_0$:
\[
\cV_0\eqdef \mathsf{Prune}(\cL', 2g r_{0}) \;.
\]
Our goal is to show that for each $t\in[T]$, \eqref{eq:induction-cond} holds. We prove this claim by induction on $t.$

\begin{claim}[Induction Base] \label{clm:ind-basis}
At $t=0$,  \eqref{eq:induction-cond} condition is satisfied.
\end{claim}
\begin{proof}
By \Cref{coro:obtain-hidden-directions-using-init-assum}, the call to $\init$ in \new{Line 8} 
gives a list $\cL$ such that for every $\eps_0$-heavy $\w_i$, 
there exists \(\bv_i\in \cL\) with the property that
\[
\theta_\pm(\bv_i,\bw_i)\le r_{0}.
\]
Further, after the testing step in \Cref{Line:cover-prune-test-1}, 
for any $\eps_0$-heavy $\w_i$, any vector $\bv_i\in\cL$ such that 
$\theta_\pm(\bv_i,\w_i)\leq r_0$ will be kept in $\cL'$, 
by \Cref{lem:smallest-angle-tester}.
If \(\bv_i\) survives pruning in \Cref{Line:cover-prune-prune-1}, it covers
\(\bw_i\) with radius $r_0$ (i.e., $\theta_{\pm}(\bv_i,\w_i)\leq r_0$). If, on the other hand, $\bv_i$ is discarded, then there must exist some other vector \(\bv_i'\in \cV_0\) that satisfies
\[
\theta_\pm(\bv_i',\bv_i)\le 2g r_{0}.
\]
Hence, by the triangle inequality, we get that
\[
\theta_\pm(\bv_i',\bw_i)
\le\theta_\pm(\bv_i',\bv_i) + \theta_\pm(\bv_i,\bw_i)\leq
2gr_{0} + r_{0}
=
R_0.
\]
This establishes the second part of the claim. 
It remains to bound the size of the list.

Every vector in \(\cV_0\) passed the smallest-angle test 
with cutoff
\(r_{0}\). By \Cref{lem:smallest-angle-tester}, every such vector is 
within
\((1+\tau) r_{0}\) of some hidden weight vector \(\bw_j\) with respect to $\theta_{\pm}$. 
If two retained vectors $\vec v_j,\vec v_j'$ cover the same
\(\bw_j\) with radius $(1+\tau)r_0$, then, since $g = 1+2\tau$, we would have 
\[
\theta_\pm(\bv_j',\bv_j)\leq \theta_\pm(\bv_j',\bw_j)+\theta_\pm(\bw_j,\bv_j)\leq 2 (1+\tau)r_{0}<2g r_{0}
\]
contradicting the pruning rule 
in \Cref{Line:cover-prune-prune-1}. Therefore, we have
$|\cV_0|\le k$, completing the proof of \Cref{clm:ind-basis}.
\end{proof}

Now it remains to show that the invariant holds for all $t>0$.
This is done in the following claim. 

\begin{claim}[Inductive Step] \label{clm:induction}
Assume that \eqref{eq:induction-cond} holds at step $t$.
Then, \new{conditioned on successful $\mathsf{Refine}$ and $\mathsf{Test}$ procedures in \Cref{Line:cover-prune-refinement-step,Line:cover-prune-test-2}}, it also holds at step $t+1$.
\end{claim}
\begin{proof}
    Fix an $\eps_0$-heavy $\w_i$ and choose
\(\vec v\in \cV_t\) such that
$\theta_\pm(\vec v,\bw_i)\le R_t$. 
Such a $\vec v$ is guaranteed to exist 
because \eqref{eq:induction-cond} holds for step $t$.
In \Cref{Line:cover-prune-refinement-step} of \nameref{alg:cover-and-prune}, 
the 
$\mathsf{Refine}$ subroutine (\nameref{alg:refine}) is called with radius $R_t$. Therefore, by \Cref{lem:add-vectors},  it gives a list $\cV_t^{\rm new}$ that 
contains a unit 
vector $\vec v'\in \cV_t^{\rm new}$ such that 
$\theta_\pm(\vec v',\bw_i)\leq \chi^3 R_t = r_{t+1}$ (since $\w_i$ is $\eps_0$-heavy, $\beta_i \ge \eps_0$). Thus, by \Cref{lem:smallest-angle-tester}, 
$\bv'$ will pass the test in \Cref{Line:cover-prune-test-2}, i.e.,
$
\sat(\bv,r_{t+1},\eps_0,\tau,\delta';\D)
$, 
and $\vec v' \in \cV_t^{\rm closer}$.

It remains to show that after pruning (\Cref{Line:cover-prune-prune-2}), $\cV_{t+1}$ covers all $\eps_0$-heavy $\vec w_i$'s at resolution $R_{t+1}.$ 
\new{Recall that, in \Cref{Line:cover-prune-prune-2}, \nameref{alg:global-prune} is applied to $\cV = \cV_{t}^{\mathrm{closer}}$ with separation radius $r = 2g r_{t+1}$.}
Fix any $\eps_0$-heavy $\w_i$ and let $\vec v_i\in \cV_t^{\rm closer}$ be 
such that $\theta_\pm(\bv_i,\bw_i)\le r_{t+1} = \chi^3 R_{t}$ 
(we have argued above that this vector will pass the smallest-angle test).
If \(\bv_i\in \cV_{t+1}\), it covers \(\bw_i\) 
at scale \(\chi^3 R_{t}\leq R_{t+1}\). 
If it is discarded, then some retained vector 
\(\bv_i'\in \cV_{t+1}\) satisfies
\[
\theta_\pm(\bv_i',\bv_i)\le 2g r_{t+1}.
\]
Hence,
\[
\theta_\pm(\bv_i',\bw_i)
\le \theta_\pm(\bv_i',\bz_i)+\theta_\pm(\bv_i,\bw_i)
\le (2g+1)r_{t+1}
= R_{t+1} \;.
\]

It remains to bound \(|\cV_{t+1}|\). Similar to the base argument, 
since every vector in $\cV_{t+1}$ passed the smallest-angle test 
in \Cref{Line:cover-prune-test-2}, 
by \Cref{lem:smallest-angle-tester}, 
for every $\bv_j\in\cV_{t+1}$ there is a hidden weight vector 
$\w_j$ such that $\theta_{\pm}(\bv_j,\w_j)\leq (1+\tau)r_{t+1}$. 
\new{For the sake of contradiction,} suppose there were 
two distinct vectors $\bv_j, \bv_j'$ both covering $\w_j$ 
retained in $\cV_{t+1}$. Then, by the triangle inequality, 
$\theta_\pm(\bv_j',\bv_j)\leq 2(1+\tau)r_{t+1}<2g r_{t+1}$, 
indicating that one of them will not pass the pruning step in \Cref{Line:cover-prune-prune-2}. 
\new{This provides the desired contradiction and establishes that 
\(|\cV_{t+1}| \leq k\).}
\end{proof}

\new{To complete the proof of \Cref{prop:cover-and-prune}, 
it remains to bound the sample and computational complexities.}
Note that
$
R_t=\rho^t R_0$. Choosing
\[
T=
\left\lceil
\frac{\log(R_0/\eps_0)}{\log(1/\rho)}
\right\rceil,
\]
we have $
R_T\le \eps_0.$
Thus, for all $\eps_0$-heavy $\w_i$, 
there exists \(\widehat\bv_i\in \cV_T\) such that
$\theta_\pm(\widehat\bv_i,\bw_i)\le \eps_0.$

 Let $L_\init$ be the uniform upper bound on the length of the list 
 returned by the $\init$ subroutine, $N_\init$ be the {corresponding }sample complexity, 
 and $\delta_\init$ be the $\init$ subroutine's failure probability.

 At each iteration $t\in [T]$, the smallest angle tester (\nameref{alg:sat}) 
 is called 
 $$|\cV_t^{\rm new}| =|\cV^t|\times |\text{number of vectors returned by $\mathsf{Refine}$}| = O(kL_\init^3)$$ many times, where 
 we used the fact that 
 $$\text{number of vectors returned by $\mathsf{Refine}$} = O(L_\init^3)$$
 from \Cref{lem:add-vectors}. Therefore, by a union bound over the success events of the tester routine,
the probability that all $\sat$ calls succeed at step $t$ is 
 $1 - O(k L_\init^3\delta_\test)$. 

Meanwhile, at step $t$, the $\mathsf{Refine}$ subroutine 
is called $|\cV_t|\leq k$ times. 
Therefore, by a union bound over all the calls of $\mathsf{Refine}$ for all $T$ rounds and by \Cref{lem:add-vectors}, with probability at least 
$1 - O(TkL_\init^3\delta_\init)$, the $\mathsf{Refine}$ routine successfully
returns a list $\cV_t^{\rm new}$ containing new candidate vectors 
that are closer to the hidden weight vectors compared to the ones in $\cV_t$, at each step $t\in[T]$. 
Thus, by setting $\delta_\init =\delta'= O(\delta/(T k L_\init^3))$ and 
$\delta_\test =\delta'= O(\delta/(T k L_\init^3))$, the subroutine \nameref{alg:cover-and-prune}
succeeds with probability at least 
$1 - O(TkL_\init^3(\delta_\init+\delta_\test)) \geq 1-\delta$.

By \Cref{lem:smallest-angle-tester}, it follows that the tester 
$\mathsf{Test}$ succeeds with 
$(1/\delta_\test) \poly(k/\eps_0) = \poly(k/\eps_0) TkL_\init^3/\delta$ samples for each run. By \Cref{lem:add-vectors}, the $\mathsf{Refine}$ 
subroutine requires $\widetilde{O}(L_\init^3 N_\init\log(1/\delta_\init)/\gamma)$ samples from $\D$,  
\new{where $\gamma = \lambda_\init /100 = \Theta(\eps_0^2)$}. 
Therefore, in total, we need\begin{align*}
    N &\gtrsim \text{ number of $\sat$ runs}\times N_\sat + \text{number of $\mathsf{Refine}$ runs }\times N_{\mathsf{Refine}}\\
    &\gtrsim TkL_\init^3\times \poly(k/\eps_0) TkL_\init^3/\delta + Tk\times \poly(1/\eps_0)L_\init^3 N_\init \log(1/\delta_\init) \\
    &\gtrsim \log^2(1/\eps_0)\poly(k/\eps_0,1/\delta)L_\init^3(L_\init^3 + N_\init) \end{align*}
samples. \new{By \Cref{lem:add-vectors}, $\mathsf{Refine}$ has runtime $\poly(d,1/\eps_0, L_\init, N_\init,T_\init)$ and $\mathsf{Test}$ has runtime $O(dN_\sat)$. Thus, in summary, \nameref{alg:cover-and-prune} has runtime $\poly(d,k/\eps_0,1/\delta,L_\init, N_\init,T_\init) = \poly(d,N,T_\init)$. This completes the proof of \Cref{prop:cover-and-prune}.} \end{proof}

\subsection{Proof of \Cref{prop:optimization}}\label{subsec:proof-for-optimization}

We are now ready to prove the main result of this section.

\begin{proof}[Proof of \Cref{prop:optimization}]
Let $\eps_0=\sqrt\eps/(4k).$
Run \nameref{alg:cover-and-prune} with failure probability $1/9$. 
By \Cref{prop:cover-and-prune}, with
probability at least $8/9$, it returns
$\widehat \cV=\{\bv_1,\ldots,\bv_m\}$, where $m\leq k$, such that for every
$i$ with $\beta_i\geq\eps_0$ there exists $\vec v\in\widehat \cV$ such that
$\theta_{\pm}(\w_i,\vec v)\leq\eps_0.$
We condition on this event. Assume without loss of generality that all the weights 
$\beta_i\geq \eps_0$, as otherwise, \new{ignoring these ReLUs will only increase the $L_2^2$ error by $\eps$. 
Indeed, we have the following sequence of inequalities:
\begin{align*}
    \Exy[(y - \wh{y}(\x))^2] &= \E_{\x\sim\calN(\vec 0,\mI_d)}\bigg[\bigg(\sum_{i:\beta_i\geq \eps_0}\beta_i\sigma(\w_i\cdot\x) - \wh{y}(\x) + \sum_{i:\beta_i\leq \eps_0}\beta_i\sigma(\w_i\cdot\x)\bigg)^2\bigg]\\
    &\leq 2\E_{\x\sim\calN(\vec 0,\mI_d)}\bigg[\bigg(\sum_{i:\beta_i\geq \eps_0}\beta_i\sigma(\w_i\cdot\x) - \wh{y}(\x)\bigg)^2\bigg] + 2\E\bigg[\bigg(\sum_{i:\beta_i\leq \eps_0}\beta_i\sigma(\w_i\cdot\x)\bigg)^2\bigg]\\
    &\leq 2\E_{\x\sim\calN(\vec 0,\mI_d)}[(y'(\x) - \wh{y}(\x))^2] + 2k\sum_{i:\beta_i\leq \eps_0}^k \beta_i^2\E[(\sigma(\w_i\cdot\x))^2]\\
    &\leq 2\E_{\x\sim\calN(\vec 0,\mI_d)}[(y'(\x) - \wh{y}(\x))^2] + \eps \;,
\end{align*}
where $y'(\bx) = \sum_{i:\beta_i\geq \eps_0}\beta_i\sigma(\w_i\cdot\x).$ 
}

Observe that, by standard concentration arguments, with $N_{\vec g} = O(d\poly(k)/\eps^2)$ samples we can get an estimate $\widehat{\vec g}$ of the
vector $\vec g=\E[y\x]$ up to \new{$\ell_2$} error $\eps/16$ with probability at least $8/9$. Since
$\E[\sigma(\w\cdot\x)\x]=\w/2$ for every unit vector $\w$, we have that 
\[
    y-\vec g\cdot\x=\frac12\sum_{i=1}^k\beta_i|\w_i\cdot\x|.
\]
Let $
    \sigma_{\rm abs}(z)=
    \sigma(z)-z/2
    =|z|/2$. Then
\begin{align*}
   \Big(\Exy [(y-\widehat{\vec g}\cdot\x-\sum_{j=1}^k\beta_j\sigma_{\rm abs}(\bv_j\cdot\x))^2]\Big)^{1/2}
    \leq
     \frac{\sqrt{\eps}}{4}+12\sum
       \beta_i\min_{s_i\in\{\pm 1\}}\|\w_i-s_i\bv_i\|_2
    \leq\frac{\sqrt{\eps}}{4}+
    \frac{k\eps_0}{4}
    \le \frac{\sqrt\eps}{2}.
\end{align*}
For $\bm{\alpha}\in [0, k]^m$ with $\sum \alpha_i\leq k$, write
\[
    \mathcal R(\bm{\alpha})
    =\Exy\bigg[\bigg( y- \wh{\vec g}\cdot\x-\sum_{j=1}^m\alpha_j\sigma_{\rm abs}(\bv_j\cdot\x)\bigg)^2\bigg]
\]
Setting 
$\vec Z=(\new{y(\x) - \wh{\vec g}\cdot \bx},\sigma_{\rm abs}(\bv_1\cdot\x),\ldots,\sigma_{\rm abs}(\bv_\new{m}\cdot\x))$, 
we have that for any coordinate of $\vec Z \vec Z^\top$,  $\E[(\vec Z \vec Z^\top)_{i,j}^2]\lesssim k^4$. 
Construct an empirical distribution $\widehat{\D}$ of $\D$ using $N_{\mathrm{quad}}=k^{10}/(\eps^2\delta)$ 
samples. \new{Then, by Chebyshev's inequality and a union bound over all indices $i,j$}, 
we have that
\[
 \pr\left(\max_{i,j}|(\E_{\widehat{\D}}[\vec Z\vec Z^\top]-\E [\vec Z\vec Z^\top])_{i,j}|>
        \frac{\eps}{(k+1)^2}\right)
 \lesssim\frac{k^6(k+1)^4}{N_{\mathrm{quad}}\eps^2}\lesssim\delta.
\]
Let $\overline{\bm{\alpha}} = (1,\alpha_1,\dots,\alpha_m)$.
Therefore, if the event inside the probability above is satisfied, we have that 
\new{
\begin{align*}
    |\widehat{\mathcal R}(\bm{\alpha})-\mathcal R(\bm{\alpha})| &= |\overline{\bm{\alpha}}^\top (\E_{\widehat{\D}}[\vec Z\vec Z^\top]-\E [\vec Z\vec Z^\top])\overline{\bm{\alpha}}|\\
    &\leq \max_{i,j}|(\E_{\widehat{\D}}[\vec Z\vec Z^\top]-\E [\vec Z\vec Z^\top])_{i,j}|\|\overline{\bm{\alpha}}\|_1^2\leq \frac{\eps}{(k+1)^2}(m+1)^2\leq \eps.
\end{align*}
}
Minimizing over $\bm{\alpha}$ is a convex program and the empirical minimizer using $N_{\mathrm{quad}} + N_{\vec g} = d\,\poly(k/\eps)$ samples satisfies \[
    \mathcal R(\wh{\bm{\alpha}})
    \lesssim
    \mathcal R(\bm{\alpha}^*)+\eps
    \lesssim\eps\;.
\]
\new{By \Cref{prop:cover-and-prune}, the  procedure \nameref{alg:cover-and-prune} uses 
$N_{\mathrm{opt}} = \widetilde{O}(\poly(k/\eps)L_\init^3(L_\init^3 + N_\init))$ samples 
(recall that we have set $\eps_0 = \sqrt{\eps}/(4k)$ and $\delta = 1/9$). 
Thus, the total number of samples required for \nameref{alg:main} is 
$N = \widetilde{O}(d\,\poly(k/\eps) + \poly(k/\eps)L_\init^3(L_\init^3 + N_\init))$. 
The runtime of the algorithm is $\poly(d, N, T_\init)$.}
This completes the proof.
\end{proof}

\newpage

\section{Efficient Coarse Direction Recovery}\label{sec:PCA}

In this section, we show that the $\init$ procedure postulated in 
\Cref{assum:initialization-algorithm} can be implemented in polynomial time. 
At a high-level, an algorithm that satisfies \Cref{assum:initialization-algorithm} 
needs to guarantee the following:  Given sample access to 
a distribution $(\x,y)\sim{\D}$ such that (i) the $\x$-marginal is 
the standard Gaussian, i.e., $\x\sim\calN(\vec 0,\mI)$, and 
(ii) $y(\x) = \sum_{i=1}^k \beta_i\sigma(\w_i\cdot\x)$, with $\sigma$ being the ReLU activation, $b_i\in(0,1]$ and $\|\w_i\|_2 = 1$, the algorithm returns a list $\mathcal L$ of unit vectors 
such that for each target weight vector $\vec w_i$, there exists 
$\vec v\in \mathcal L$ that correlates non-trivially with $\vec w_i$, 
i.e., $|\vec w_i\cdot\vec v| = 1 - c$, for a small constant $c > 0$. 
To this end, we prove that the top eigenvector of a carefully selected matrix (see Equation~\eqref{eq:target-matrix-M(u)}) 
aligns well with one of the directions $\w_i$.

Our main result in this section is the following polynomial-time implementation of the $\init$ 
subroutine. The proof of \Cref{prop:poly-time-initialization-via-implicit-tensor} is deferred to \Cref{sec:poly-time-initialization}.

\begin{restatable}[Polynomial-Time $\init$]{proposition}{ImplicitTensorPolyTimeInitialization}
\label{prop:poly-time-initialization-via-implicit-tensor}
Let $k, d\in \Z_+$ and \new{$\lambda,\delta\in (0,1)$.} 
Let $\D$ be a Gaussian sum of $k$ ReLUs.
Then, there is an algorithm (\nameref{alg:implicit-tensor-approximate-matrix}) 
that draws 
$N=\poly(d, k, 1/\lambda, 1/\delta)$ 
i.i.d.\ samples from $\D$, runs in $T = \poly(d,N)$ time, 
and outputs a list $\cL$ of unit vectors of size 
$|\cL| = \poly(k/\lambda)\log(1/\delta)$ such that
with probability at least $1 - \delta$, for any $\w_i$, $i\in[k]$, with $\beta_i\geq \lambda$, 
there exists a vector $\br_i\in\cL$ such that $|\br_i\cdot\w_i|\geq 0.99$. 
\end{restatable}

We begin this section by formally arguing that the guarantee of the $\init$ 
subroutine can be reduced to finding the top eigenvector of a carefully chosen matrix 
\eqref{eq:target-matrix-M(u)} (\Cref{sec:Low-Degree-Filters Reveal-Hidden-Directions}). 
Then, in \Cref{sec:quasi-poly}, we give a simple  estimation procedure for this matrix 
that uses a quasi-polynomial number of samples from $\D$. \Cref{sec:quasi-poly} highlights 
why more sophisticated techniques are needed for estimation and motivates the techniques 
based on implicit tensor computation that we employ in \Cref{sec:poly-time-initialization}. 
Specifically, in \Cref{sec:poly-time-initialization}, we present an algorithm (that uses 
techniques from implicit tensor computation) to estimate the matrix 
\eqref{eq:target-matrix-M(u)} with polynomially many samples 
and in polynomial time. Using this algorithm, we then 
prove \Cref{prop:poly-time-initialization-via-implicit-tensor}. \new{Throughout this section, we take without loss of generality that $\lambda$ is sufficiently smaller than one (namely, $\lambda \le \sqrt{\xi}/10$, where $\xi$ as absolute constant). This only impacts the constants suppressed by the big-Oh notation.}

\subsection{Low-Degree Filtering: Revealing Hidden Directions}\label{sec:Low-Degree-Filters Reveal-Hidden-Directions}

As discussed in \Cref{sec:tech-overview-CDR}, our implementation of $\init$ is based on the properties of Hermite tensors of ReLU functions; namely, the fact that degree-$p$ moments of the ReLU sums are proportional to rank-tensors (recall \eqref{eq:low-degree-tensor-form}). 
For technical reasons, instead of working with the tensor $\vec Q_{2p}$ introduced in \eqref{eq:low-degree-matrix-form} in \Cref{sec:tech-overview-CDR}, we use a closely related matrix that is likewise proportional to $ \sum_{i=1}^k \beta_i |\w_i\cdot\bu|^{2p}\w_i\w_i^{\top}$
\begin{align}\label{eq:target-matrix-M(u)}
    \mM_{2p+2}(\bu) &\overset{(i)}{:=} \new{\frac{1}{Z_p(\cU)}\bigg\la\frac{1}{\wh{\sigma}_{2p+2}}\E_{(\x,y)\sim\D}[y\he_{2p+2}(\x)],\, \bu^{\ot p}\ot \bar{\bu}^{\ot p} \bigg\ra} \nonumber\\
    &\overset{(ii)}{=} \frac{1}{Z_p(\cU)} \sum_{i=1}^k \beta_i |\w_i\cdot\bu|^{2p}\w_i\w_i^{\top} \;, \tag{M} \end{align}
where $Z_p(\cU )>0$ is a fixed parameter depending on the sampling distribution $\cU $ from which $\bu$ is drawn and on the degree $p$. 
\new{
The probing vector $\bu$ can be chosen from a complex distribution, 
where $\bar{\bu}$ is its complex conjugate. 
When $\bu$ is a real vector, we have $\bar{\bu} = \bu$, and in this case, 
the tensor contraction in $(i)$ becomes $\la\E[y\he_{2p+2}(\x)], \bu^{\ot 2p}\ra$.
Intuitively, the tensor $\bu^{\ot p}\ot \bar{\bu}^{\ot p}\in(\mathbb{C}^d)^{\ot 2p}$ 
can be viewed as a ``filter'' that projects the signal tensor 
$\E[y\he_{2p+2}(\x)]$ to the direction of $\bu$. 
Choosing the order $p$ and the distribution $\cU$ properly, 
we can guarantee that with large probability, only the target vector $\w_i$, $i\in[k]$ 
will come through the filter---as discussed in \Cref{sec:tech-overview-CDR}.
To see why $(ii)$ is true, recall that by the definition of Hermite tensors 
and the definition of our learning problem (\Cref{def:gaussian-pac-learning}), 
we have $\E[y\he_{2p+2}(\x)] = \sum_{i=1}^k \wh{\sigma}_{2p+2}\beta_i\w_i^{\ot 2p+2}$,
which yields $(ii)$ after contraction with the filter.
The matrix \ref{eq:target-matrix-M(u)} can be estimated via samples from $\D$. 
Naively, one can estimate the population tensor $\E[y\he_{2p+2}(\x)]$ using $d^{O(p)}$ samples, 
then perform contraction with the filter. 
However, since the target matrix only concerns the projection to the probing vector $\bu$, 
estimating the tensor $\E[y\he_{2p+2}(\x)]$ accurately in all the directions is unnecessary.
In \Cref{sec:quasi-poly}, we show that \ref{eq:target-matrix-M(u)} can be estimated approximately 
via a polynomial; then in \Cref{sec:poly-time-initialization} we give an algorithm 
that estimates the contraction 
$\la\E_{(\x,y)\sim\D}[y\he_{2p+2}(\x)],\, \bu^{\ot p}\ot \bar{\bu}^{\ot p} \ra$ 
as a whole using polynomially many samples, 
avoiding a separate estimation of the tensor $\E[y\he_{2p+2}(\x)]$. 
}

{In the following, we show that the $\init$ subroutine can be reduced 
to estimating the top eigenvector of the empirical estimator 
of the matrix \ref{eq:target-matrix-M(u)}.}

\begin{lemma}\label{prop:low-degree-separates-direction}
Let $\xi \in (0, 1/300]$ be a small absolute constant.
Let $\w_i$ be unit vectors in $\R^d$, $\beta_i\in[0,1]$, $i\in[k]$, and  
$\lambda\in(0, \sqrt{\xi}/10]$ be less than a sufficiently small constant. 
Let $p = \lceil\log(k/\lambda)/\log(1+\xi)\rceil$.
Consider any $i^*\in[k]$ such that $\beta_{i^*}\geq \lambda$, 
and denote $\cA_{i^*}(\xi) =\{i\in[k]:|\w_i\cdot\w_{i^*}|\leq 1 - \xi\}$.
Draw and fix $\bu$ from the distribution $\cU$. 
Let $\mM_{2p+2}(\bu)$ be the matrix defined in \eqref{eq:target-matrix-M(u)}. 
Define the following ``correlation'' and ``separation'' events on $\cU$:
\begin{align}\label{eq:cor-sep-events}
    \evt_{i^*}^{\mathrm{Cor}} = \{\bu: |\w_{i^*}\cdot\bu|\geq (Z_p(\cU))^{1/(2p)}\},\;\evt_{i^*}^{\mathrm{Sep}} = \{\bu: |\w_{i^*}\cdot\bu|\geq (1  +\xi)|\w_j\cdot\bu|,\,\forall j\in \cA_{i^*}(\xi)\}.
\end{align}
Suppose $\bu\in\evt_{i^*}^{\mathrm{Cor}}\cap\evt_{i^*}^{\mathrm{Sep}}$.
Let $\wh{\mM}_{2p+2}(\bu)$ be any real matrix such that $\|\wh{\mM}_{2p+2}(\bu) - \mM_{2p+2}(\bu)\|_2\leq (\sqrt{\xi}/3)\|\mM_{2p+2}(\bu)\|_2$. Then, the top eigenvector $\wh{\br}_1$ of $\wh{\mM}_{2p+2}(\bu)$ satisfies $|\wh{\br}_1\cdot\w_{i^*}|\geq 0.99$.
\end{lemma}
\begin{proof}
The high-level idea of the proof is that under the assumptions of \Cref{prop:low-degree-separates-direction}, 
the target matrix $\mM_{2p+2}(\bu)$ is approximately rank-$1$ and 
proportional to $\w_{i^*}\w_{i^*}^\top$. 
Let us define
\begin{align}\label{eq:weights-W_A-W_B}
    W_A = \frac{1}{Z_p(\cU)}\sum_{i\in \cA_{i^*}(\xi)} \beta_i|\w_i\cdot\bu|^{2p},\; W_B = \frac{1}{Z_p(\cU)}\sum_{i\notin \cA_{i^*}(\xi)} \beta_i|\w_i\cdot\bu|^{2p},\;\text{and }\mM^B_{2p+2}(\bu) = W_B \w_{i^*}\w_{i^*}^\top.
\end{align}
Observe that {$\mathrm{tr}(\mM_{2p+2}(\bu)) = W_A + W_B$.}
We establish the following: 
\begin{claim}\label{claim:close-to-rank-1-matrix}
Under the assumptions of \Cref{prop:low-degree-separates-direction}, it holds:
\begin{align}\label{eq:||M-M_b||_2}
      W_A\leq (\lambda/k)W_B\,\text{and }\,\|\mM_{2p+2}(\bu) - \mM^B_{2p+2}(\bu)\|_2\leq (\lambda/k + \sqrt{2\xi})\|\mM^B_{2p+2}(\bu)\|_2.
  \end{align}
\end{claim}
\begin{proof}[Proof of~\Cref{claim:close-to-rank-1-matrix}]
Since $i^*\notin \cA_{i^*}(\xi)$, $\beta_{i^*}\geq \lambda$, 
and $|\w_{i^*}\cdot\bu|^{2p}\geq Z_p(\cU)$ for $\bu\in\evt_{i^*}^{\mathrm{Cor}}$, 
the weight $W_B$ is large, namely 
$W_B\geq  \beta_{i^*}|\w_{i^*}\cdot\bu|^{2p}/Z_p(\cU)\geq \lambda$. 
On the other hand, for any $i\in \cA_{i^*}(\xi)$, 
the probing vector $\bu$ satisfies 
$|\w_i\cdot\bu|\leq |\w_{i^*}\cdot\bu|/(1+\xi)$ (because $\bu\in \evt_{i^*}^{\mathrm{Sep}}$). 
Hence, using the assumption $\beta_i\in(0,1]$ and our choice of $p$, 
the weight $W_A$ is small. Specifically, we have 
\begin{align*}
    W_A&\leq \frac{1}{Z_p(\cU)}\sum_{i\in \cA_{i^*}(\xi)}\frac{|\w_{i^*}\cdot\bu|^{2p}}{(1+\xi)^{2p}}\leq \frac{|\w_{i^*}\cdot\bu|^{2p}}{Z_p(\cU)} k\exp(-(2p)\log(1+\xi))\\
    &\leq \frac{\lambda^2}{k}\frac{|\w_{i^*}\cdot\bu|^{2p}}{Z_p(\cU)}\leq \frac{\lambda}{k}\frac{\beta_{i^*}|\w_{i^*}\cdot\bu|^{2p}}{Z_p(\cU)}\leq \frac{\lambda}{k} W_B.
\end{align*}
Therefore, we have \begin{align*}
&\quad \|\mM_{2p+2}(\bu) - \mM^B_{2p+2}(\bu)\|_2\\
&= \bigg\|\frac{1}{Z_p(\cU)}\sum_{i\in \cA_{i^*}(\xi)} \beta_i|\w_i\cdot\bu|^{2p}\w_i\w_i^\top + \frac{1}{Z_p(\cU)}\sum_{i\notin \cA_{i^*}(\xi)} \beta_i|\w_i\cdot\bu|^{2p}(\w_i\w_i^\top - \w_{i^*}\w_{i^*}^\top) \bigg\|_2\\
    &\leq \frac{1}{Z_p(\cU)}\sum_{i\in \cA_{i^*}(\xi)} \beta_i|\w_i\cdot\bu|^{2p}\|\w_i\w_i^\top\|_2 + \frac{1}{Z_p(\cU)}\sum_{i\notin \cA_{i^*}(\xi)} \beta_i|\w_i\cdot\bu|^{2p}\|\w_i\w_i^\top - \w_{i^*}\w_{i^*}^\top \|_2 \end{align*}
Observe that \new{$\|\bv\bv^\top - \w\w^\top\|_2\leq \sqrt{1 - (\bv\cdot\w)^2}$ for any unit vectors $\bv,\w$.
By the definition of the index set $\cA_{i^*}(\xi)$, 
for any $i\notin \cA_{i^*}(\xi)$, it holds $|\w_i\cdot\w_{i^*}|\geq 1 - \xi$, therefore we have $\|\w_i\w_i^\top - \w_{i^*}\w_{i^*}^\top \|_2\leq \sqrt{2\xi}$.}
Thus, we conclude that:
\begin{align*}
    \|\mM_{2p+2}(\bu) - \mM^B_{2p+2}(\bu)\|_2\leq W_A + \sqrt{2\xi}W_B\leq (\lambda/k + \sqrt{2\xi})W_B \;,
\end{align*}
completing the proof of \Cref{claim:close-to-rank-1-matrix}.
\end{proof}

It remains to prove that $\wh{\mM}_{2p+2}(\bu)$ is close to $\mM^B_{2p+2}(\bu)$ in spectral norm, which would imply that the top-eigenvector of $\wh{\mM}_{2p+2}(\bu)$ aligns with $\w_{i^*}$. 
\new{As $\lambda\in(0,\sqrt{\xi}/10]$, it holds that $\lambda/k\leq \sqrt{\xi}/10$.}
By the triangle inequality and \Cref{claim:close-to-rank-1-matrix}, we have
\begin{align*}
    \|\wh{\mM}_{2p+2}(\bu) - \mM^B_{2p+2}(\bu)\|_2 &\leq \|\wh{\mM}_{2p+2}(\bu) - \mM_{2p+2}(\bu)\|_2 + \|{\mM}_{2p+2}(\bu) - \mM^B_{2p+2}(\bu)\|_2\\
&\leq \bigg(\frac{\sqrt{\xi}}{10} + \sqrt{2\xi}\bigg)W_B + \frac{\sqrt{\xi}}{3}\|{\mM}_{2p+2}(\bu)\|_2.
\end{align*}
In addition, since $\|{\mM}_{2p+2}(\bu)\|_2\leq W_A + W_B$ and $W_A\leq (\lambda/k) W_B\leq (\sqrt{\xi}/10)W_B$, we have $\|{\mM}_{2p+2}(\bu)\|_2\leq (1 + (\sqrt{\xi}/10))W_B$.
Thus, we conclude that
\begin{align*}
    \|\wh{\mM}_{2p+2}(\bu) - \mM^B_{2p+2}(\bu)\|_2 &\leq \bigg(\frac{\sqrt{\xi}}{10} + \sqrt{2\xi} + \frac{\sqrt{\xi}}{3}\bigg(1 + \frac{\sqrt{\xi}}{10}\bigg)\bigg)W_B.
\end{align*}
Denote the top (unit) eigenvector of $\wh{\mM}_{2p+2}$ by $\wh{\br}_1$ 
and the top eigenvector of $\mM^B_{2p+2}(\bu)$ by $\br_1$; 
note that $\br_1 =\w_{i^*}$.
Then, using Wedin's inequality (\Cref{fct:wedin}), with $\xi\leq 1/300$, we have 
\begin{align*}
\sin(\theta(\wh{\br}_1, \br_1))&\leq 
\frac{\|\wh{\mM}_{2p+2}(\bu) - \mM^B_{2p+2}(\bu)\|_2}{\|\mM^B_{2p+2}(\bu)\|_2 - \|\wh{\mM}_{2p+2}(\bu) - \mM^B_{2p+2}(\bu)\|_2} \\
&\leq \frac{(\frac{\sqrt{\xi}}{10} + \sqrt{2\xi} + \frac{\sqrt{\xi}}{3}(1 + \frac{\sqrt{\xi}}{10}))W_B}{W_B - (\frac{\sqrt{\xi}}{10} + \sqrt{2\xi} + \frac{\sqrt{\xi}}{3}(1 + \frac{\sqrt{\xi}}{10}))W_B}
\leq\sqrt{1 - 0.99^2} \;.
\end{align*}
Thus, $|\wh{\br}_1\cdot\w_{i^*}|\geq 0.99$, 
completing the proof of \Cref{prop:low-degree-separates-direction}.
\end{proof}

What remains are the following steps: (1) pick a suitable distribution of $\bu$ such that the event $\bu\in\evt_{i^*}^{\mathrm{Cor}}\cap\evt_{i^*}^{\mathrm{Sep}}$ happens with probability at least $\poly(\lambda/k)$; (2) approximate the matrix $\mM_{2p+2}(\bu)$ efficiently using samples from $(\x,y)\sim\D$. In \Cref{sec:quasi-poly}, we choose $\cU$ to be the uniform distribution on the sphere, obtaining quasi-polynomial sample complexity; then in \Cref{sec:poly-time-initialization}, we choose $\cU$ to be a complex distribution, and by applying implicit tensor computation techniques, we argue that  $\mM_{2p+2}$ can be estimated efficiently, using $\poly(k/\lambda)$ samples.

\subsection{Warm-up: Quasi-Polynomial Time $\init$ Algorithm}\label{sec:quasi-poly}

Given a distribution $(\x,y)\sim \D$ such that 
$y(\x) = \sum_{i=1}^k \beta_i\sigma(\w_i\cdot\x)$ and $\x \sin \calN(\vzero, \mI_d)$, a simple choice of the probing vector is to 
sample from the \new{uniform distribution on the} unit sphere in a subspace containing $\spanned(\w_1,\dots,\w_k)$. The latter can be estimated efficiently via the following standard result. 
\begin{fact}[see, e.g.,~Lemma 3.4 of~\cite{DKKZ20}]\label{fact:subspace-approximate}
    Let $(\x,y)\sim\D$ be the distribution on $\R^d \times \R$ such that 
    $\x\sim\calN(\vec 0,\mI_d)$ and $y(\x) = \sum_{i=1}^k 
    \beta_i\sigma(\w_i\cdot\x)$, $\beta_i\in \new{[} 0,1]$, 
    $\|\w_i\|_2 = 1$, and $\sigma$ is the ReLU activation. Let $\wh{\sigma}_2$be the second-
    order Hermite coefficient of $\sigma$. Let $\eta > 0$.
    Then, using $N = \poly(d,k,1/\eta)$ samples 
    \new{$(\x^{(i)}, y^{(i)})$ drawn} from $\D$, it holds \new{with high probability} that: \begin{align*}
        \bigg\|\frac{1}{\sqrt{2}N}\sum_{i=1}^N y^{(i)}(\x^{(i)}\x^{(i)\top} - \mI_d) - \sum_{\new{j}=1}^k |\wh{\sigma}_2|  \beta_{\new{j}}\w_{\new{j}}\w_{\new{j}}^{\top}\bigg\|_2\leq\eta \;.  
    \end{align*}
\end{fact}
\new{
Let $\wh{\sW}$ be the subspace spanned by the top-$k$ vectors of the empirical matrix constructed in \Cref{fact:subspace-approximate}. Then,
\Cref{fact:subspace-approximate} implies that for any $\w_i$, $i\in[k]$, it holds $\|\w_i - \mP_{\wh{\sW}}\w_i\|_2 \leq\sqrt{2\eta/\beta_i}$ (with high probability, otherwise, ascribe to the failure event of the algorithm). For any $\lambda$-heavy weights $\w_{i^*}$ with coefficients $\beta_{i^*}\geq \lambda$, choosing $\eta = \poly(\lambda/k)$ yields that
$$\mathrm{dist}(\wh{\sW}, \spanned(\{\w_{i^*}: \beta_{i^*}\geq \lambda\}))\leq \lambda. $$ 
Therefore, for a probing vector $\bu$ sampled uniformly from $\wh{\sW}\cap\spd$, the probability that it satisfies the correlation event  $\evt_{i^*}^{\mathrm{Cor}}$ and the separation $\evt_{i^*}^{\mathrm{Sep}}$ is almost the same as if sampled from ${\sW}\cap\spd$, with only differences in the absolute constants.
Thus, to simplify the analysis, in the remainder of this section, we assume access to the 
subspace
$\sW=\spanned(\w_1,\dots,\w_k)$ with dimension $m$ and we sample the probing vectors $\bu$ from $\cU = \mathrm{Unif}(\sW\cap\spd)$.}

\begin{remark}[Range of parameters.]\label{rmk:range-parameters}
{\em Throughout this section, we assume that 
the dimension $m$ of the subspace $\sW = \spanned(\w_1,\dots,\w_k)$ satisfies $m\gtrsim \log(k)$. 
Indeed, the $\init$ procedure
for small $m$ admits a simpler approach.
More generally, consider the regime
$m\lesssim \log(k)$.
In this case, we can obtain the desired $\init$
guarantee by sampling random directions from the
$m$-dimensional subspace $\sW$.
For any fixed unit vector in $\sW$, a uniformly random
unit vector has correlation at least $0.99$ with probability
$2^{-O(m)}$. By taking $q=2^{O(m)}\log(k/\delta)\leq\poly(k,1/\eps)\log(1/\delta)$ random vectors $\vec u_1,\ldots,\vec u_q$, with probability at least $1-\delta$, we have that $\max_{j\in[q]}|\vec u_j \cdot \w_i|\geq 0.99$ for all $i\in[k]$.}
\end{remark}

The pseudocode of our algorithm (\nameref{alg:initialization-via-spectral}) 
is given below.

\begin{algorithm}[H]
\caption{$\mathsf{CDRviaSpectral}(\D,\lambda,\delta)$}
\label{alg:initialization-via-spectral}
\begin{algorithmic}[1]
\State\textbf{Input:} A subspace $\sW$, access to a distribution of a sum of \new{$k$} ReLUs with Gaussian marginals $(\x,y)\sim\D$,
positive parameters $\lambda,\delta\in(0,1)$.
\State \textbf{Output:} A list of vectors satisfying \Cref{assum:initialization-algorithm}.

\medskip

\State Fix an absolute constant $\xi\in(0,1/300]$, 
let $p \gets \lceil\log(k/\lambda)/\log(1+\xi)\rceil$\new{, $\lambda \gets \min\{\lambda, \sqrt{\xi}/10\}$}.

\State Let $\wh{\sigma}_{2p+2}$ be the $(2p+2)^{\text{th}}$ Hermite coefficient of ReLU and $\wh{\sigma}''_{2p}\gets \wh{\sigma}_{2p+2}\sqrt{(2p+1)(2p+2)}$. 

\State $Z_p(\cU) \gets k^{-p}$, $\cL\gets\emptyset$. \State Draw a set $\wh{\cU}$ of $\poly(k)\log(\new{1}/\delta)$ probing vectors from $\cU = \mathrm{Unif}(\spd\cap\sW)$.
\For{$\bu\in\wh{\cU}$}
\State Draw $N = \new{d^2\,\poly(p)k^{O(\log(k/\lambda))}}$ samples from $\D$ and construct empirical distribution $\wh{\D}$. \State $\wh{\mM}_{2p+2}(\bu) \gets (1/(\wh{\sigma}''_{2p} Z_p(\cU)))\E_{(\x,y)\sim\wh{\D}}[y\hep_{2p}(\bu\cdot\x)(\x^{\perp\bu}\x^{\perp\bu\top} - \mI_{\perp\bu})]$. \label{line:emprical-matrix-quasi} \State Find the top eigenvector of $\br_1 = \br_1(\wh{\mM}_{2p+2}(\bu))$, $\cL\gets\cL\cup\{\pm \br_1\}$.
\EndFor
\State\textbf{Return: }$\cL$.
\end{algorithmic}
\end{algorithm}

Our main result for this subsection is the following:

\begin{restatable}[Quasi-Polynomial-Time Implementation of $\init$]{proposition}{QuasiPolyTimeInitialization}\label{prop:quasi-poly-time-initialization-via-spectral}
Let $\delta\in(0,1)$, \new{$\xi \in (0,1/300]$ be an absolute constant and let $\lambda\in(0,\sqrt{\xi}/10]$}. 
Let $(\x,y)\sim \D$ be a distribution such that 
$y(\x) = \sum_{i=1}^k \beta_i\sigma(\w_i\cdot\x)$ where $\sigma$ is the ReLU activation, $\beta_i\in[0,1]$, 
$\|\w_i\|_2 = 1$, and $\x \sim \calN(\vzero, \mI_d)$. 
Let  $p \geq \lceil \log(k/\lambda)/\log(1+\xi)\rceil$. Then, \nameref{alg:initialization-via-spectral}  
draws $N= d^2\,k^{O(\log(k/\lambda))}\log(1/\delta)$ samples from $\D$, runs in $T = \poly(d,N)$ time, 
and outputs a list $\cL$ of unit vectors of size 
$|\cL| = \poly(k)\log(1/\delta)$ that with probability at least $1 - \delta$ 
satisfies the following guarantee: 
for any $\w_i$, $i\in[k]$ with $\beta_i\geq \lambda$, 
there exists $\br_i\in\cL$ such that $|\br_i\cdot\w_i|\geq 0.99$. 
\end{restatable}

To prove \Cref{prop:quasi-poly-time-initialization-via-spectral}, 
we need to bound from below the probability of the event $\evt_{i^*}^{\mathrm{Cor}}\cap\evt_{i^*}^{\mathrm{Sep}}$ for a fixed index $i^*\in[k]$, defined in \eqref{eq:cor-sep-events}.
\new{
Intuitively, when the probing vector $\bu$ is sampled uniformly from the sphere, with high probability the values $|\bu\cdot\w_{i^*}|$ and $|\bu\cdot\w_j|$ would be separated when the vectors $\w_{i^*}$ and $\w_j$ are away from each other. This is the separation event $\evt_{i^*}^{\mathrm{Sep}}$.
At the same time, for the correlation event, $\evt_{i^*}^{\mathrm{Cor}}$, with high probability, $|\bu\cdot\w_{i^*}|$ cannot be too small and would be concentrated near $1/\sqrt{k}$.
Therefore, the probability that $\evt_{i^*}^{\mathrm{Cor}}\cap\evt_{i^*}^{\mathrm{Sep}}$ holds is bounded from below.}
We then determine the number of samples required to estimate the target matrix \ref{eq:target-matrix-M(u)}.
\new{Finally, we prove \Cref{prop:quasi-poly-time-initialization-via-spectral} by invoking \Cref{prop:low-degree-separates-direction}, as outlined previously.}
The proof of \Cref{prop:quasi-poly-time-initialization-via-spectral} is deferred to the end of this subsection.

\paragraph{Determining the Probability of the Probing Event.} 
We first determine the probability of event $\evt_{i^*}^{\mathrm{Cor}}\cap\evt_{i^*}^{\mathrm{Sep}}$ under the distribution $\cU = \mathrm{Unif}(\sW\cap\spd)$. \new{In fact, we prove a slightly stronger result: with probability at least $\poly(1/k)$ in addition to $\bu\in\evt_{i^*}^{\mathrm{Cor}}\cap\evt_{i^*}^{\mathrm{Sep}}$, $\bu$ also satisfies a ``near-orthogonal'' event $\evt^{\mathrm{Ort}}$. The ``near-orthogonal'' event $\evt^{\mathrm{Ort}}$ states that $\bu$ satisfies $|\bu\cdot\w_i|\leq c$ for a sufficiently small absolute constant $c$ for all hidden weight vectors $\w_i$.
This event is needed in estimating the target matrix $\mM_{2p+2}$ to high accuracy, as will become apparent in the proof of \Cref{lem:quasi-poly-sample-spectral}.}

\begin{lemma}\label{prop:sampling}
\new{Let parameters $p, Z_p(\cU)$ and constant $\xi$ be as prescribed in \Cref{prop:quasi-poly-time-initialization-via-spectral}.}
Let $\w_1,\dots,\w_k$ be $k$ unit vectors and $\sW = \spanned\{\w_1,\dots,\w_k\}\subset\R^d$ be a linear subspace of dimension \new{$m\leq k$}.
Let $\bu \sim\cU$ where $\cU = \mathrm{Unif}(\sW\cap\spd)$. 
Fix (any) index $i^*\in[k]$ and let $\cA_{i^*}(\xi) =\{i\in[k]:|\w_i\cdot\w_{i^*}|\leq 1 - \xi\}$. Define  the ``correlation'' and ``separation'' events on $\cU$ as in \Cref{prop:low-degree-separates-direction}:
\begin{align*}
    \evt_{i^*}^{\mathrm{Cor}} = \{\bu: |\w_{i^*}\cdot\bu|\geq (Z_p(\cU))^{1/(2p)}\},\;\evt_{i^*}^{\mathrm{Sep}} = \{\bu: |\w_{i^*}\cdot\bu|\geq (1  +\xi)|\w_j\cdot\bu|,\,\forall j\in \cA_{i^*}(\xi)\}.
\end{align*}
\new{In addition, let the ``near-orthogonal'' event on $\cU$ be:
\begin{align*}
    \evt^{\mathrm{Ort}} = \{\bu:|\bu\cdot\w_{i}|\leq \sqrt{\xi}/20,\,\forall i\in[k]\}. \end{align*}
Then, $\pr_{\bu\sim\cU}[\evt_{i^*}^{\mathrm{Cor}}\cap\evt_{i^*}^{\mathrm{Sep}}\cap\evt^{\mathrm{Ort}}]\geq \poly(1/k)$.}
\end{lemma}

\begin{proof}
Let $\bp\sim\calN(\vec 0,\mI_{\sW})$ and $Z_j \eqdef \bp\cdot\bw_j$. 
Then, by definition, $\bu = \bp/\|\bp\|_2$ and $\bu\cdot\w_i = \bp\cdot\w_i/\|\bp\|_2$. 
The roadmap of the proof is the following: we first show that $\evt_{i^*}^{\mathrm{Sep}}\cap\{t\leq |\bp\cdot\w_{i^*}|\leq 2t\}$ happens with probability at least $\poly(1/k)$ \new{for some $t\lesssim \sqrt{\log(k)}$} then we argue that $\pr[\|\bp\|_2\in(\sqrt{m/2},\sqrt{m})]\geq1 - \poly(1/m)$. \new{Since we have assumed without loss of generality (see \Cref{rmk:range-parameters}) that the dimension $m$ of the subspace $\sW$ satisfies $m\gtrsim \log(k)$,} applying a union bound on these events gives a lower bound on $\pr[\evt_{i^*}^{\mathrm{Sep}}\cap\{t/\sqrt{m}\leq |\bu\cdot\w_{i^*}|\leq 2t/\sqrt{m/2}\}]= \pr[\evt_{i^*}^{\mathrm{Cor}}\cap\evt_{i^*}^{\mathrm{Sep}}\cap\evt^{\mathrm{Ort}}]$,
completing the proof.

For each $j\in \cA_{i^*}(\xi)$, let $\rho_{{i^*},j} = \w_{i^*}\cdot\w_j = \cos(\theta(\w_{i^*},\w_j))$. By the definition of $\cA_{i^*}(\xi)$, it holds $|\rho_{{i^*},j}|\le \rho \eqdef 1 - \xi$. 
Then, we can write $Z_j$ as 
$$Z_j = \bp \cdot \cos(\theta(\w_{i^*},\w_j)) \w_{i^*} + \bp \cdot \sin(\theta(\w_{i^*},\w_j))\wh{\w_{i^*}^{\perp\w_j}} = \rho_{{i^*},j} Z_{i^*} + \sqrt{1-\rho_{{i^*},j}^2} G_j\;,$$
where $G_j$ is a standard normal random variable.
By the definition of $\bu$, it holds $\bu\cdot\w_i = \bp\cdot\w_i/\|\bp\|_2$, and hence the event $|\bu\cdot\w_j|\leq |\bu\cdot\w_{i^*}|/(1+\xi)$ is equivalent to $|Z_j|\leq |Z_{i^*}|/(1+\xi)$.
Conditional on $Z_{i^*}=z$,
consider the event that $\evt_{i^*}^{\mathrm{Sep}}$ is violated, i.e., there exists some $Z_j$, $j\in \cA_{i^*}(\xi)$, such that $|Z_j|\geq |Z_{i^*}|/(1+\xi)$. Note that $|Z_j|\leq |\rho_{{i^*},j}||z|+ |G_j|$ and using that $|\rho_{{i^*},j}|\leq 1-\xi$ and rearranging, we get that on a violation, $ |G_j|\geq |z|\xi^2/(1+\xi)$. Therefore, from a union bound, we have
\begin{align*}
    \pr[\cup_{j\in \cA_{i^*}(\xi)}\{ |Z_j|\geq|Z_{i^*}|/(1+\xi) \} \mid Z_{i^*}=z]&\leq \sum_{j\in \cA_{i^*}(\xi)}\pr[|G_j|\geq|z|\xi^2/(1+\xi)  \mid Z_{i^*}=z]
    \\&\leq k \max_{j\in \cA_{i^*}(\xi)}\pr[|G_j|\geq|z|\xi^2/(1+\xi)  \mid Z_{i^*}=z]\;.
\end{align*}
Using that $\pr[|G_j|\geq q]\leq 2\exp(-q^2/2)$, we get that
\[
 \pr[\cup_{j\in \cA_{i^*}(\xi)}\{ |Z_j|\geq|Z_{i^*}|/(1+\xi) \} \mid Z_{i^*}=z]\leq 2k \exp(-z^2\xi^4/(2(1+\xi)^2) )\;.
\] 
Let $C_{\xi}= 4(1+\xi)^2/\xi^4$.
We get that for any $z\geq t\eqdef \sqrt{C_\xi \log(2k)}
$, it holds \[\pr[\cup_{j\in \cA_{i^*}(\xi)}\{ |Z_j|\geq|Z_{i^*}|/(1+\xi) \} \mid  |Z_{i^*}|=z]\leq 1/2\;.\]
Therefore, the probability of no violation conditioned on $|Z_{i^*}| = Z\geq t$ is at least
\begin{align*}
    \pr[\cap_{j\in \cA_{i^*}(\xi)}\{ |Z_j|\leq|Z_{i^*}|/(1+\xi) \} \mid  |Z_{i^*}|=z]\geq 1/2.
\end{align*}
Furthermore, by the anti-concentration inequality for standard Gaussians, we have:
\[
\pr[2t\geq |Z_{i^*}|\ge t]\geq \frac{1}{\sqrt{2\pi}t}e^{-t^2/2}\geq e^{-C_\xi \log(2k)}=(2k)^{-C_\xi}.
\]
Therefore, we get that ($t = \sqrt{C_\xi \log(2k)}$):
\begin{align}\label{eq:lower-bound-prob}
\pr[\evt_{i^*}^{\mathrm{Sep}}\cap \{2t\geq|Z_{i^*}|\geq t\}]&=\pr[\cap_{j\in \cA_{i^*}(\xi)}\{ |Z_j|\leq|Z_{i^*}|/(1+\xi) \}\cap \{2t\geq|Z_{i^*}|\geq t\}] \nonumber\\
&\geq \pr[\cap_{j\in \cA_{i^*}(\xi)}\{ |Z_j|\leq|Z_{i^*}|/(1+\xi) \} \mid t\leq |Z_{i^*}|\leq 2t]\pr[2t\geq |Z_{i^*}| \geq t]\nonumber\\
&\geq (1/2)(2k)^{-C_\xi}\;.
\end{align}
Next observe that the random variable $\|\vec p\|_2^2$ is a sum of $m$ squares of independent Gaussians (as $\sW$ is an $m$-dimensional subspace).  
Therefore, $\|\vec p\|_2^2$ follows the $\chi_m^2$ distribution. 
From the concentration of $\chi_m^2$, we have that 
\begin{equation*}
    \pr[\|\vec p\|_2\not\in(\sqrt{m/2},\sqrt{2m})]\leq 2\exp(- m/16)\;.
\end{equation*}
It is without loss of generality that $m$ is greater than a large absolute constant (\Cref{rmk:range-parameters}), and we have $m\geq 17 C_\xi\log (2k)$.
Thus, we obtain 
\begin{equation}\label{eq:lower-bound-2}
    \pr[\|\vec p\|_2\not\in(\sqrt{m/2},\sqrt{2m})]\leq (1/8)(2k)^{-C_\xi}\;.
\end{equation}

\new{Finally, since $Z_j$, $j\in[k]$, are standard Gaussian random variables, by the Gaussian concentration inequality and a union bound, we obtain:
\begin{align}\label{eq:lower-bound-3}
    \pr \left[\max_{j\in[k]} |Z_j|\geq 4t \right] 
    \leq 2k\exp(-8t^2)= 2k\exp(-8C_\xi\log(2k)) = (2k)^{-8C_\xi + 1}\leq (1/8)(2k)^{-C_\xi}.
\end{align}
}
Combining \eqref{eq:lower-bound-prob}, \eqref{eq:lower-bound-2} and \eqref{eq:lower-bound-3}, since $\pr[A\cap B \cap C]\geq 1 - \pr[A^c] - \pr[B^c] - \pr[C^c]$,  with 
\begin{align*}
    A = \evt_{i^*}^{\mathrm{Sep}}\cap \{2t\geq|Z_{i^*}|\geq t\}, \; B = \|\vec p\|_2\in(\sqrt{m/2},\sqrt{2m}),\; C = \max_{j\in[k]} |Z_j|\geq 4t \;,
\end{align*}
we get that
\[
\pr[\evt_{i^*}^{\mathrm{Sep}}\cap \{2t\geq|Z_{i^*}|\geq t\} \cap\{\|\vec p\|_2\in(\sqrt{m/2},\sqrt{2m})\} \cap\{\max_{j\in[k]} |Z_j|\leq 4t\}]\geq  (1/8)(2k)^{-C_\xi}\;.
\]
Therefore, plugging in the value $t = \sqrt{C_\xi\log(2k)}$, on this event, we have 
\[
 |\vec u\cdot \w_{i^*}|\geq \sqrt{\frac{C_\xi\log(2k)}{2m}} \geq 1/\sqrt{k}\;,
\]
and since $m\geq 2\times 80^2 (C_\xi/\xi)\log(2k)$, we also have
\begin{align*}
    |\vec u\cdot \w_{j}|\leq \frac{4t}{\sqrt{m/2}}\leq \frac{4\sqrt{2C_\xi\log(2k)}}{\sqrt{m}}\leq \frac{\sqrt{\xi}}{20}, \,\forall j\in[k].
\end{align*}
Since $(Z_p(\cU))^{1/(2p)} = 1/\sqrt{k}$, we conclude that $\pr_{\bu\sim\cU}[\evt_{i^*}^{\mathrm{Cor}}\cap\evt_{i^*}^{\mathrm{Sep}}\cap\evt^{\mathrm{Ort}}]\geq \poly(1/k)$.
\end{proof}

\paragraph{Estimating the Target Matrix in Quasi-Polynomial Time.} 
We now show that, with quasi-polynomial sample and computational complexity,  
we can approximate the target matrix $\mM_{2p+2}(\bu)$ \eqref{eq:target-matrix-M(u)} 
in operator norm to high accuracy.

\begin{proposition}[Sample complexity for the spectral method]\label{lem:quasi-poly-sample-spectral}
Let $0<\xi\leq 1/300$ be an absolute constant and let $0<\lambda\leq\sqrt{\xi}/10$.
Let $(\x,y)\sim\D$ be a sum of \new{$k$} ReLUs under Gaussian marginals.
Let $p = \lceil\log(k/\lambda)/\log(1+\xi)\rceil $, $Z_p(\cU) = k^{-p}$ and $\wh{\sigma}''_{2p} = \wh{\sigma}_{2p+2}\sqrt{(2p+1)(2p+2)}$. Fix $i^*\in[k]$ and suppose the unit probing vector $\bu$ satisfies the ``correlation'', ``separation'' and ``near orthogonality'' events defined in \Cref{prop:sampling}: $\bu\in\evt_{i^*}^{\mathrm{Cor}}\cap\evt_{i^*}^{\mathrm{Sep}}\cap\evt^{\mathrm{Ort}}$.
Draw $N = d^2\,\poly(k,p,1/\lambda)k^{O(\log(k/\lambda))}$ i.i.d.\ samples 
    $(\x,y)\sim\D$ and construct the empirical distribution $\widehat{\D}$. Then, for
    \begin{align*}
   \wh{\mM}_{2p+2}(\bu)^{\perp\bu} = \frac{1}{\wh{\sigma}_{2p}''Z_p(\cU)}\E_{(\x,y)\sim \widehat{\D}}\left[y\hep_{2p}(\bu\cdot\x)(\x^{\perp\bu}\x^{\perp\bu\top} - \mI_{\perp\bu})\right]\;,
    \end{align*}
     with probability at least $2/3$, it holds $\|\wh{\mM}_{2p+2}(\bu)^{\perp\bu} - \mM_{2p+2}(\bu)\|_2\leq (\sqrt{\xi}/3)\|\mM_{2p+2}(\bu)\|_2$.
\end{proposition}

\begin{proof}

The high-level idea of the proof is that the target matrix $\mM_{2p+2}(\bu)$ 
defined in \eqref{eq:target-matrix-M(u)} can be approximately written 
as the following expectation. 

\begin{claim}\label{claim:expectation-form-of-orthogonal-to-u-target}
Given any unit vector $\bu$, define $\mM_{2p+2}(\bu)^{\perp\bu} = \mP_{\perp\bu}\mM_{2p+2}(\bu)\mP_{\perp\bu} $. Then, 
    \begin{align*}
        \mM_{2p+2}(\bu)^{\perp\bu} = \frac{1}{\wh{\sigma}_{2p}'' Z_p(\cU)}\E_{(\x,y)\sim{\D}}\bigg[y\hep_{2p}(\bu\cdot\x)\he_2^{\bu^{\perp}}(\mP_{\perp\bu}\x)\bigg].
    \end{align*}
\end{claim}
\begin{proof}[Proof of \Cref{claim:expectation-form-of-orthogonal-to-u-target}]
 By our definition, 
the distribution $(\x,y)\sim{\D}$ satisfies $y(\x) = \sum_{i=1}^k \beta_i\sigma(\w_i\cdot\x)$ 
and $\x\sim\calN(0,\mI_d)$. 
Therefore, denoting $\theta_{i\bu}:= \theta(\w_i, \bu),$ we have that for any unit vector $\bu\in\R^d$,
\begin{align}
    &\quad \E_{(\x,y)\sim{\D}}\Big[y\hep_{2p}(\bu\cdot\x)\he_2^{\bu^{\perp}}(\mP_{\perp\bu}\x)\Big] \nonumber\\
    &= \sum_{i=1}^k \beta_i\E_{\bu\cdot\x}\bigg[\E_{\x^{\perp\bu}\sim\calN(\vec 0,\mI_{\bu^{\perp}})}\Big[\sigma(\cos(\theta_{i\bu})\wh{\w_i^{\perp\bu}}\cdot\x^{\perp\bu})\he_2^{\bu^{\perp}}(\x^{\perp\bu}) \,\Big|\, \bu\cdot\x\Big] \, \hep_{2p}(\bu\cdot\x) \bigg] \nonumber\\
    &\overset{(i)}{=} \sum_{i=1}^k \beta_i\E_{\x\sim\calN(0,\mI_{d})}\Big[\sigma''(\cos(\theta_{i\bu})\bu\cdot\x + \sin(\theta_{i\bu})\wh{\w_i^{\perp\bu}}\cdot\x)\hep_{2p}(\bu\cdot\x)\Big]\sin^2(\theta_{i\bu})\wh{\w_i^{\perp\bu}}\wh{\w_i^{\perp\bu}}^{\top} \nonumber\\
    &\overset{(ii)}{=}\sum_{i=1}^k \beta_i\E_{\x\sim\calN(0,\mI_d)}\Big[\mathrm{T}_{\cos(\theta_{i\bu})}(\sigma'')(\bu\cdot\x)\hep_{2p}(\bu\cdot\x)\Big]\w_i^{\perp\bu}\w_i^{\perp\bu\top}, \label{eq:sample-compelxity-quasi-poly-1}
\end{align}
where in $(i)$ we used \Cref{fct:hermite-tensor}, by which $\E_{\bz\sim\calN}[f(\br\cdot\bz)\he_2(\bz)] =\E_{\bz\sim\calN}[f''(\br\cdot\bz)]\br\br^\top$ for any unit vector $\br$, and in $(ii)$ $\mathrm{T}_{\rho}(\sigma'')$, $\rho\in[0,1]$, 
is the smoothing operator defined in \Cref{fact:hermite-poly}.
By \Cref{fact:hermite-poly}, Part 2, we know that $\sigma''$ has Hermite expansion 
$\sigma''(z) = \sum_{i\geq 0}c_i'' \hep_i(z)$, 
where $\wh{\sigma}_i'' = \wh{\sigma}_{i+2}\sqrt{(i+1)(i+2)}$.
Thus, applying \Cref{fact:hermite-poly}, Part 3 to \eqref{eq:sample-compelxity-quasi-poly-1} 
leads to:
\begin{align*}
\frac{1}{\wh{\sigma}_{2p}''Z_p(\cU)}\E_{(\x,y)\sim{\D}}\Big[y\hep_{2p}(\bu\cdot\x)\he_2^{\bu^{\perp}}(\mP_{\perp\bu}\x)\Big] 
&=\frac{1}{Z_p(\cU)}\sum_{i=1}^k \beta_i|\w_i\cdot\bu|^{2p} \w_i^{\perp\bu}\w_i^{\perp\bu\top}\\
&= \mP_{\perp\bu}\mM_{2p+2}(\bu)\mP_{\perp\bu} =: \mM_{2p+2}(\bu)^{\perp\bu} \;,
\end{align*}
completing the proof of \Cref{claim:expectation-form-of-orthogonal-to-u-target}.
\end{proof}

We next show that when $\bu$ is a unit vector such that the event \new{$\bu\in\evt_{i^*}^{\mathrm{Cor}}\cap\evt_{i^*}^{\mathrm{Sep}}\cap\evt^{\mathrm{Ort}}$ holds} for $\lambda$-heavy hidden weight $\w_{i^*}$, then the spectral difference 
$\|\mM_{2p+2}(\bu)^{\perp\bu} - \mM_{2p+2}(\bu)\|_2$ is small.

\begin{claim}\label{claim:orthogonal-projected-matrix-is-close-to-original}
    Under the assumptions of \Cref{lem:quasi-poly-sample-spectral},
it holds
    \begin{equation*}
        \|\mM_{2p+2}(\bu)^{\perp\bu} - \mM_{2p+2}(\bu)\|_2 
        \leq \new{\frac{\sqrt{\xi}}{8}}\|\mM_{2p+2}(\bu)\|_2 \;\text{ and }\; \|\mM_{2p+2}(\bu)\|_2\geq \lambda \;.
    \end{equation*}
\end{claim}
\begin{proof}[Proof of \Cref{claim:orthogonal-projected-matrix-is-close-to-original}]
\new{
Note that for any unit vector $\w_i$, by the triangle inequality, it holds
\begin{align*}
    \|\w_i\w_i^\top - \w_i^{\perp\bu}\w_i^{\perp\bu\top}\|_2\leq \|(\w_i - \w_i^{\perp\bu})\w_i^\top\|_2 + \|\w_i^{\perp\bu}(\w_i - \w_i^{\perp\bu})\|_2\leq 2\|\w_i - \w_i^{\perp\bu}\|_2 = 2|\bu\cdot\w_i|.
\end{align*}
Since $\bu$ also satisfies the near-orthogonal event $\evt^{\mathrm{Ort}}$, i.e, for all $i\in[k]$ it holds $|\bu\cdot\w_i|\leq \sqrt{\xi}/20$, we have $\|\w_i\w_i^\top - \w_i^{\perp\bu}\w_i^{\perp\bu\top}\|_2\leq \sqrt{\xi}/10$. This implies the following inequality
\begin{align*}
    \|\mM_{2p+2}(\bu)^{\perp\bu} - \mM_{2p+2}(\bu)\|_2\leq \frac{1}{Z_p(\cU)}\sum_{i=1}^k \beta_i|\w_i\cdot\bu|^{2p}\|\w_i\w_i^\top - \w_i^{\perp\bu}\w_i^{\perp\bu\top}\|_2\leq \frac{\sqrt{\xi}}{10}\sum_{i=1}^k \beta_i|\w_i\cdot\bu|^{2p}.
\end{align*}
Recall that in \eqref{eq:weights-W_A-W_B} we defined: $\cA_{i^*}(\xi) =\{i\in[k]:|\w_i\cdot\w_{i^*}|\leq 1 - \xi\}$, and 
\begin{align*}
    W_A = \frac{1}{Z_p(\cU)}\sum_{i\in \cA_{i^*}(\xi)} \beta_i|\w_i\cdot\bu|^{2p},\; W_B = \frac{1}{Z_p(\cU)}\sum_{i\notin \cA_{i^*}(\xi)} \beta_i|\w_i\cdot\bu|^{2p}, \end{align*}
so that $W_A + W_B = ({1}/{Z_p(\cU)})\sum_{i=1}^k \beta_i|\w_i\cdot\bu|^{2p}$.
Then in \Cref{claim:close-to-rank-1-matrix}, equation \eqref{eq:||M-M_b||_2}, we proved that when $\bu\in\evt_{i^*}^{\mathrm{Cor}}\cap\evt_{i^*}^{\mathrm{Sep}}$, it holds $W_A\leq(\lambda/k) W_B$. 
Thus, we conclude that
\begin{align*}
     \|\mM_{2p+2}(\bu)^{\perp\bu} - \mM_{2p+2}(\bu)\|_2\leq\frac{\sqrt{\xi}}{10}(W_A + W_B)\leq \frac{\sqrt{\xi}(\lambda/k + 1)W_B}{10}.
\end{align*}
}\new{
On the other hand, since $\beta_i\geq 0$ for all $i\in[k]$, we have:
\begin{align*}
    \|\mM_{2p+2}(\bu)\|_2\geq \w_{i^*}^\top\mM_{2p+2}(\bu)\w_{i^*}\geq \frac{1}{Z_p(\cU)}\sum_{i\notin \cA_{i^*}(\xi)}\beta_i|\w_i\cdot\bu|^{2p} (\w_i\cdot\w_{i^*})^2\geq (1 - \xi)^2 W_B.
\end{align*}
Therefore, with the choice $\xi\in(0,1/300]$ and $\lambda\leq \sqrt{\xi}/10$, we conclude that:
\begin{align*}
    \|\mM_{2p+2}(\bu)^{\perp\bu} - \mM_{2p+2}(\bu)\|_2 \leq \frac{\sqrt{\xi}(\lambda/k + 1)}{10(1 - \xi^2)}\|\mM_{2p+2}(\bu)\|_2\leq\frac{\sqrt{\xi}}{8}\|\mM_{2p+2}(\bu)\|_2.
\end{align*}

Finally, note that $\|\mM_{2p+2}(\bu)\|_2\geq \beta_{i^*}|\w_{i^*}\cdot\bu|^{2p}/Z_p(\cU)$, since $\w_{i^*}$ is a $\lambda$-heavy hidden weight and $\bu$ satisfies $\evt_{i^*}^{\mathrm{Cor}}$, it holds that $\|\mM_{2p+2}(\bu)\|_2\geq \lambda$. This completes the proof of \Cref{claim:orthogonal-projected-matrix-is-close-to-original}.}
\end{proof}

\Cref{claim:orthogonal-projected-matrix-is-close-to-original} implies that an accurate estimation of the matrix $\mM_{2p+2}(\bu)^{\perp\bu}$ approximates the target matrix from \eqref{eq:target-matrix-M(u)} in operator norm.

The rest of the proof is devoted to determining how many samples are required to estimate $\mM_{2p+2}(\bu)^{\perp\bu}$. We show that:

\begin{lemma}\label{lem:sample-complexity-quasi}
Using
    $N \gtrsim\frac{d^2 3^{2p} k^2}{(\wh{\sigma}''_{2p})^2 \lambda^2 Z_p(\cU)^2}$ samples, with probability at least $2/3$, it holds
    \begin{align*}
        \big\|\wh{\mM}_{2p+2}(\bu)^{\perp\bu} - {\mM}_{2p+2}(\bu)^{\perp\bu}\big\|_2\leq \frac{\sqrt{\xi}}{5}\|\mM_{2p+2}(\bu)\|_2.
    \end{align*}
\end{lemma}

\begin{proof}[Proof of \Cref{lem:sample-complexity-quasi}.]
Set $Q= \wh{\sigma}''_{2p}Z_p(\cU)$, and
let $ \vec X_i=Q^{-1} y^{(i)}\hep_{2p}(\bu\cdot\x^{(i)})(\x^{(i)\perp\bu}\x^{(i)\perp\bu\top} - \mI_{\perp\bu})$ and note that $\E[\vec X_i]=\mM_{2p+2}(\bu)^{\perp\bu}$ by \Cref{claim:expectation-form-of-orthogonal-to-u-target}.
By Chebyshev inequality, we have that 
\[
\pr\Big[\Big\|\frac{1}{N}\sum \vec X_i-\mM_{2p+2}(\bu)^{\perp\bu}\Big\|_2 
\geq t \|\mM_{2p+2}(\bu)^{\perp\bu}\|_2\Big] 
\leq \frac{\E[\|\frac{1}{N}\sum \vec X_i-\mM_{2p+2}(\bu)^{\perp\bu}\|_2^2]}{t^2 \|\mM_{2p+2}(\bu)^{\perp\bu}\|_2^2} \;.
\]
By the independence of samples, 
$\E[\|\frac{1}{N}\sum \vec X_i-\mM_{2p+2}(\bu)^{\perp\bu}\|_2^2]
\leq \frac{2}{N}\E[\| \vec X_i\|_F^2]$.
In addition, by the definition of the labels, 
we know that $y = \sum_{i=1}^k \beta_i\sigma(\w_i\cdot\x) $.
Furthermore, by the independence between $\bu\cdot\x$ and $\x^{\perp\bu}$ 
and by \CS inequality, the second-order moment of $\vec X_i$ is bounded by
\begin{align*}
  \E[\| \vec X_i\|_F^2] 
  &\leq   Q^{-2} \E\left[y^2\hep_{2p}^2(\bu\cdot\x)\|(\x^{\perp\bu}\x^{\perp\bu\top} - \mI_{\perp\bu})\|_F^2\right]\\
  &=\frac{1}{Q^2}\E\Big[\Big(\sum_{i=1}^k \beta_i\sigma(\w_i\cdot\x)\Big)^2\hep_{2p}^2(\bu\cdot\x)\|(\x^{\perp\bu}\x^{\perp\bu\top} - \mI_{\perp\bu})\|_F^2\Big]\\
&\leq  \frac{1}{Q^2}\sqrt{\E\Big[\Big(\sum_{i=1}^k \beta_i\sigma(\w_i\cdot\x)\Big)^4\Big]\E[\hep_{2p}^4(\bu\cdot\x)\|(\x^{\perp\bu}\x^{\perp\bu\top} - \mI_{\perp\bu})\|_F^4 ]}\\
&=\frac{1}{Q^2}\sqrt{\E\Big[\Big(\sum_{i=1}^k \beta_i\sigma(\w_i\cdot\x)\Big)^4\Big]\E[\hep_{2p}^4(\bu\cdot\x)]\E[\|(\x^{\perp\bu}\x^{\perp\bu\top} - \mI_{\perp\bu})\|_F^4 ]} \;.
\end{align*}
The fourth moment of $y$ is bounded above by 
$\E[(\sum_{i=1}^k \beta_i\sigma(\w_i\cdot\x))^4]\lesssim k^4$.
In addition, using \Cref{fct:gaussian-hypercontractivity}, we have that $\E[\hep_{2p}^4(\bu\cdot\x)]\leq 3^{4p}$. Further,  $\E[\|(\x^{\perp\bu}\x^{\perp\bu\top} - \mI_{\perp\bu})\|_F^4]\leq d^4$.
Therefore, we have:
\begin{align*}
    \E[\| \vec X_i\|_F^2]\lesssim  \frac{d^2 3^{2p} k^2 }{ Q^2}.
\end{align*}
Plugging the upper bound on the second moment $\E[\| \vec X_i\|_2^2]$ back into Chebyshev inequality yields
\begin{align*}
    \pr\Big[\Big\|\frac{1}{N}\sum \vec X_i-\mM_{2p+2}(\bu)^{\perp\bu}\Big\|_2\geq t \|\mM_{2p+2}(\bu)^{\perp\bu}\|_2\Big]\lesssim \frac{d^2 3^{2p}\sqrt{k^4 }}{Q^2 N t^2 \|\mM_{2p+2}(\bu)^{\perp\bu}\|_2^2} \;.
\end{align*}
When  $\xi\in(0,1/300]$, by \Cref{claim:orthogonal-projected-matrix-is-close-to-original} 
and using the triangle inequality, it holds that 
$$\|\mM_{2p+2}(\bu)^{\perp\bu}\|_2\geq (1 - \sqrt{\xi}/8)\|\mM_{2p+2}(\bu)\|_2\geq (1/3)\lambda .$$
Thus, choosing $t\lesssim\sqrt{\xi}$, if the sample size satisfies\begin{align*}
    N \gtrsim \frac{d^2 3^{2p} k^2 }{Q^2 \lambda^2} = \frac{d^2 3^{2p} k^2 }{(\wh{\sigma}''_{2p})^2 \lambda^2 Z_p(\cU)^2} \;,
\end{align*}
we have that with probability at least $2/3$,
\begin{align*}
    \Big\|\wh{\mM}_{2p+2}(\bu)^{\perp\bu} - {\mM}_{2p+2}(\bu)^{\perp\bu}\Big\|_2 
    \leq \frac{\sqrt{\xi}\lambda}{10} 
    \leq \frac{\sqrt{\xi}}{10}\|\mM_{2p+2}(\bu)\|_2 \;,
\end{align*}
where in the last inequality we used the fact that conditioned on $\bu\in\evt_{i^*}^{\mathrm{Cor}}\cap\evt_{i^*}^{\mathrm{Sep}}\cap\evt^{\mathrm{Ort}}$, 
it holds $\|\mM_{2p+2}(\bu)\|_2\geq \lambda$ by 
\Cref{claim:orthogonal-projected-matrix-is-close-to-original}. 
This completes the proof of \Cref{lem:sample-complexity-quasi}.
\end{proof}

Now combining \Cref{claim:orthogonal-projected-matrix-is-close-to-original}, 
\Cref{lem:sample-complexity-quasi} and using the triangle inequality yields:
\begin{align*}
    \Big\|\wh{\mM}_{2p+2}(\bu)^{\perp\bu} - {\mM}_{2p+2}(\bu)\Big\|_2 
    &\leq \Big\|\wh{\mM}_{2p+2}(\bu)^{\perp\bu} - {\mM}_{2p+2}(\bu)^{\perp\bu}\Big\|_2 + 
    \Big\|\mM_{2p+2}(\bu)^{\perp\bu} - \mM_{2p+2}(\bu)\Big\|_2\\
    &\leq \new{\frac{\sqrt{\xi}}{3}}\Big\|\mM_{2p+2}(\bu)\Big\|_2 \;.
\end{align*}
Since $p = \lceil\log(k/\lambda)/\log(1+\xi)\rceil$ and 
$Z_p(\cU) = k^{-p} = k^{-O(\log(k/\lambda))}$, 
$\wh{\sigma}_{2p+2} \geq 1/(2p+2)^{3}$ (by \Cref{fct:hermite-relu}), 
by \Cref{lem:sample-complexity-quasi} it follows 
that the number of samples required is 
\begin{align*}
    N = d^2\,\poly(p)k^{O(\log(k/\lambda))}.
\end{align*} The probability of success is at least $2/3$, 
conditioned on the event $\bu\in\evt_{i^*}^{\mathrm{Cor}}\cap\evt_{i^*}^{\mathrm{Sep}}\cap\evt^{\mathrm{Ort}}$.
\end{proof}

\paragraph{Proof of the Main Result of the Subsection: 
Quasi-Polynomial Complexity $\init$.}

We are now ready to prove \Cref{prop:quasi-poly-time-initialization-via-spectral}, which we restate below:
\QuasiPolyTimeInitialization*

\begin{proof}[Proof of \Cref{prop:quasi-poly-time-initialization-via-spectral}]
Fix any $i^*\in[k]$ such that $\beta_{i^*}\geq \lambda$. 
By \Cref{prop:sampling}, we have that 
$\pr[\evt_{i^*}^{\mathrm{Cor}}\cap\evt_{i^*}^{\mathrm{Sep}}]\geq \pr[\evt_{i^*}^{\mathrm{Cor}}\cap\evt_{i^*}^{\mathrm{Sep}}\cap\evt^{\mathrm{Ort}}]\geq \poly(1/k)$. 
Combining \Cref{prop:low-degree-separates-direction} and 
\Cref{lem:quasi-poly-sample-spectral}, 
for a probing vector $\bu$ satisfying $\evt_{i^*}^{\mathrm{Cor}}\cap\evt_{i^*}^{\mathrm{Sep}}\cap\evt^{\mathrm{Ort}}$, 
using $N = d^2 k^{O(\log(k/\lambda))}$ samples, 
the empirical matrix $\wh{\mM}_{2p+2}(\bu)$ 
from \Cref{line:emprical-matrix-quasi} of \nameref{alg:initialization-via-spectral} satisfies 
$|\br_1(\wh{\mM}_{2p+2}(\bu))\cdot\w_{i^*}|\geq 0.99$ 
with probability at least $2/3$. 
Thus, for any fixed $i^*\in[k]$ with $\beta_{i^*}\geq \lambda$, the probability that the top eigenvector of the empirical matrix $\br_{i^*} = \br_1(\wh{\mM}_{2p+2}(\bu))$ satisfies $|\br_{i^*}\cdot\w_{i^*}|\geq 0.99$
is at least $\poly(1/k)$.
Hence, by a union bound and a Chernoff bound, 
drawing $\poly(k)\log(1/\delta)$ probing vectors 
and constructing $\poly(k)\log(1/\delta)$ empirical matrices 
ensures that with probability at least $1 - \delta$, 
for any $i^*\in[k]$ with $\beta_{i^*}\geq \lambda$, the list $\cL$ returned 
by \nameref{alg:initialization-via-spectral} contains 
a unit vector $\br_{i^*}$ such that $|\br_{i^*}\cdot\w_{i^*}|\geq 0.99$. This completes the proof. 
\end{proof}

\subsection{Polynomial Time $\init$ via Implicit Tensor Computation}\label{sec:poly-time-initialization}

\new{
In this section, we provide a $\init$ subroutine that uses polynomially 
many samples and runs in polynomial time. Similar to \Cref{sec:quasi-poly}, the core idea
is to estimate the target matrix $\mM_{2p+2}(\bu)$ \eqref{eq:target-matrix-M(u)}, constructed by projecting the moment tensor 
$\E_{(\x,y)\sim\D}[y\he_{2p+2}(\x)]$ to a low-degree filter 
$\bu^{\ot p}\ot \bar{\bu}^{\ot p}$:
\begin{align*}
    \mM_{2p+2}(\bu) &= \new{\frac{1}{Z_p(\cU)}\bigg\la\frac{1}{\wh{\sigma}_{2p+2}}\E_{(\x,y)\sim\D}[y\he_{2p+2}(\x)],\, \bu^{\ot p}\ot \bar{\bu}^{\ot p} \bigg\ra} \nonumber\\
    &= \frac{1}{Z_p(\cU)} \sum_{i=1}^k b_i |\w_i\cdot\bu|^{2p}\w_i\w_i^{\top} , \end{align*}
where $\bu$ is a random probing vector sampled from some distribution $\cU$.
Once we obtain an accurate estimation $\wh{\mM}_{2p+2}(\bu)$, we can invoke \Cref{prop:low-degree-separates-direction}, which states that the top eigenvector of $\mM_{2p+2}(\bu)$ aligns well with the target vector $\w_i$, conditions on $\bu$ that are satisfied with inverse-polynomial probability.

To ensure polynomial complexity, we take a significant departure 
from \Cref{sec:quasi-poly} in both components of the algorithm: (1) instead of using the uniform distribution on the sphere, we choose $\cU$ to be a complex Gaussian distribution: $\bu = \bp + \im\bq$, $\bp,\bq\sim\calN(\vec 0, \mI_d)$; and (2) instead of estimating the matrix $\mM_{2p+2}(\bu)$ via low-degree polynomials, we use more advanced techniques from implicit tensor computation~\cite{diakonikolas2025implicit}.
Choosing $\cU$ to be the complex Gaussian has two main advantages: first, the operator norm of the target matrix $\wh{\mM}_{2p+2}(\bu)$ becomes larger, since $|\w_i\cdot\bu|$ is no longer restricted to be smaller than one, as was the case when $\bu$ were sampled uniformly on the sphere; 
second, the covariance of $\bu^{\ot p}$ for $\bu = \bp + \im\bq$ is bounded above by $2^p p!$ (cf.\ \Cref{lem:tensor-variance-S})---when $p\lesssim k$, this is much smaller compared to $O(k^{p/2})$, the covariance of $\bu^{\ot p}$ for $\bu\sim\mathrm{Unif}(\mathbb{S}^{k-1})$.
These benefits together ensure that we can apply the implicit tensor computation techniques to estimate the target matrix accurately using only a polynomial number of samples, in polynomial time. 
Finally, while it is possible to choose other types of distributions $\cU$ with similar properties; the complex Gaussian distribution is the most convenient one for our analysis.

The pseudo-code of the polynomial-time $\init$ subroutine is given in \nameref{alg:implicit-tensor-approximate-matrix}.
}

\newcommand{\AlgSection}[1]{\Statex \Statex \texttt{#1}
    \Statex }

\begin{algorithm}[H]
\caption{$\mathsf{CDRviaImplicitTensor}(\D,\lambda,\delta)$}
\label{alg:implicit-tensor-approximate-matrix}
\begin{algorithmic}[1]
\State\textbf{Input:} Sample access to a sum of $k$ ReLUs with Gaussian marginals $(\x,y)\sim\D$, positive parameters $\lambda,\delta\in(0,1)$. 
\State {\bf Output: } A list of unit vectors satisfying \Cref{assum:initialization-algorithm}.

\medskip

\State Fix an absolute constant $\xi\in(0,1/300]$\new{, $\lambda \gets \min\{\lambda, \sqrt{\xi}/10\}$}.
\State $p \gets \lceil\log(k/\lambda)/\log(1+\xi)\rceil$. \Comment{Filter degree}\State $N_\bu\gets \poly(k/\lambda)\log(\new{1}/\delta)$. \Comment{Number of probing vectors/Low-degree filter} 
\State $\tau\gets \delta/(10 d^2 N_\bu)$. \Comment{Failure probability of the immediate steps}
\State $N_0\gets \poly(d,k/\lambda,\new{1/\tau},1/\delta)$. \Comment{Number of samples}

\State Let the distribution $\cI\sim\D(\cI)$ be $\cI = (\x,y,\bz)$, where $(\x,y)\sim\D$ and $\bz\sim\calN(\vec 0,\mI_d)$.

\State $\wh{\Phi}_p\gets$\nameref{alg:implicit-tensor-computation-adapted}\label{line:obtain-recursive pseudo-projection} 

\State $\mA_p\gets $ \nameref{alg:CompressEmpiricalMomentTensor} \label{line:projected-monent}

\State $\cL\gets\emptyset$. \For{$j = 1,\dots,N_\bu$}
\State Draw a fresh vector $\bu = \bp + \im\bq$, where $\bp,\bq\sim\calN(\vec 0,\mI_d)$ 
\State $\mB_p(\bu)\gets$\nameref{alg:CompressFilterTensor} \label{line:projected-filter}

\State $\wh{\mM}_{2p+2}(\bu) = \la\mA_p, \mB_p(\bu)\ra$ \Comment{Empirical matrix, cf.~\ref{eq:intuition-implict-3}}\label{line:matrix-estimation-after-pseudo-projection}
\State $\cL\gets\cL\cup\{\pm \br_1(\wh{\mM}_{2p+2}(\bu))\}$. \Comment{Top eigenvectors of the empirical matrix}
\EndFor
\State\textbf{Return: }$\cL$.
\end{algorithmic}
\end{algorithm}

The whole section is devoted to proving \Cref{prop:poly-time-initialization-via-implicit-tensor}, which is restated below. We defer its proof to \Cref{sec:proof-poly-initialization}, after completing all the ingredients.
\ImplicitTensorPolyTimeInitialization*

\new{
\begin{remark}{\em
Throughout this section, we assume without loss of generality 
that $k\leq d$. If $k>d$, before running our $\init$ procedure 
\nameref{alg:implicit-tensor-approximate-matrix}, set $d'=\max\{d,k\}$.
If $d'>d$, augment every example $\x$ by an independent vector
$\bz\sim\mathcal{N}(0,\mI_{d'-d})$, and use $\x'=(\x,\bz)$ with the same
label. The lifted hidden weights are $\w_i'=(\w_i,0)$, so the target
and the Gaussian model are preserved. Run the $\init$ procedure (\nameref{alg:implicit-tensor-approximate-matrix}) in
dimension $d'$. If a unit vector $\br\in\mathbb{R}^{d'}$ is returned,
retain $\mP \, \br/\|\mP \, \br\|_2$ when $\mP \, \br\neq 0$, where $\mP$ projects onto the
original $d$ coordinates; discard a zero projection. For every
original unit weight $\w_i$, it holds
\[
    \bigg|\w_i\cdot\frac{\mP \, \br}{\|\mP \, \br\|_2}\bigg|
    =
    \frac{|\w_i'\cdot \br|}{\|\mP \, \br\|_2}
    \ge
    |\w_i'\cdot \br|.
\]
Thus, a useful output retains its correlation guarantee after
projection, the output list cannot grow, and the cost remains
polynomial in $d$ and $k$. In the tensor analysis below, relabel
$d'$ as $d$, so $d\ge k$.}
\end{remark}
}

\subsubsection{Determining the Probability of the Probing Event}

Our first step is to determine the probability that the signal of the target vectors passes through the random filter---in other words, the probability that the ``separation'' and ``correlation'' events $\evt_{i^*}^{\mathrm{Sep}}$, $\evt_{i^*}^{\mathrm{Cor}}$ (defined  in \eqref{eq:cor-sep-events}) happen 
under the complex distribution $\cU$.
We prove that, with $p = \lceil\log(k/\lambda)/\log(1+\xi)\rceil$ and $Z_p(\cU) =p!$, the event $\evt_{i^*}^{\mathrm{Sep}}\cap\evt_{i^*}^{\mathrm{Cor}}$ happens with probability at least $\poly(\lambda/k)$.

\begin{lemma}[Probability of Successful Filtration]\label{lem:angle-separation-random}
Let $\bu = \bp + \im\bq,$ where $ \bp,\bq \stackrel{iid}{\sim}\calN(0,\mI_{d})$. 
Let $\xi\in(0, 1/300]$ be a small absolute 
constant, $p = \lceil\log(k/\lambda)/\log(1+\xi)\rceil$ and 
$Z_p(\cU) =p!$. Let $\w_1,\dots,\w_k$ be $k$ unit vectors.
Fixing any $i^*\in[k]$, let $\cA_{i^*}(\xi) =\{i\in[k]:|\w_i\cdot\w_{i^*}|\leq 1 - \xi\}$ 
and let the ``correlation'' and ``separation'' events on $\cU$ be defined by \eqref{eq:cor-sep-events}. Then, $\pr_{\bu\sim\cU}[\evt_{i^*}^{\mathrm{Cor}}\cap\evt_{i^*}^{\mathrm{Sep}}]\geq\poly(\lambda/k)$.
\end{lemma}

\begin{proof}
Without loss of generality, let $i^* =1$ and fix $\w_1$. 
Let $C_\xi = \sqrt{12}((1+\xi)/\xi^2)$ and $R=C_\xi\sqrt{\log( k/\lambda)}$. 
By our choice of parameters, we have $R^2\geq p$.
Furthermore, since $R^{2p}\geq p^p\geq p! = Z_p(\cU)$, we have
\begin{align*}
    \evt_{1}^{\mathrm{Cor}} = \{\bu: |\w_{1}\cdot\bu|\geq (Z_p(\cU))^{1/(2p)}\}\supset \{\bu: |\w_{1}\cdot\bu|\geq R\} := \evt_{1}^{\mathrm{Cor}, 2}.
\end{align*}

We proceed by bounding the probability of the event $\evt_{1}^{\mathrm{Cor}}$. For any fixed unit vector $\br\in\R^d$, we have $|\br\cdot\bu|^2 = (\bp\cdot\br)^2 + (\bq\cdot\br)^2 = X_1^2 + X_2^2$, where $X_1, X_2$ are independent standard Gaussian random variables; therefore, by the anti-concentration inequality for the standard Gaussian distribution, it holds that
\begin{align*}
    \pr[|\br\cdot\bu|\geq R] &= \pr[|\br\cdot\bu|^2\geq R^2]\geq \pr[X_1^2\geq R^2/2, X_2^2\geq R^2/2] 
    \geq \pr[|X_1|\geq R/\sqrt{2}]^2\\
    &\geq \exp(-cR^2) = \poly(\lambda/k). 
\end{align*}
Therefore, $\pr[\evt_{1}^{\mathrm{Cor}}]\geq \pr[\evt_{1}^{\mathrm{Cor}, 2}]\geq \poly(\lambda/k)$.

It remains to bound $\pr[\evt_{1}^{\mathrm{Sep}}|\evt_{1}^{\mathrm{Cor}, 2}]$.
Denote $z_i\eqdef \bu\cdot\w_i$. 
Consider the joint distribution of $z_1 = \w_1\cdot\bu$ and $z_i = \w_i\cdot\bu$ 
for any $i\in\cA_1(\xi)$. Letting $\cos\theta_i\eqdef \cos(\theta_{\pm}(\w_i,\w_1)) = |\w_i\cdot\w_1|$, 
  we have $z_i = \bu\cdot(\cos\theta_i\w_1 + \sin\theta_i\mathbf{s}_i)$, where 
  $\mathbf{s}_i = \w_i^{\perp\w_1}/\|\w_i^{\perp\w_1}\|_2$. 
  Denote for simplicity $g_i = \bu\cdot\mathbf{s}_i$. 
We condition on $\evt_{1}^{\mathrm{Cor}, 2}$, i.e., $|z_1| = Z \geq R$. 
    For any $i\in\cA_1(\xi)$, we have 
    $\cos(\theta_i) = |\w_i\cdot\w_1| \leq 1 - \xi$.
    By the triangle inequality, whenever 
    $|\w_i\cdot \vec u|\geq |\w_1\cdot \vec u|/(1+\xi)$ and $|\w_1\cdot \vec u| = Z$, it holds
    \[
       |g_i|\geq -\frac{Z\cos(\theta_i)}{\sin(\theta_i)} + 
       \frac{Z}{(1+\xi)\sin(\theta_i)}\geq 
       \frac{Z}{\sin(\theta_i)}\bigg(\frac{1}{1+\xi} - (1-\xi)\bigg)\geq \frac{Z \xi^2}{1+\xi}=\frac{\sqrt{12}Z}{C_\xi}\;,
    \]
   where we used that $\cos(\theta_i) \leq 1 - \xi$.
  Therefore, we have:
    \begin{align}\label{eq:prob-evt_2}
         \pr\bigg[|z_i|\geq \frac{|z_1|}{1+\xi}\Big|\, |z_1| = Z, Z\geq R\bigg] & \leq \pr\bigg[|g_i|\geq \frac{\sqrt{12}Z}{C_\xi} \Big|\, Z\geq R\bigg] \leq \pr\bigg[|g_i|\geq \frac{\sqrt{12}R}{C_\xi}\bigg]. \end{align}
    Recall that $|g_i|^2 = |(\bp + \im\bq)\cdot\bs_i|^2 = X_1^2 + X_2^2\sim \chi_2^2$, 
    therefore, by the standard concentration inequality of $\chi_2^2$ random variables~\cite{laurent2000adaptive},
    we have
    \begin{align*}
        \pr[|g_i|\geq t] = \pr[|g_i|^2\geq t^2]\leq \exp(-t^2/6)
    \end{align*}
    when $t^2\geq 2$. Plugging in $t = \sqrt{12}R/C_\xi = \sqrt{12\log(k/\lambda)}$, and combining with \eqref{eq:prob-evt_2}, we obtain that:
    \begin{align*}
       \pr\bigg[|z_i|\geq \frac{|z_1|}{1+\xi}\Big|\,|z_1|\geq R\bigg]\leq \exp(-2\log(k/\lambda)) = (\lambda/k)^2.
    \end{align*}
    Applying a union bound on all such events for $i\in \cA_1(\xi)$, we obtain 
    \begin{align*}
        \pr[\exists i\in \cA_1(\xi),\;|z_i|\geq |z_1|/(1+\xi)\mid|z_1| \geq R]\leq \lambda^2/k.
    \end{align*}
    Therefore, we get that $\pr[\evt_{1}^{\mathrm{Sep}}|\evt_{1}^{\mathrm{Cor}, 2}]\geq 1 -\lambda^2/k$ and hence $\pr[\evt_{1}^{\mathrm{Sep}}\cap \evt_{1}^{\mathrm{Cor}}]\geq \pr[\evt_{1}^{\mathrm{Sep}}\cap \evt_{1}^{\mathrm{Cor}, 2}]\geq \poly(\lambda/k)$, 
    completing the proof of \Cref{lem:angle-separation-random}.
\end{proof}

\subsubsection{Estimating the Target Matrix Using Implicit Tensor Computation}

In this section, we show how to use implicit tensor computation 
to estimate the target matrix \eqref{eq:target-matrix-M(u)}. 
Our main result in this context is the following:

\begin{restatable}[Matrix Estimation Using Implicit Tensor Computation]{proposition}{MatrixEstImplicitTensor}\label{prop:estimate-matrix-via-implicit-tensor}
Let $\D$ be a Gaussian sum of $k$ ReLUs.
     Let $p, N_0, N_\bu\in \Z_+$, $\xi,\lambda,\delta\in(0,1)$ be 
     the parameters prescribed in~\nameref{alg:implicit-tensor-approximate-matrix}. 
Let ${\mM}_{2p+2}(\bu)$ be the matrix defined in \eqref{eq:target-matrix-M(u)}.
Let $\wh{\cU}$ be a set of $N_\bu$ random probing vectors 
following the distribution $\bu = \bp + \im\bq$, $\bp,\bq\stackrel{iid}{\sim}\calN(0,\mI_d)$.
Then the matrices $\wh{\mM}_{2p+2}(\bu)$ computed 
in \Cref{line:matrix-estimation-after-pseudo-projection} of \nameref{alg:implicit-tensor-approximate-matrix} satisfy the following guarantee: with probability at least $1 - \delta/2$ over 
all the samples and the probing vectors,  
$\max_{\bu\in\wh{\cU}}\|\wh{\mM}_{2p+2}(\bu) - \mM_{2p+2}(\bu)\|_2 \leq \sqrt{\xi}\lambda \;.$
\end{restatable}

The rest of this subsection is devoted to the proof of \Cref{prop:estimate-matrix-via-implicit-tensor}.

\medskip

\noindent Our first observation is that $\mM_{2p+2}(\bu)$ is obtained via contracting 
the moment tensor $\E_{(\x,y)\sim{\D}}[y\he_{2p+2}(\x)]$ with $\mF_p(\bu)\ot\mF_p(\bar\bu)$, $\mF_p(\bu) = \bu^{\ot p}/\sqrt{p!}$, where we remind the reader that $\bar{\bu}$ 
is the complex conjugate of the vector $\bu$.
\begin{claim}\label{claim:explicit-expression-of-M}
     Let $\mM_{2p+2}(\bu)$ be the matrix defined by $(i)$ in \eqref{eq:target-matrix-M(u)}. Let $\mF_p(\bu) = \bu^{\ot p}/\sqrt{p!}$.
Then:
    \begin{align*}
        \mM_{2p+2}(\bu)=\frac{1}{\wh{\sigma}_{2p+2}}\bigg\la\E_{(\x,y)\sim{\D}}\bigg[y\he_{2p+2}(\x)\bigg], \mF_p(\bu)\ot\mF_p(\bar\bu)\bigg\ra. 
    \end{align*}
\end{claim}
\begin{proof}
By the Hermite expansion of the ReLU activation $\sigma$, we have
\begin{align*}
    \E_{(\x,y)\sim{\D}}\bigg[y\he_{2p+2}(\x)\bigg] = \sum_{i=1}^k \beta_i \E_{(\x,y)\sim{\D}}\bigg[\sigma(\w_i\cdot\x)\he_{2p+2}(\x)\bigg] = \sum_{i=1}^k \beta_i \wh{\sigma}_{2p+2}\w_i^{\ot(2p+2)}.
\end{align*}
    Taking the tensor product and noting that $\la\w_i^{\ot 2p}, \bu^{\ot p}\ot\bar{\bu}^{\ot p}\ra = (\w_i\cdot\bu)^p(\w_i\cdot\bar{\bu})^p = |\w_i\cdot\bu|^{2p}$ gives the stated claim.
\end{proof}

Therefore, the problem boils down to estimating the contraction of two tensors $\E_{(\x,y)\sim{\D}}[y\he_{2p+2}(\x)]$ 
and $\mF_p(\bu)\ot\mF_p(\bar\bu)$ using samples from $\D$. 
To this end, we apply the implicit tensor technique 
from~\cite{diakonikolas2025implicit}.

\paragraph{Preliminaries for Implicit Tensor Computation.} 
\new{The central idea is the following: 
Recall that the target matrix $\mM_{2p+2}(\bu)$ is defined by
\begin{align*}
    \mM_{2p+2}(\bu) = \frac{1}{\wh{\sigma}_{2p+2}}\bigg\la\E_{(\x,y)\sim{\D}}\bigg[y\he_{2p+2}(\x)\bigg], \mF_p(\bu)\ot\mF_p(\bar\bu)\bigg\ra = \sum_{i=1}^k \la\w_i^{\ot p}, \mF_p(\bu)\ra \la\w_i^{\ot p}, \mF_p(\bar\bu)\ra\w_i\ot\w_i.
\end{align*}
Let $\bh_1,\dots,\bh_k\in(\R^d)^{\ot p}$ be the basis of the subspace $\sU\subset(\R^d)^{\ot p}$ spanned by ``vectors'' $\w_i^{\ot p}$, $i\in[k]$. Suppose we have the projection $\Phi_p\in\R^{k\times d^p}$, $\Phi_p^\top = [\bh_1,\dots,\bh_k]$ that maps a ``vector'' from $(\R^d)^{\ot p}$ to the subspace $\sU\subset(\R^d)^{\ot p}$ 
followed by an isometry from $\sU$ to $\R^k$, then since $\w_i^{\ot p}$ is a ``vector'' in $\sU$, it holds $\Phi_p^\top\Phi_p(\w_i^{\ot p}) = \w_i^{\ot p}$. In addition, since $\Phi_p$ is a matrix with real entries, consider its natural complex extension on a complex tensor $\mF_p(\bu)$: $\Phi_p(\mF_p(\bu)) = \Phi_p(\Real(\mF_p(\bu))) + \im\Phi_p(\mathrm{Im}(\mF_p(\bu)))$. Then its Hermitian adjoint is $\Phi^*_p = \Phi_p^\top$. Therefore, we have
\begin{align*}
    \la\w_i^{\ot p}, \mF_p(\bu)\ra = \la\Phi_p^\top\Phi_p(\w_i^{\ot p}), \mF_p(\bu)\ra = \la\Phi_p(\w_i^{\ot p}), (\Phi_p^\top)^* (\mF_p(\bu))\ra =\la\Phi_p(\w_i^{\ot p}), \Phi_p(\mF_p(\bu))\ra. 
\end{align*}
By the linearity of expectation, this implies that the target matrix $\mM_{2p+2}(\bu)$ can be written as:
\begin{align}\label{eq:intuition-implicit-1}
    \mM_{2p+2}(\bu) &= \sum_{i=1}^k \la\Phi_p(\w_i^{\ot p}), \Phi_p(\mF_p(\bu))\ra \la\Phi_p(\w_i^{\ot p}), \Phi_p(\mF_p(\bar\bu))\ra\w_i\ot\w_i \nonumber\\
    &=\sum_{i=1}^k \bigg\la(\Phi_p\ot\Phi_p\ot\mI_d\ot\mI_d)(\w_i^{\ot 2p+2}), \Phi_p(\mF_p(\bu))\ot \Phi_p(\mF_p(\bar\bu))\bigg\ra \nonumber\\
    &= \bigg\la\E_{(\x,y)\sim{\D}}\bigg[(\Phi_p\ot\Phi_p\ot\mI_d\ot\mI_d)\bigg(\frac{1}{\wh{\sigma}_{2p+2}}y\he_{2p+2}(\x)\bigg)\bigg], \Phi_p(\mF_p(\bu))\ot\Phi_p(\mF_p(\bar\bu))\bigg\ra.
\end{align}
Notably, the right-hand side of the equality above is a product of a $k\times k\times d\times d$ tensor and a $k\times k$ matrix, rather than two higher-order tensors.
This implies that the dimension of the estimation problem is greatly reduced from $d^{2p}$ to $k^2 + (dk)^2$, via the ``compressor'' $\Phi_p$.
The main idea of~\cite{diakonikolas2025implicit} is that, when $\sum_{i=1}^k \w_i^{\ot t} = \E_{\cI}[\mS_t(\cI)]$ for $t = 2,4,\dots,2p$, where $\mS_t(\cI)$ are recursively-constructed random tensors depending on the random variable $\cI$, the compressor $\Phi_p$ can be discovered and constructed recursively and efficiently from samples $\cI$.
Furthermore, $\Phi_p$ can also be efficiently applied to other recursively-constructed tensors. }

\new{
These recursively-constructed tensors are formally defined as tensors generated from {\it  Sequential Tensor Computation (STC) circuits}.
}

\begin{definition}[Sequential Tensor Computation~\cite{diakonikolas2025implicit}]\label{def:stc}
A \textit{sequential tensor computation} (STC) is an arithmetic circuit $\sfS_t$ that takes as input $\mathcal{I}$ a set of scalars and vectors in $\mathbb{C}^d$ and outputs an order-$t$ tensor of dimension $d$, denoted by $\mS_t(\mathcal{I}) \in (\mathbb{C}^d)^{\otimes t}$, for some $t \in \mathbb{Z}_+$, by applying the following operations: scalar multiplication,
addition of objects of the same order and dimension, and appending a vector by tensor product (in the last component). When $\cI$ is sampled from a distribution $\D(\cI)$, we denote $\mS_t(\mathcal{I})\sim\sfS_t(\D(\cI))$. The size of $\sfS_t$ is defined as 
the size of the underlying arithmetic circuit.
\end{definition}

\new{In short, Sequential Tensor Computation (STC) circuits capture a class of tensors that can be constructed via recursive simple operations. 
One simple example is order-$p$ rank-1 tensors, i.e., tensors of the form $\bu^{\ot p}$, since these tensors can be constructed by recursively adding a new vector in its last component via tensor product. 
These tensors are exactly the objectives we consider.
}

\begin{lemma}[Moment Tensor and Filter Tensor are STCs]\label{lem:stcs}
    Let $\D$ be a Gaussian sum of $k$ ReLUs,
    and let $\bz\sim\calN(\vec 0,\mI_d)$. 
    Let $\cI = (\x,y,\bz)$ and denote the distribution of $\cI$ by $\D(\cI) = \D\ot\calN(\vec 0,\mI_d)$. For even $t\in\sZ_+$, denote the $t^{\text{th}}$ Hermite coefficient of the ReLU activation $\sigma$ by $\wh{\sigma}_t$. 
    Then, for any even $t\in\sZ_+$, there is a real STC $\mathsf{S}_t(\D(\cI))$ of size $O(t)$ that given a vector $\cI\sim\D(\cI)$, outputs:
    \begin{align}\label{eq:stc-S}
        \mS_t(\cI) = \mS_t(\x,y,\bz)\eqdef \frac{y}{\wh{\sigma}_t\sqrt{t!}}\Real((\x + \im \bz)^{\ot t}).
    \end{align}
    The tensor $\mS_{t}(\x,y,\bz)$
 satisfies
 \begin{align*}
     \E_{\bz\sim\calN(\vec 0,\mI_d)}[\mS_t(\x,y,\bz)] = \frac{1}{\wh{\sigma}_{t}} y\he_{t}(\x),\quad \E_{\cI\sim\D(\cI)}[\mS_t(\cI)] = \sum_{i=1} \beta_i\w_i^{\ot t}\;.
 \end{align*}
 Furthermore, let $\bu \sim\cU, \bu = \bp + \im\bq$, $\bp,\bq\sim\calN(\vec 0, \mI_d)$.
    Then, for any $p\in\sZ_+$, there is a complex STC $\sfF_p(\cU)$ of size $O(p)$ that given a vector $\bu$, outputs:
\begin{align}\label{eq:stc-F}
        \mF_p(\bu) = \frac{\bu^{\ot p}}{\sqrt{p!}} \;.
    \end{align}
\end{lemma}
\begin{proof}
    We first show that $\mathsf{S}_t(\D(\cI))$ is a real STC of size $O(t)$. Let $a= {y}/({\wh{\sigma}_t \sqrt{t!}})$. Define tensors $\vec R_j, \vec {Im}_j$ such that $\vec R_j+\im \vec {Im}_j=a (\vec x+\im \vec z)^{\ot j}$. To construct $j+1$, note that
\begin{align*}
   a (\vec x+\im \vec z)^{\ot (j+1)}=  (\vec R_j+\im \vec {Im}_j)\otimes(\vec x+\im \vec z)= (\vec R_j\otimes \vec x -  \vec {Im}_j\otimes \vec z)+ \im  (\vec R_j\otimes \vec z +  \vec {Im}_j\otimes \vec x) \;.
\end{align*}
This gives the recursive construction: $\vec R_{j+1}=\vec R_j\otimes \vec x -  \vec {Im}_j\otimes \vec z$ and $\vec {Im}_{j+1}= \vec R_j\otimes \vec z +  \vec {Im}_j\otimes \vec x$ with $\vec R_{0}=a$ and $\vec {Im}_0=0$. Each step uses an STC circuit of size $O(1)$, therefore the real STC $\mathsf{S}_t(\D(\cI))$ is of size $O(t)$.

    Next, applying \Cref{fct:hermite-tensor}, we have 
    \begin{align*}
        \E_{\bz\sim\calN(0,\mI_d)}[\mS_t(\x,y,\bz)] = \frac{y}{\wh{\sigma}_t}\E_{\bz\sim\calN(0,\mI_d)}\bigg[\frac{\Real((\x + \im \bz)^{\ot t})}{\sqrt{t!}} \bigg] = (1/\wh{\sigma}_t){y}\he_t(\x) \;.
    \end{align*} 
    Therefore, by the law of total expectation, we have $\E_{\cI\sim\D(\cI)}[\mS_t(\cI)] = (1/\wh{\sigma}_t)\E_{(\x,y)\sim\D}[y\he_t(\x)] = \sum_{i=1}^k \beta_i\w_i^{\ot t}$, using the definition of Hermite expansions. 

    For $\mF_p(\bu)$, directly by the definition of STCs, we know that it can be generated by an STC $\sfF_p(\cU)$ of size $O(p)$. This completes the proof of \Cref{lem:stcs}.
\end{proof}

\new{Combining \Cref{claim:explicit-expression-of-M} and \Cref{lem:stcs}, and using \eqref{eq:intuition-implicit-1}, the task of estimating the matrix $\mM_{2p+2}(\bu)$ is reduced to estimating the tensor contraction of these two STCs mentioned above:
\begin{align}\label{eq:intuition-implicit-2}
    \mM_{2p+2}(\bu) &= \bigg\la\E_{\cI\sim\D(\cI)}[\mS_{2p+2}(\cI)], \mF_p(\bu)\ot\mF_p(\bar\bu)\bigg\ra \nonumber\\
    &=\bigg\la\E_{\cI\sim\D(\cI)}\bigg[(\Phi_p\ot\Phi_p\ot\mI_d\ot\mI_d)\bigg(\mS_{2p+2}(\cI)\bigg)\bigg], \Phi_p(\mF_p(\bu))\ot\Phi_p(\mF_p(\bar\bu))\bigg\ra, \end{align} 
where $\Phi_p$ is the ``compressor'' we are looking for that maps from $(\R^{d})^{\ot p}$ to the subspace spanned by the hidden weight tensors $\w_i^{\ot p}$.
}
\new{
These compressors are formally called (recursive) pseudo-projections:}

\begin{definition}[Recursive pseudo-projections; Definitions 3.4 and 3.6 in \cite{diakonikolas2025implicit}]
\label{def:recursive-pseudo-projections}
A real linear map $\Phi:\sU\to\sV$ between finite-dimensional inner product spaces $\sU$ and $\sV$ is a \emph{pseudo-projection} if
$\Phi\Phi^\top = \mI_\sV$.
A $(d,t,k)$-recursive pseudo-projection $\Phi$ has dimensions
$n_0=1,n_1,\ldots,n_t=k$ and base pseudo-projections
$
\phi_j:\R^{n_{j-1}}\otimes\R^d
\to\R^{n_j}
$
such that
\[
\Phi_0=\mI_{d},\qquad
\Phi_j=\phi_j(\Phi_{j-1}\otimes \mI_d),
\quad 1\le j\le t,\qquad
\Phi=\Phi_t.
\]
Its order is $t$ and its size is $k$.
\end{definition}
\new{
\begin{remark}[Decomposition of Pseudo-Projections]
{\em As proved in \cite[Lemma 3.5]{diakonikolas2025implicit}, every pseudo-projection $\Phi$ can be written as a projection $\mP:\sU\to\sU'$ (i.e., $\mP$ is a linear map satisfying $\mP\mP^\top = \mI$) composed with an isometry $\mR$ from $\sU'$ to $\sV$; and in addition, the composition of pseudo-projections is a pseudo-projection.
    This implies the recursive map defined above is indeed a pseudo-projection, as tensorization with an identity and composition preserves the equation $\Phi_j\Phi_j^\top=\mI_{n_j}$.}
\end{remark}
}

\new{
\begin{remark}[Complex Extension of the Pseudo-Projections]{\em
    The mapping $\Phi$ is a linear map with real matrix entries. 
We also use the same notation
for its complex-linear extension, defined by
\[
    \Phi(\mW+\im \mT)=\Phi(\mW)+\im\Phi(\mT),
    \qquad \mW,\mT \text{ real tensors in real space $\sU$}.
\]
Thus, $\Phi$ can be applied to complex tensors.
Since the matrix entries of $\Phi$ are real, its Hermitian
adjoint coincides with its transpose $\Phi^* = \Phi^\top$, and $\Phi\Phi^*=\mI_\sV,
    \Phi(\overline{\mW+\im \mT})=\overline{\Phi(\mW+\im \mT)}.$
In the rest of the section, whenever we apply this real linear map $\Phi$ to a complex tensor, we are applying its natural complex extension.
}\end{remark}
}

\begin{fact}\label{fact:pseudo-projection-tensor-facts}
  Let $\Phi:\sU\to\sV$ be a real pseudo-projection. The following all hold:
  \begin{enumerate}
      \item For any $\vec W\in \sV$, $ \|\Phi^\top \vec W\|_F= \| \vec W\|_F$.\label{pseudo-prop-frob-equl}
      \item For every $\vec T\in \sU$, $ \|\Phi \vec T\|_F= \|\Phi^\top \Phi \vec T\|_F\leq \|\vec T\|_F$.\label{pseudo-prop-frob-equal-ineq}
      \item For any $\vec V\in \sV,\vec T\in \sU$, 
      $ \la \Phi \vec T ,\vec V\ra= \la  \vec T ,\vec \Phi^\top \vec V\ra$.\label{pseudo-prop-frob-change-order}
  \end{enumerate}
\end{fact}
\begin{proof}
For \Cref{pseudo-prop-frob-equl}, we have that \[
    \|\Phi^\top \vec W\|_F^2=\la \Phi^\top \vec W, \Phi^\top \vec W \ra=\la  \vec W, \Phi\Phi^\top \vec W \ra=\la  \vec W, \vec W \ra=\|\vec W\|_F^2\;.\]
    For \Cref{pseudo-prop-frob-equal-ineq}, the first part follows by the same argument as  \Cref{pseudo-prop-frob-equl}. The second part follows from  the fact that every pseudo-projection $\Phi$ can be written as a projection $\mP:\sU\to\sU'$ (i.e., $\mP$ is a linear map satisfying $\mP\mP^\top = \mI$) composed with an isometry from $\sU'$ to $\sV$. 
\end{proof}

To obtain the pseudo-projection $\Phi_p$, we apply the results from~\cite{diakonikolas2025implicit}. 
However, the output of the main algorithm in~\cite{diakonikolas2025implicit} 
is a real-valued inner product between two projected tensors. 
In our application, 
we aim to utilize the recursive pseudo-projections explicitly.
Due to this small difference, we slightly adapt the main algorithm from~\cite{diakonikolas2025implicit} to our setting, 
as stated in \nameref{alg:implicit-tensor-computation-adapted}.
\begin{algorithm}[H]
\caption{$\mathsf{TensorCompressor}
(p,N_0,\D(\cI))$
}
\label{alg:implicit-tensor-computation-adapted}
\begin{algorithmic}[1]
\State {\bf Input:}
Parameters $k,d,p,N_0\in\mathbb Z_+$; sample access to $\D(\cI)$
\State {\bf Output:}
A recursive pseudo-projection $\wh{\Phi}_p$,
represented by $\phi_1,\ldots,\phi_p$.

\medskip
\State Construct STC circuit $\sfS_t$, $t=2,4,\dots,2p+2$, according to \Cref{lem:stcs}. \State $\wh{\Phi}_0\gets \mI_d$.
\For{$r=1,\ldots,p$}
    \State $\Psi_r \gets \wh{\Phi}_{r-1}\otimes I_d$. \Comment{$\wh{\Phi}_{r-1}:(\R^d)^{\ot (r-1)}\to\R^k$, $\Psi_r: (\R^d)^{\ot (r-1)}\ot \R^d \to\R^k\ot \R^d$}
    \State Draw $N_0$ fresh samples from $\D(\cI)$, construct empirical distribution $\widehat{\D}(\cI)$.
    \State Let $\mS_{2r}(\mathcal I)$ be the output of $\sfS_{2r}$ on sample $\cI$. \State
    $\mV_r \gets \E_{\cI\sim\widehat{\D}(\cI)}[(\Psi_r\otimes\Psi_r)\mS_{2r}(\mathcal I)]$ \Comment{Viewed as a $d\times d$ matrix for $r=1$
    and a $dk\times dk$ matrix otherwise.}
    \State $\mU_r\gets \mathsf{Top(k)Eigenvectors}[\mV_r]$
    \State
    $\phi_r \gets \mU_r^\top$;
    $\wh{\Phi}_r \gets \phi_r\Psi_r$.
\EndFor
\State \Return $\wh{\Phi}_p$.
\end{algorithmic}
\end{algorithm}

\new{\nameref{alg:implicit-tensor-computation-adapted} has the following guarantee:}

\begin{proposition}[Implicit Tensor Computation]\label{thm:implict tensor adapted}
Let $\D$ be a Gaussian sum of $k$ ReLUs and denote $\mM_{p} = \sum_{i=1}^k \beta_i\w_i^{\ot p}$. Let $\cI = (\x,y,\bz)$ where $(\x,y)\sim\D$, $\bz\sim\calN(\vec 0,\mI_d)$ and let $\sfS_t(\D(\cI))$ be the STC defined in \Cref{lem:stcs}, Equation~\eqref{eq:stc-S}. 
Let $p \in \mathbb{Z}_{+}$. Suppose that for all positive even integers $t \leq 2p$, 
the output tensor $\mS_{t}(\mathcal{I})$, $\mathcal{I} \sim \mathcal{D}_{t}(\cI)$, 
has covariance bounded above by $V_\sfS$, for some $V_\sfS > 0$.

Then \nameref{alg:implicit-tensor-computation-adapted} has the following guarantee:
the algorithm draws $pN_0$ samples from $\mathcal{D}(\cI)$, runs in $\poly(d, p, N_0)$ time, outputs
a $(d,p,k)$-recursive pseudo-projection $\wh{\Phi}_p$ from $(\R^d)^{\otimes p}$ to $\R^k$, 
such that with probability at least $1 - \tau$ (over the samples drawn from $\mathcal{D}(\cI)$), $\wh{\Phi}_p$ satisfies
    \begin{align*}
        \|\wh{\Phi}_p^\top\wh{\Phi}_p(\mM_p) - \mM_p\|_F\leq kp\sqrt{\zeta}, \;\text{and for each $i\in[k]$, }
        \|\wh{\Phi}_p^\top\wh{\Phi}_p({\w}^{\otimes p}_i) - \w_i^{\otimes p}\|_F\leq  p\sqrt{\zeta/\beta_i} \;, \end{align*}
    where $\zeta = O({pdkV_\sfS}/{\sqrt{N_0\tau}})$. 
\end{proposition}
\begin{proof}[Proof of \Cref{thm:implict tensor adapted}]
Recall we proved that $\E_{\cI\sim\D(\cI)}[\mS_t(\cI)] = \sum_{i=1} \beta_i\w_i^{\ot t}$ in \Cref{lem:stcs}, which satisfies the prerequisites of Proposition 3.2 in~\cite{diakonikolas2025implicit}.  Then the proof follows directly from  \cite[Lemma 3.9 and Claim 3.10]{diakonikolas2025implicit}. 
\end{proof}

\Cref{thm:implict tensor adapted} states that, using $\poly(d,k,p,1/\tau,V_\sfS)$ samples, we can approximately recover the compressor $\Phi_p$ we are looking for: the recursive pseudo-projection $\wh{\Phi}_p$ will almost project to the subspace spanned by the hidden weight tensors $\w_i^{\ot p}$ whose coefficients are not too small.

\paragraph{Applying Implicit Tensor Computation for Target Matrix Approximation.}

To apply \nameref{alg:implicit-tensor-computation-adapted} and \Cref{thm:implict tensor adapted}, we need to determine the covariance of the STC output $\mS_t(\cI)$ defined in \Cref{lem:stcs}. To do so, we establish the following key lemma: 

\begin{lemma}\label{lem:tensor-variance-S}
    Let $\D$ be a Gaussian sum of $k$ ReLUs,
    and let $\bz\sim\calN(\vec 0,\mI_d)$. 
    Let $\cI = (\x,y,\bz)$ and denote the distribution of $\cI$ by $\D(\cI) = \D\ot\calN(\vec 0,\mI_d)$.
Let $\mS_t(\cI)$ be the output of the STC circuit defined in \Cref{lem:stcs}, Equation \eqref{eq:stc-S}. Then 
    the covariance of $\mS_t(\cI)$ is bounded by 
    \begin{align*}
        \sup_{\|\mW\|_F = 1}\E_{\cI\sim\D(\cI)}[|\la\mS_t(\cI) - \E_{\cI\sim\D}[\mS_t(\cI)], \mW\ra|^2]\leq V_{\sfS_t} := k^2 \new{2^{O(t)}}/\wh{\sigma}_t^2 \;. \end{align*}
\end{lemma}

\begin{proof}
        By the definition of tensor covariance and the distribution of the labels $y$, we have:
    \begin{align}\label{eq:variance-bound-0}
&\quad \sup_{\|\mW\|_F = 1}\E_{(\x,y,\bz)\sim\D(\cI)}\bigg[\bigg|\bigg\la \frac{y}{\wh{\sigma}_t \sqrt{t!}}\Real((\x + \im \bz)^{\ot t}), \mW\bigg\ra\bigg|^2\bigg] \nonumber\\
&\leq \sup_{\|\mW\|_F = 1}\frac{1}{\wh{\sigma}_t^2 t!}\sqrt{\Ey[y^4]\E_{\x,\bz\sim\calN(\vec 0,\mI_d)}\bigg[\bigg|\bigg\la\Real((\x + \im \bz)^{\ot t}), \mW\bigg\ra\bigg|^4\bigg]} \;,
\end{align}
    where the last inequality is by Cauchy-Schwarz. We prove the following claim:
\begin{claim}\label{claim:complex-gaussian-tensor-momemt-bounds}
        For any $\mW$ such that $\|\mW\|_F = 1$, it holds that 
        $$\sqrt{\E_{\x,\bz\sim\calN(\vec 0,\mI_d)}\bigg[\bigg|\bigg\la\Real((\x + \im \bz)^{\ot t}), \mW\bigg\ra\bigg|^4\bigg]} \leq 2^{O(t)} t!. $$
    \end{claim}
    \begin{proof}
    Note that \new{$|\la(\Real((\x + \im \bz)^{\ot t}),\mW\ra|^2= (\la \Real((\x + \im \bz)^{\ot t}),\Real(\mW)\ra)^2 + (\la \Real((\x + \im \bz)^{\ot t}),\mathrm{Im}(\mW)\ra)^2 $}, where  $a(\x,\bz) = \la \Real((\x + \im \bz)^{\ot t}),\Real(\mW)\ra$ and $b(\x,\bz) = \la \Real((\x + \im \bz)^{\ot t}),\mathrm{Im}(\mW)\ra$ are both degree-$t$ polynomials. 
    Therefore, by Gaussian Hypercontractivity (\Cref{fct:gaussian-hypercontractivity}), 
    we have that
\new{
\begin{align*}
    \|a(\x,\bz)^2\|_{L_2(\calN)}^{1/2} = \bigg(\E_{\x,\bz\sim\calN(\vec 0,\mI_d)}
    [a(\x,\bz)^4]\bigg)^{1/4}\leq 3^{t/2}\bigg(\E_{\x,\bz\sim\calN(\vec 0,\mI_d)}
    [a(\x,\bz)^2]\bigg)^{1/2} = 3^{t/2}\|a(\x,\bz)\|_{L_2(\calN)} \;,
\end{align*}
and a similar inequality holds for $b(\x,\bz)$.
Therefore, we obtain:
\begin{align*}
    \bigg(\E_{\x,\bz\sim\calN(\vec 0,\mI_d)}\bigg[\bigg|\bigg\la\Real((\x + \im \bz)^{\ot t}), \mW\bigg\ra\bigg|^4\bigg]\bigg)^{{1/2}}
    &=\bigg(\E_{\x,\bz\sim\calN(\vec 0,\mI_d)}\bigg[\bigg(a(\x,\bz)^2 + b(\x,\bz)^2\bigg)^2\bigg]\bigg)^{{1/2}} \\
    &= \|a(\x,\bz)^2 + b(\x,\bz)^2\|_{L_2(\calN)}\\
    &\leq \|a(\x,\bz)^2\|_{L_2(\calN)} + \|b(\x,\bz)^2\|_{L_2(\calN)}\\
    &\leq 3^t(\|a(\x,\bz)\|_{L_2(\calN)}^2 + \|b(\x,\bz)\|_{L_2(\calN)}^2) \;. 
\end{align*}
By  \cite[Lemma 4.3]{diakonikolas2025implicit}, we have that 
\begin{align*}
    \|a(\x,\bz)\|_{L_2(\calN)}^2&=\E_{\x,\bz\sim\calN(\vec 0,\mI_d)}\bigg[\bigg\la \Real((\x + \im \bz)^{\ot t}),\Real(\mW)\bigg\ra^2\bigg]\\
    &\leq \sup_{\|\wt{\mW}\|_F = 1}\E_{\x,\bz\sim\calN(\vec 0,\mI_d)}\bigg[\bigg|\bigg\la\Real((\x + \im \bz)^{\ot t}), \wt{\mW}\bigg\ra\bigg|^2\bigg]\leq t! 2^{O(t)} \;.
\end{align*}
The same inequality also holds for $b(\x,\bz)$. Thus we conclude that:} 
   \begin{align*}
   \bigg(\E_{\x,\bz\sim\calN(\vec 0,\mI_d)}\bigg[\bigg|\bigg\la\Real((\x + \im \bz)^{\ot t}), \mW\bigg\ra\bigg|^4\bigg]\bigg)^{{1/2}} &\leq 3^t\E_{\x,\bz\sim\calN(\vec 0,\mI_d)}\bigg[\bigg|\bigg\la\Real((\x + \im \bz)^{\ot t}), \mW\bigg\ra\bigg|^2\bigg]\\
   &\leq t! \, 3^t 2^{O(t)} = t!\, 2^{O(t)}\;. \end{align*}
\end{proof}
    Therefore, by plugging   $\E_{\x\sim\calN(\vec 0,\mI_d)}[(\sum_{i=1}^k \beta_i\sigma(\w_i\cdot\x))^4]\lesssim k^4$ back into \eqref{eq:variance-bound-0}, we obtain:
\begin{align*}
       \sup_{\|\mW\|_F = 1}\E_{(\x,y,\bz)\sim\D(\cI)}\bigg[\bigg|\bigg\la \frac{y}{c_t \sqrt{t!}}\Real((\x + \im \bz)^{\ot t}), \mW\bigg\ra\bigg|^2\bigg]\lesssim \frac{k^2{\new{2^{O(t)}} t!}}{\wh{\sigma}_t^2 t!} = \frac{k^2{\new{2^{O(t)}}}}{\wh{\sigma}_t^2 }\;. \end{align*}
This completes the proof of \Cref{lem:tensor-variance-S}. 
\end{proof}

Combining \Cref{lem:tensor-variance-S} with \Cref{thm:implict tensor adapted} immediately implies that when $p = O(\log(k/\lambda))$ (as was set in \nameref{alg:implicit-tensor-approximate-matrix}), using $N_0 = \poly(d,k,1/\tau)2^{O(p)} = \poly(d,k/\lambda,1/\tau)$ samples from $\D$, the pseudo-projection $\wh{\Phi}_p$ maps a tensor $O(\lambda)$-close to the space of the $\lambda$-heavy hidden weight tensors $\w_i^{\ot p}$ (i.e., $\beta_i\geq \lambda$) with probability at least $1 - \tau$.
The algorithm \nameref{alg:implicit-tensor-computation-adapted} runs efficiently, in time $\poly(d,p,N_0)$.
In short, we can efficiently approximate the compressor $\Phi_p$ using $\wh{\Phi}_p$ to reduce the problem dimension. 
Recalling from \eqref{eq:intuition-implicit-2} that the target matrix equals
$$\mM_{2p+2}(\bu) = \bigg\la\E_{\cI\sim\D(\cI)}\bigg[(\Phi_p\ot\Phi_p\ot\mI_d\ot\mI_d)\bigg(\mS_{2p+2}(\cI)\bigg)\bigg], \Phi_p(\mF_p(\bu))\ot\Phi_p(\mF_p(\bar\bu))\bigg\ra\;,$$
an estimation of $\mM_{2p+2}(\bu)$ can therefore be constructed via:
\begin{align}\label{eq:intuition-implict-3}
    \wh{\mM}_{2p+2}(\bu) = \bigg\la\E_{\cI\sim\wh{\D}(\cI)}\bigg[(\wh{\Phi}_p\ot\wh{\Phi}_p\ot\mI_d\ot\mI_d)\bigg(\mS_{2p+2}(\cI)\bigg)\bigg], \wh{\Phi}_p(\mF_p(\bu))\ot\wh{\Phi}_p(\mF_p(\bar\bu))\bigg\ra,\tag{EstM}
\end{align}
where $\wh{\D}(\cI)$ is a fresh empirical distribution of $\D(\cI)$ (independent of the samples used to construct $\wh{\Phi}_p$).
The two tensors on the right-hand side of \eqref{eq:intuition-implict-3} are computed in the following procedures: \nameref{alg:CompressEmpiricalMomentTensor} and \nameref{alg:CompressFilterTensor}.
In \nameref{alg:implicit-tensor-approximate-matrix}, \Cref{line:matrix-estimation-after-pseudo-projection}, we compute the empirical matrix $\wh{\mM}_{2p+2}(\bu)$ displayed in \eqref{eq:intuition-implict-3}, using the tensors returned by \nameref{alg:CompressEmpiricalMomentTensor} and \nameref{alg:CompressFilterTensor}.

\begin{algorithm}[H]
\caption{$\mathsf{CompressEmpMoment}(p,N_0,\wh{\Phi}_p,\D(\cI))$}
\label{alg:CompressEmpiricalMomentTensor}
    \begin{algorithmic}[1]
    \State{\bf Input:} Moment degree $p$, number of samples $N_0$, tensor compressor $\wh{\Phi}_p$, sample access to $\D(\cI)$. 
    \State{\bf Output: } A tensor $\mA_p\in\R^{k\times k\times d\times d}$
    \medskip
        \State Draw $N_0$ fresh samples $\wh{\D}(\cI) = \{\cI^{(i)}\}_{i=1}^{N_0}$ from $\D(\cI)$.
        \State Construct STC circuit $\sfS_{2p+2}$, according to \Cref{lem:stcs}. \State Compute $(\wh{\Phi}_p\ot\wh{\Phi}_p\ot\mI_d\ot\mI_d)\mS_{2p+2}(\cI^{(i)})$, 
where $\mS_{2p+2}(\cI^{(i)})$ is the output of $\sfS_{2p+2}$, using \Cref{cor:efficiently-computing-tensor-computation}.
\State $\mA_p\gets \frac{1}{N_0}\sum_{i=1}^{N_0} (\wh{\Phi}_p\ot\wh{\Phi}_p\ot\mI_d\ot\mI_d)\mS_{2p+2}(\cI^{(i)})$. \State{\bf Return: }$\mA_p$.
    \end{algorithmic}
\end{algorithm}

\begin{algorithm}[H]
\caption{$\mathsf{CompressFilter}(p,\wh{\Phi}_p,\bu)$}
\label{alg:CompressFilterTensor}
    \begin{algorithmic}[1]
    \State{\bf Input:} Moment degree $p$, tensor compressor $\wh{\Phi}_p$, probing vector $\bu$. 
    \State{\bf Output: } A matrix $\mB_p(\bu)\in\R^{k\times k}$
    \medskip
    
\State Construct STC circuits $\sfF_p$ according to \Cref{lem:stcs}. \State Compute $\wh{\Phi}_p(\mF_p(\bu))$ and $\wh{\Phi}_p(\mF_p(\overline{\bu}))$, where $\mF_p(\bu)$ 
is the output of $\sfF_p$, using \Cref{cor:efficiently-computing-tensor-computation}. \State $\mB_p(\bu)\gets \wh{\Phi}_p(\mF_p(\bu))\ot\wh{\Phi}_p(\mF_p(\overline{\bu}))$
\State{\bf Return: }$\mB_p(\bu)$.
    \end{algorithmic}
\end{algorithm}

We next show that these procedures run efficiently.

\begin{lemma}[Efficient Computation of Tensor Projections]\label{cor:efficiently-computing-tensor-computation}
     Let $\sfS_t(\D(\cI))$ and $\sfF_p(\cU)$ be the STCs defined in \Cref{lem:stcs}, equations~\eqref{eq:stc-S} and~\eqref{eq:stc-F} respectively. Let $\Phi_p$ be a $(d,p,k)$-recursive pseudo-projection. Then, for any $\mF_p(\bu)\sim\sfF_p(\cU)$, $\wh{\Phi}_p(\mF_p(\bu))$ can be computed in time $\poly(d,k,p)$; and for any $\mS_{2p+2}(\cI)\sim\sfS_{2p+2}(\D(\cI))$, the tensor $(\wh{\Phi}_p\ot\wh{\Phi}_p\ot\mI_d\ot\mI_d)(\mS_{2p+2}(\cI))$ can be computed in time $\poly(d,k,p)$.
 \end{lemma}
 \begin{proof}
  \cite[Lemma 3.7]{diakonikolas2025implicit} states that given an order-$t$ recursive pseudo-projection $\Psi$ along with their defining pseudo-projections $\theta_j$, and an STC $\mS$ of order $t$, dimension $d$, and size $S$, the projection $\Psi(\mS)$ can be computed in time $\poly(d,t,S)$. This immediately implies that $\wh{\Phi}_p(\mF_p(\bu))$ can be computed in $\poly(d,p)$ time.
Furthermore, according to  \cite[Lemma 3.8]{diakonikolas2025implicit}, given recursive pseudo-projections $\Phi : (\mathbb{R}^d)^{\otimes t} \to \mathbb{R}^n$ and $\Psi : (\mathbb{R}^d)^{\otimes s} \to \mathbb{R}^m$ along with their defining pseudo-projections $\phi_i, \theta_j$, one can efficiently construct the recursive pseudo-projection $\Phi \otimes \Psi : (\mathbb{R}^d)^{\otimes (t+s)} \to \mathbb{R}^{nm}$ of size $\max(\operatorname{size}(\Phi), n \operatorname{size}(\Psi))$ in time $\poly(d,t,s,\max(\operatorname{size}(\Phi), n \operatorname{size}(\Psi)))$.
     Therefore, the recursive-pseudo projection $\wh{\Phi}_p\ot\wh{\Phi}_p:(\R^{d})^{\ot p}\ot(\R^{d})^{\ot p}\to\R^k\ot\R^k$ can be efficiently constructed.
     Since $\mI_d:\R^d\to\R^d$ is also a pseudo-projection, we can similarly efficiently construct the recursive pseudo-projection $(\wh{\Phi}_p\ot\wh{\Phi}_p)\ot\mI_d$, and then construct $((\wh{\Phi}_p\ot\wh{\Phi}_p)\ot\mI_d)\ot\mI_d$. The process takes time $\poly(d,k,p)$.
     After constructing the pseudo-projection $\wh{\Phi}_p\ot\wh{\Phi}_p\ot\mI_d\ot\mI_d$, since $\mS_{2p+2}(\cI)$ is generated from a size $O(p)$ STC, the projection $(\wh{\Phi}_p\ot\wh{\Phi}_t\ot\mI_d\ot\mI_d)\mS_{2p+2}(\cI)$ can also be computed in $\poly(d,k,p)$ time.
\end{proof}

\paragraph{Main Result of this Subsection: 
Matrix Estimation Using Implicit Tensor Computation.} 
What remains is to prove that the empirical matrix estimation in \eqref{eq:intuition-implict-3} is close to the target matrix $\mM_{2p+2}(\bu)$. 
However, we need one last technical piece, which is to bound the covariance of the tensor $\mF_p(\bu)$. 
This is the place where we crucially use the property of complex Gaussian tensors: in \Cref{lem:moment-bound-query-tensor-after-projection}, we show that the covariance of the complex Gaussian filter $\mF_p(\bu)$ is bounded above by $2^{O(p)}$, whereas in \Cref{lem:real-complex-gaussian-comparison} we prove that were we using a real Gaussian filter, the covariance would be much larger.

To see why a bound on the covariance of $\mF_p(\bu)$ is required, consider the 
difference between the empirical matrix and the target matrix:
\begin{align*}
    \wh{\mM}_{2p+2}(\bu) - \mM_{2p+2}(\bu)&=\bigg\la\E_{\cI\sim\wh{\D}(\cI)}\bigg[(\wh{\Phi}_p\ot\wh{\Phi}_p\ot\mI_d\ot\mI_d)\bigg(\mS_{2p+2}(\cI)\bigg)\bigg], \wh{\Phi}_p(\mF_p(\bu))\ot\wh{\Phi}_p(\mF_p(\bar\bu))\bigg\ra\\
    &\quad - \bigg\la\E_{\cI\sim\D(\cI)}\bigg[\mS_{2p+2}(\cI)\bigg], \mF_p(\bu)\ot\mF_p(\bar\bu)\bigg\ra\\
    &=\la\Delta_1, \wh{\Phi}_p(\mF_p(\bu))\ot\wh{\Phi}_p(\mF_p(\bar\bu))\ra + \la\Delta_2, \mF_p(\bu)\ra + \la\Delta_2, \mF_p(\bar\bu)\ra \;,
\end{align*}
where $\Delta_1$, $\Delta_2$, and $\Delta_3$ are error terms (see, e.g., \eqref{eq:implicit-tensor-matrix-estimte-error-1} in the proof of \Cref{prop:estimate-matrix-via-implicit-tensor}). 
Even though $\Delta_1$, $\Delta_2$, and $\Delta_3$ can be tiny, if the spectral norm of $\mF_p(\bu)$ or $\wh{\Phi}_p(\mF_p(\bu))\ot\wh{\Phi}_p(\mF_p(\bar\bu))$ is large, the difference displayed above is not necessarily small.
To argue that $\mF_p(\bu)$ and $\wh{\Phi}_p(\mF_p(\bu))\ot\wh{\Phi}_p(\mF_p(\bar\bu))$ have small spectral norm with large probability, we bound their covariance and then apply standard concentration inequalities.

\begin{lemma}\label{lem:moment-bound-query-tensor-after-projection}
Let $\bu \sim\cU, \bu = \bp + \im\bq$, $\bp,\bq\sim\calN(\vec 0, \mI_d)$, and let $\mF_p(\bu)$ be the output of the STC $\sfF_p(\cU)$ defined in \eqref{eq:stc-F} in \Cref{lem:stcs}.
    Then, we have $\E_{\bu\sim\cU}[|\la \mF_p(\bu),\mW\ra|^2]\leq 2^{O(p)}\|\mW\|_F^2,\;\E_{\bu\sim\cU}[|\la \mF_p(\bu),\mW\ra|^4]\leq 2^{O(p)} \|\mW\|_F^4$.
In addition, let $\wh{\Phi}_p$ be the complex extension of a real-valued pseudo-projection from $(\mathbb{R}^d)^{\ot p}$ to $\mathbb{R}^k$. Then:
\begin{align*}
            \E_{\bu\sim\cU}[\|\wh{\Phi}_p(\mF_p(\bu))\otimes \wh{\Phi}_p(\mF_p(\bar{\bu}))\|_F^2]\leq k^2 2^{O(p)}.
        \end{align*}
\end{lemma}
\begin{proof}

Let $Z(\bu)= \la \mF_p(\bu),\mW\ra$. The random variables $Z(\bu)$ and $\im Z(\bu)$ have the same distribution; to see that, note that for any $\theta\in \R$, we have $ \exp(\im \theta)\vec u= \cos\theta \vec p - \sin\theta \vec q + \im (\cos\theta\vec q + \sin\theta \vec p)$ and both the real part and imaginary part are independent standard normal from orthogonality. Therefore, $Z(\exp(-\im \pi/(2p))\bu )=\im Z(\bu)$ and because $Z(\exp(-\im \pi/(2p))\bu )$ has the same distribution as $Z(\bu)$, it holds that $Z(\bu)$ and $\im(Z(\bu))$ have the same distribution.
This further indicates that $\Real(Z(\bu))$ and $\mathrm{Im}(Z(\bu))$ have the same distribution.
Similarly, $\Real(\mF_p(\bu))$ and $\mathrm{Im}(\mF_p(\bu))$ also have the same distribution.
Taking imaginary parts, we conclude that $\Real(Z(\bu))$
and $\mathrm{Im}(Z(\bu))$ have the same distribution.
The same rotation also shows that
$\Real(\mF_p(\bu))$ and $\mathrm{Im}(\mF_p(\bu))$
have the same distribution.
\new{
By \cite[Lemma~4.3]{diakonikolas2025implicit}, for every
real tensor $\mathbf T$ of order $p$, we have
\begin{align*}
    \E_{\bu\sim\cU}\left[
        \la \Real(\mF_p(\bu)),\mathbf T\ra^2
    \right]
    &=
    \frac{1}{p!}
    \E_{\bu\sim\cU}\left[
        \la \Real(\bu^{\ot p}),\mathbf T\ra^2
    \right] \leq
    2^{O(p)}\|\mathbf T\|_F^2.
\end{align*}
Since $\Real(\mF_p(\bu))$ and $\mathrm{Im}(\mF_p(\bu))$
have the same distribution, the same bound holds for
$\mathrm{Im}(\mF_p(\bu))$.

Since the inner product is conjugate-linear in its first
argument, we have
\[
    \Real(Z(\bu))
    =
    \la \Real(\mF_p(\bu)),\Real(\mW)\ra
    +
    \la \mathrm{Im}(\mF_p(\bu)),\mathrm{Im}(\mW)\ra.
\]
Using $(a+b)^2\leq 2a^2+2b^2$, it follows that
\begin{align*}
    \E_{\bu\sim\cU}\left[\Real(Z(\bu))^2\right]
    &\leq
    2\E_{\bu\sim\cU}\left[
        \la \Real(\mF_p(\bu)),\Real(\mW)\ra^2
    \right] +
    2\E_{\bu\sim\cU}\left[
        \la \mathrm{Im}(\mF_p(\bu)),
            \mathrm{Im}(\mW)\ra^2
    \right]\\
    &\leq
    2^{O(p)}
    \left(
        \|\Real(\mW)\|_F^2+
        \|\mathrm{Im}(\mW)\|_F^2
    \right) =
    2^{O(p)}\|\mW\|_F^2.
\end{align*}
The inequality also holds for $\mathrm{Im}(Z(\bu))$. Therefore,
\begin{align*}
    \E_{\bu\sim\cU}\left[
        |\la \mF_p(\bu),\mW\ra|^2
    \right]
    &=
    \E_{\bu\sim\cU}\left[
        \Real(Z(\bu))^2+\mathrm{Im}(Z(\bu))^2
    \right]=
    2\E_{\bu\sim\cU}\left[\Real(Z(\bu))^2\right] \leq
    2^{O(p)}\|\mW\|_F^2.
\end{align*}

For the fourth moment, we apply
\Cref{claim:complex-gaussian-tensor-momemt-bounds}.
By the definition of 
$\mF_p(\bu)=\bu^{\ot p}/\sqrt{p!}$, for every real tensor
$\mathbf T$ of order $p$, we have
\begin{align*}
    \E_{\bu\sim\cU}\left[
        |\la \Real(\mF_p(\bu)),\mathbf T\ra|^4
    \right]
    &=
    \frac{1}{(p!)^2}
    \E_{\bu\sim\cU}\left[
        |\la \Real(\bu^{\ot p}),\mathbf T\ra|^4
    \right] \leq
    2^{O(p)}\|\mathbf T\|_F^4.
\end{align*}
Since $\Real(\mF_p(\bu))$ and $\mathrm{Im}(\mF_p(\bu))$
have the same distribution, the same bound holds with
$\Real(\mF_p(\bu))$ replaced by $\mathrm{Im}(\mF_p(\bu))$.
Recalling that $\Real(Z(\bu))
    =
    \la \Real(\mF_p(\bu)),\Real(\mW)\ra
    +
    \la \mathrm{Im}(\mF_p(\bu)),\mathrm{Im}(\mW)\ra,$
and using $(a+b)^4\leq 8(a^4+b^4)$, we obtain
\begin{align*}
    \E_{\bu\sim\cU}[\Real(Z(\bu))^4]
    &\leq
    8\E_{\bu\sim\cU}[
        \la \Real(\mF_p(\bu)),\Real(\mW)\ra^4
    ] +
    8\E_{\bu\sim\cU}[
        \la \mathrm{Im}(\mF_p(\bu)),
            \mathrm{Im}(\mW)\ra^4
    ]\\
    &\leq
    2^{O(p)}
    \left(
        \|\Real(\mW)\|_F^4+
        \|\mathrm{Im}(\mW)\|_F^4
    \right) \leq
    2^{O(p)}\|\mW\|_F^4,
\end{align*}
where the last inequality uses
\[
    \|\Real(\mW)\|_F^4+\|\mathrm{Im}(\mW)\|_F^4
    \leq
    \left(
        \|\Real(\mW)\|_F^2+\|\mathrm{Im}(\mW)\|_F^2
    \right)^2
    =
    \|\mW\|_F^4.
\]
Finally, using $(a^2+b^2)^2\leq 2a^4+2b^4$ and the fact
that $\Real(Z(\bu))$ and $\mathrm{Im}(Z(\bu))$ have the
same distribution, we conclude that
\begin{align}\label{eq:4th-moment-b}
    \E_{\bu\sim\cU}[
        |\la \mF_p(\bu),\mW\ra|^4
    ]
    &=
    \E_{\bu\sim\cU}[
        (
            \Real(Z(\bu))^2+\mathrm{Im}(Z(\bu))^2
        )^2
    ]\leq
    2\E_{\bu\sim\cU}[
        \Real(Z(\bu))^4+\mathrm{Im}(Z(\bu))^4
    ] \nonumber\\
    &=
    4\E_{\bu\sim\cU}[\Real(Z(\bu))^4]\leq
    2^{O(p)}\|\mW\|_F^4.
\end{align}
}

        For the second part of the statement, note that 
        since $\wh{\Phi}_p$ is a map with real entries, we have $\wh{\Phi}_p(\mF_p(\bar{\bu})) = \overline{\wh{\Phi}_p(\mF_p(\bu))}$, therefore, it holds that 
        \begin{align*}
            \E_{\bu\sim\cU}[\|\wh{\Phi}_p(\mF_p(\bu))\otimes \wh{\Phi}_p(\mF_p(\bar{\bu}))\|_F^2] &\leq \E_{\bu\sim\cU}[\|\wh{\Phi}_p(\mF_p(\bu))\|_F^2\| \wh{\Phi}_p(\mF_p(\bar{\bu}))\|_F^2] \\
            &= \E_{\bu\sim\cU}[\|\wh{\Phi}_p(\mF_p(\bu))\|_F^2\| \overline{\wh{\Phi}_p(\mF_p(\bu))}\|_F^2] = \E_{\bu\sim\cU}[\|\wh{\Phi}_p(\mF_p(\bu))\|_F^4].
        \end{align*}
               
Recall that $\wh{\Phi}_p$ maps $\mF_p(\bu)$ to a vector in $\mathbb{C}^k$, thus we have $\|\wh{\Phi}_p(\mF_p(\bu))\|_2^2 = \sum_{i=1}^k |\la\wh{\Phi}_p(\mF_p(\bu)), \e_i\ra|^2$. Furthermore, by Jensen's inequality, it holds:
        \begin{align*}
            \E_{\bu}[\|\wh{\Phi}_p(\mF_p(\bu))\|_2^4 ]=\E_{\bu}\bigg[ \bigg( \sum_{i=1}^k |\la\wh{\Phi}_p(\mF_p(\bu)), \e_i\ra|^2\bigg)^2\bigg]&\leq \E_{\bu}\bigg[k\sum_{i=1}^k |\la\wh{\Phi}_p(\mF_p(\bu)), \e_i\ra|^4\bigg]\\
            &= k\sum_{i=1}^k \E_{\bu}[|\la\mF_p(\bu), \wh{\Phi}_p^\top(\e_i)\ra|^4]\leq k^2 2^{O(p)}\;,
        \end{align*}
        where for the last inequality we use that $\wh{\Phi}_p \wh{\Phi}_p^{\top}=\vec I_k$ (and therefore $\|\wh{\Phi}_p^{\top}\vec e_i\|_F=1$) and \eqref{eq:4th-moment-b}.
\end{proof}

For completeness, in the following lemma, we demonstrate the crucial difference between complex and real Gaussian tensors in terms of their covariance bounds, which justifies the use of complex Gaussians in our approach. 

\begin{lemma}[Real vs Complex Gaussian Tensor Powers]
\label{lem:real-complex-gaussian-comparison}
Let $d,p\ge1$ be integers, let
$\bp,\bq\sim \calN(0,\mI_d)$, and put
$\bu=\bp+\im\bq$. Define
\[
 \mF_p^{\mathrm R}(\bp)
 =\frac{\bp^{\otimes p}}{\sqrt{(2p-1)!!}},
 \qquad
 \mF_p^{\mathrm C}(\bu)
 =\frac{\bu^{\otimes p}}{\sqrt{2^pp!}}=\mF_p(\bu)/\sqrt{2^p}.
\]
Both normalizations give, for every real unit vector $\bw$,
\[
 \E\big[|\langle\mF_p^{\mathrm R}(\bp),\bw^{\otimes p}\rangle|^2\big]
 =\E[|\langle\mF_p^{\mathrm C}(\bu),\bw^{\otimes p}\rangle|^2]
 =1.
\]
Their uncentered second moments in general tensor directions
differ as follows.
\begin{enumerate}
 \item For every complex order-$p$ tensor $\mR$,
 \begin{equation*}
  \E[|\langle\mF_p(\bu),\mR\rangle|^2]
  =\|\mathrm{Sym}(\mR)\|_F^2
  \le\|\mR\|_F^2.
 \end{equation*}
\item  For every even $p\ge2$, the real tensor
 $\mR_p=(\mI_d/\sqrt d)^{\otimes p/2}$ has $\|\mR_p\|_F=1$ and
 \begin{equation*}
  \E[\big|\langle\mF_p^{\mathrm R}(\bp),\mR_p\rangle\big|^2]
  =\frac{\prod_{j=0}^{p-1}(d+2j)}
         {d^{p/2}(2p-1)!!}
  \ge\frac{d^{p/2}}{(2p)^p}.
 \end{equation*}
\end{enumerate}
\end{lemma}

\begin{proof}

Let $p$ be even. Multiplicativity of the Frobenius norm
under tensor products gives $\|\mR_p\|_F=1$, and direct contraction
gives
\[
 \langle\bp^{\otimes p},\mR_p\rangle
 =\frac{\|\bp\|_2^p}{d^{p/4}}.
\]
The radial Gaussian density is proportional to
$r^{d-1}e^{-r^2/2}$. Integration by parts gives
$\E\|\bp\|_2^{2s}=(d+2s-2)\E\|\bp\|_2^{2s-2}$ for $s\ge1$,
so $\E\|\bp\|_2^{2p}=\prod_{j=0}^{p-1}(d+2j)$.
Squaring the contraction and dividing by $(2p-1)!!$ proves
the equality.
Its lower bound follows from $d+2j\ge d$ and
$(2p-1)!!\le(2p)^p$.
\end{proof}

We are now ready to prove the main result of this subsection, 
\Cref{prop:estimate-matrix-via-implicit-tensor}.

\MatrixEstImplicitTensor*

\begin{proof}[Proof of \Cref{prop:estimate-matrix-via-implicit-tensor}]
Let $\D(\cI)$ be the distribution of $\cI = (\x,y,\bz)$ where $(\x,y)\sim\D$ and $\bz\sim\calN(\vec 0,\mI_d)$.
By \Cref{thm:implict tensor adapted}, the pseudo-projection $\wh{\Phi}_p$ returned by \nameref{alg:implicit-tensor-computation-adapted} (see \Cref{line:obtain-recursive pseudo-projection} in \nameref{alg:implicit-tensor-approximate-matrix}) satisfies: with probability $1 - \tau_0$ over the samples drawn in \nameref{alg:implicit-tensor-computation-adapted}, the following event holds
\begin{align*}
            &\evt_0 := \bigg\{\text{$\wh{\Phi}_p$ satisfies }\|\wh{\Phi}_p^\top\wh{\Phi}_p(\mM_p) - \mM_p\|_F\leq kp\sqrt{\zeta},\,\, \|\wh{\Phi}_p^\top\wh{\Phi}_p({\w}^{\otimes p}_i) - \w_i^{\otimes p}\|_F\leq  p\sqrt{\frac{\zeta}{\beta_i}}, \forall i\in[k]\bigg\}.
        \end{align*}
        where $\zeta = O(pdkV_\sfS/\sqrt{N_0\tau_0})$.
        By \Cref{lem:tensor-variance-S}, we have that $V_\sfS\leq k^2 2^{O(p)}/\wh{\sigma}_{2p+2}^2$.
    From now on, we condition on this event $\evt_0$.

In \nameref{alg:CompressEmpiricalMomentTensor}, we draw $N_0$ new samples from $\D(\cI)$ and construct a new empirical distribution, which we denote by $\wh{\D}_{\mathrm{new}}(\cI)$. 
Then, using the pseudo-projection $\wh{\Phi}_p$ returned from \nameref{alg:implicit-tensor-computation-adapted}, \nameref{alg:CompressEmpiricalMomentTensor} constructs 
\begin{align*}
    \mA_p = (\wh{\Phi}_p\ot\wh{\Phi}_p\ot\mI_d\ot\mI_d)\wh{\mS}_{2p+2},\quad \wh{\mS}_{2p+2} = \E_{\cI\sim\wh{\D}_{\mathrm{new}}(\cI)}[\mS_{2p+2}(\cI)].
\end{align*}
Given a probing vector $\bu$ from $\wh{\cU}$ (of $N_\bu$ probing vectors in total), \nameref{alg:CompressFilterTensor} generates:
\begin{align*}
    \mB_p(\bu) = \wh{\Phi}_p(\mF_p(\bu))\ot\wh{\Phi}_p(\mF_p(\bar{\bu})).
\end{align*}
Then in \Cref{line:matrix-estimation-after-pseudo-projection}, we compute the empirical estimation \eqref{eq:intuition-implict-3}
    \begin{equation*}
        \wh{\mM}_{2p+2}(\bu) = \la\mA_p,\mB_p(\bu)\ra = \la(\wh{\Phi}_p\ot\wh{\Phi}_p\ot\mI_d\ot\mI_d)\wh{\mS}_{2p+2}, \wh{\Phi}_p(\mF_p(\bu))\ot\wh{\Phi}_p(\mF_p(\bar{\bu}))\ra.
    \end{equation*}
Let $a,b\in[d]$, and define
    \begin{align*}
        (\wh{\mM}_{2p+2}(\bu))_{ab} &=\bigg\la(\wh{\Phi}_p\ot\wh{\Phi}_p\ot\mI_d\ot\mI_d)(\wh{\mS}_{2p+2}), \wh{\Phi}_p(\mF_p(\bu))\ot\wh{\Phi}_p(\mF_p(\bar{\bu}))\ot\e_a\ot\e_b\bigg\ra\;, \end{align*}
    which is the $(a,b)$-entry of the empirical matrix.

    \begin{lemma}\label{lem:implicit-tensor-empirical-last-two-slices-close}
        Conditioned on $\evt_0$, with probability at least $1 - 4\tau_0$ (over the fresh $N_{0}$ samples $\wh{\cD}_{\mathrm{new}}(\cI)$ and $\bu$), for any indices $a,b\in[d]$, it holds:
    \begin{align*}
        |(\wh{\mM}_{2p+2}(\bu))_{ab} - (\mM_{2p+2}(\bu))_{ab}|\leq \frac{d k^3 p^2 2^{O(p)}}{|\wh{\sigma}_{2p+2}| \tau_0^{3/2} N_0^{1/4}}. \end{align*}
    \end{lemma}
   
\begin{proof}
\newcommand{\abbP}{\Psi_{2p}}
\newcommand{\abbF}{\mB_p(\bu)} For this proof, we use the following notation:
\[
\abbP=\wh{\Phi}_p\ot\wh{\Phi}_p\ot\mI_d\ot\mI_d\qquad \abbF=\wh{\Phi}_p(\mF_p(\bu))\ot\wh{\Phi}_p(\mF_p(\bar{\bu}))\qquad 
\]

    First, note that the difference between $(\wh{\mM}_{2p+2}(\bu))_{ab}$ and $(\mM_{2p+2}(\bu))_{ab}$ can be decomposed as:
    \begin{align}\label{eq:implicit-tensor-matrix-estimte-error-1}
&\quad (\wh{\mM}_{2p+2}(\bu))_{ab} - (\mM_{2p+2}(\bu))_{ab} \nonumber\\
        &=\bigg\la\abbP(\wh{\mS}_{2p+2}), \abbF\ot\e_a\ot\e_b\bigg\ra \nonumber - \bigg\la\sum_{i=1}\beta_i\w_i^{\ot (2p+2)}, \mF_p(\bu)\ot\mF_p(\bar{\bu})\ot\e_a\ot\e_b\bigg\ra \nonumber\\
        &=\underbrace{\bigg\la\abbP\bigg(\wh{\mS}_{2p+2} - \E_{\cI\sim\D(\cI)}[\mS_{2p+2}(\cI)]\bigg), \abbF\ot\e_a\ot\e_b\bigg\ra}_{I_1}\nonumber\\
        &\quad\quad +\underbrace{\bigg\la\abbP\bigg(\E_{\cI\sim\D(\cI)}[\mS_{2p+2}(\cI)] - \sum_{i=1}\beta_i\w_i^{\ot (2p+2)}\bigg), \abbF\ot\e_a\ot\e_b\bigg\ra}_{I_2}\nonumber\\
        &\quad\quad + \underbrace{\bigg\la\abbP\bigg(\sum_{i=1}\beta_i\w_i^{\ot (2p+2)}\bigg), \abbF\ot\e_a\ot\e_b\bigg\ra}_{I_3}  - \underbrace{\bigg\la\sum_{i=1}\beta_i\w_i^{\ot (2p+2)}, \mF_p(\bu)\ot\mF_p(\bar{\bu})\ot\e_a\ot\e_b\bigg\ra}_{I_4}.
    \end{align}
    
    We bound each term respectively.

    \begin{claim}
        With probability at least $1 - \tau_\bu$ over the probing vector $\bu$, and probability at least $1 - \tau_{2p+2}$ over the samples from $\wh{\cD}_{\mathrm{new}}(\cI)$, it holds:
    \begin{align*}
        |I_1|\leq \frac{d k^2 2^{O(p)}}{|\wh{\sigma}_{2p+2}|\sqrt{\tau_{2p+2}\tau_\bu N_{0}}}. \end{align*}
    \end{claim}
    \begin{proof}
Let $\vec W=\wh{\mS}_{2p+2} - \E_{\cI\sim\D(\cI)}[\mS_{2p+2}(\cI)]$.
Note that $\abbP(\mW)\in\R^{k\times k\times d \times d},$ we have 
\begin{align*}
    \E[\|\abbP(\mW)\|_F^2] &= \sum_{i_1,i_2\in[k], i_3,i_4\in[d]}\E\bigg[\bigg|\bigg\la \mW,\, \abbP^\top(\e_{i_1}\ot\e_{i_2}\ot\e_{i_3}\ot\e_{i_4})\bigg\ra\bigg|^2\bigg].
\end{align*}
Note that by \Cref{pseudo-prop-frob-equl} from \Cref{fact:pseudo-projection-tensor-facts}, it holds $\|\abbP^\top(\e_{i_1}\ot\e_{i_2}\ot\e_{i_3}\ot\e_{i_4})\|_F = 1$.
Thus, by \Cref{lem:tensor-variance-S}, for any $i_1,i_2\in[k], i_3,i_4\in[d]$, we have:
\begin{align*}
    &\quad \E\bigg[\bigg|\bigg\la \mW, \;\abbP^\top(\e_{i_1}\ot\e_{i_2}\ot\e_{i_3}\ot\e_{i_4})\bigg\ra\bigg|^2\bigg]\\
    &\leq \frac{1}{N_{0}}\E_{\cI\sim\D(\cI)}\bigg[\bigg|\bigg\la \mS_{2p+2}(\cI) - \E_{\cI\sim\D(\cI)}[\mS_{2p+2}(\cI)], \;\abbP^\top(\e_{i_1}\ot\e_{i_2}\ot\e_{i_3}\ot\e_{i_4})\bigg\ra\bigg|^2\bigg]\\
    &\leq\frac{V_{\sfS}}{N_{0}}\leq \frac{k^2 2^{O(p)}}{\wh{\sigma}_{2p+2}^2 N_{0}}. \end{align*}
This implies $\E[\|\abbP(\mW)\|_F^2]\leq d^2 k^4 2^{O(p)}/(\wh{\sigma}_{2p+2}^2 N_{0})$.
Therefore, by Chebyshev's inequality, with probability at least $1 - \tau_{2p+2}$ over the samples $\wh{\D}_{\mathrm{new}}(\cI)$:
\begin{align*}
    \|\abbP(\mW)\|_F\leq \sqrt{\frac{k^4 d^2 2^{O(p)}}{\wh{\sigma}_{2p+2}^2 \tau_{2p+2} N_{0}}} =\frac{k^2 d  2^{O(p)}}{|\wh{\sigma}_{2p+2}|\sqrt{\tau_{2p+2}N_{0}}}. \end{align*}
Let us denote this event by $\evt_{\|\mW\|_F}$, and condition on this event:
\begin{align}\label{eq:prob-0}
    \pr_{\bu,\wh{\D}_{\mathrm{new}}(\cI)}[|I_1|\leq t]\geq \pr_{\bu,\wh{\D}_{\mathrm{new}}(\cI)}[|I_1|\leq t, \evt_{\|\mW\|_F}]\geq (1 - \tau_{2p+2})\bigg(1 - \pr_{\bu}[|I_1|\geq t \mid \evt_{\|\mW\|_F}]\bigg).
\end{align}
Using Chebyshev's inequality again with respect to the probing vector $\bu$, we have
\begin{align}\label{eq:prob-2}
    \pr_{\bu}[|I_1|\geq t \mid \evt_{\|\mW\|_F}] &=\pr_{\bu}[|\la\abbP(\mW),\, \abbF\ot\e_a\ot\e_b\ra|\geq t \mid \evt_{\|\mW\|_F}]\nonumber\\
    &\leq \frac{\E[|\la\abbP(\mW), \,\abbF\ot\e_a\ot\e_b\ra|^2 \mid \evt_{\|\mW\|_F}]}{t^2}
\end{align}
By \Cref{lem:moment-bound-query-tensor-after-projection}, it holds $\E[\|\abbF\ot\e_a\ot\e_b\|_F^2]\leq \E[\|\abbF\|_F^2]\leq k^2 2^{O(p)}$. Using the fact that $\la\mW,\abbF\ra\leq\|\mW\|_F\|\abbF\|_F$, we obtain
\begin{align}\label{eq:prob-3}
    &\quad \E\bigg[\bigg|\bigg\la\abbP(\mW), \,\abbF\ot\e_a\ot\e_b\bigg\ra\bigg|^2 \,\bigg|\, \evt_{\|\mW\|_F}\bigg]
    \nonumber\\
    &\leq\E\bigg[\bigg\|\abbP(\mW)\bigg\|_F^2\,\bigg\|\abbF\ot\e_a\ot\e_b\bigg\|_F^2 \,\bigg|\,\evt_{\|\mW\|_F}\bigg] \nonumber\\
    &\leq \frac{d^2 k^4 2^{O(p)}}{\wh{\sigma}_{2p+2}^2\tau_{2p+2} N_0}\E\bigg[\bigg\|\abbF\ot\e_a\ot\e_b\bigg\|_F^2 \bigg]\leq \frac{d^2 k^6  2^{O(p)}}{\wh{\sigma}_{2p+2}^2 \tau_{2p+2} N_{0}}.
\end{align}
Therefore, combining \eqref{eq:prob-3} and \eqref{eq:prob-2} and plugging back into \eqref{eq:prob-0}, we obtain that with probability at least $(1 - \tau_{2p+2})(1 - \tau_\bu)\geq 1 - \tau_\bu-\tau_{2p+2}$, it holds
\begin{align*}
    |I_1|\leq \frac{d k^3 2^{O(p)}}{|\wh{\sigma}_{2p+2}|\sqrt{\tau_{2p+2}\tau_\bu N_{0}}}.
\end{align*}
\end{proof}

For the term $I_2$, since $\E_{\cI\sim\D(\cI)}[\mS_{2p+2}(\cI)] = \sum_{i=1}\beta_i\w_i^{\ot (2p+2)}$ by \Cref{lem:tensor-variance-S}, we have  $I_2 = 0$.

    It remains to bound the last two terms in \eqref{eq:implicit-tensor-matrix-estimte-error-1}.
    \begin{claim}
        With probability at least $1 - \tau_\bu$ we have:
    \begin{align*}
        |I_3-I_4|\leq \sqrt{\frac{d k^5 p^3 2^{O(p)}}{\wh{\sigma}_{2p+2}^2 \tau_\bu\sqrt{N_0 \tau_0}}}. \end{align*}
    \end{claim}
    \begin{proof}
    First, recall that by the property of pseudo-projectors we know that $\wh{\Phi}_p$ is a composition of a isometry $\mR_p$ and a projection operator $\mP_p$ that maps $(\R^d)^{\ot p}$ to a subspace $\sU_p$: $\wh{\Phi}_p = \mR_p\mP_p$, therefore, for two tensors $\mT$ and $\mW$, it holds $\la\wh{\Phi}_p(\mT),\wh{\Phi}_p(\mW)\ra = \la\mR_p(\mP_p(\mT)), \mR_p(\mP_p(\mW))\ra = \la \mP_p^\top\mP_p(\mT), \mW\ra$ since $\mR_p^\top\mR_p = \mI_k$ as $\mR_p$ is an isometry.
    
    Note that the linear operator $\abbP$ acting on $\sum_{i=1}\beta_i\w_i^{\ot (2p+2)}$ is simply: 
    \begin{align*}
        (\wh{\Phi}_p\ot\wh{\Phi}_p\ot\mI_d\ot\mI_d)\bigg(\sum_{i=1}\beta_i\w_i^{\ot (2p+2)}\bigg) = \sum_{i=1}^k \beta_i\wh{\Phi}_p(\w_i^{\ot p})\ot\wh{\Phi}_p(\w_i^{\ot p})\ot\w_i\ot\w_i,
    \end{align*}
    and its tensor product with $\abbF\ot\e_a\ot\e_b$ is 
    \begin{align*}
        I_3&=\bigg\la(\wh{\Phi}_p\ot\wh{\Phi}_p\ot\mI_d\ot\mI_d)\bigg(\sum_{i=1}\beta_i\w_i^{\ot (2p+2)}\bigg),\abbF\ot\e_a\ot\e_b\bigg\ra \\
        &= \sum_{i=1}^k \beta_i \bigg\la \wh{\Phi}_p(\w_i^{\ot p})\ot\wh{\Phi}_p(\w_i^{\ot p})\ot\w_i\ot\w_i, \wh{\Phi}_p(\mF_p(\bu))\ot\wh{\Phi}_p(\mF_p(\bar{\bu}))\ot\e_a\ot\e_b\bigg\ra\\
        &=\sum_{i=1}^k \beta_i \bigg\la\wh{\Phi}_p(\w_i^{\ot p}), \wh{\Phi}_p(\mF_p(\bu))\bigg\ra\,\bigg\la \wh{\Phi}_p(\w_i^{\ot p}), \wh{\Phi}_p(\mF_p(\bar{\bu}))\bigg\ra \,(\w_i)_a(\w_i)_b\\
        &=\sum_{i=1}^k \beta_i \bigg\la\wh{\Phi}_p^\top\wh{\Phi}_p(\w_i^{\ot p}), \mF_p(\bu)\bigg\ra\,\bigg\la \wh{\Phi}_p^\top\wh{\Phi}_p(\w_i^{\ot p}), \mF_p(\bar{\bu})\bigg\ra \,(\w_i)_a(\w_i)_b
        \\&=\sum_{i=1}^k \beta_i \bigg|\la \wh{\Phi}_p^\top\wh{\Phi}_p(\w_i^{\ot p}), \mF_p(\bu)\ra\bigg|^2 \,(\w_i)_a(\w_i)_b \;,
    \end{align*}
   where we used that $\mF_p(\bar{\bu})=\overline{\mF_p(\bu)}$ and $(\w_i)_a$ is the $a$-th entry of the vector $\w_i$.  On the other hand, we have 
    $$I_4 = \bigg\la\sum_{i=1}\beta_i\w_i^{\ot (2p+2)}, \mF_p(\bu)\ot\mF_p(\bar{\bu})\ot\e_a\ot\e_b\bigg\ra = \sum_{i=1}^k \beta_i \Big|\la\w_i^{\ot p}, \mF_p(\bu)\ra\Big|^2(\w_i)_a(\w_i)_b.$$
    Using the complex inequality $||x|^2-|y|^2|\leq |x-y|(|x|+|y|)$, the difference between these terms equals:
    \begin{align*}
        |I_3 - I_4| &\leq \sum_{i=1}^k \beta_i \big|\la \wh{\Phi}_p^\top\wh{\Phi}_p(\w_i^{\ot p})-\w_i^{\ot p}, \mF_p(\bu)\ra\big| \Big(\big|\la \wh{\Phi}_p^\top\wh{\Phi}_p(\w_i^{\ot p}),\mF_p(\bu)\ra\big|+\big|\la\w_i^{\ot p},\mF_p(\bu) \ra\big| \Big) \,|(\w_i)_a||(\w_i)_b|\;.
    \end{align*}
    To bound the variance, we use \CS and the crude bound  $(X^2+Y^2)^2\leq 2(X^4+Y^4)$ to get 
    \begin{align*}
        \E_\bu[|I_3-I_4|^2]
        &\leq 2k\sum_{i=1}\beta_i^2\sqrt{\E_{\bu}\Big[\Big|\Big\la\wh{\Phi}_p^\top\wh{\Phi}_p(\w_i^{\ot p}) - \w_i^{\ot p}, \mF_p(\bu)\Big\ra\Big|^4\Big]\,\E_{\bu}\Big[\Big|\Big\la \wh{\Phi}_p^\top\wh{\Phi}_p(\w_i^{\ot p}), \mF_p(\bar{\bu})\Big\ra\Big|^4\Big]|}(\w_i)_a(\w_i)_b|^2\\
        &\quad \quad + 2k\sum_{i=1}^k \beta_i^2 \sqrt{\E_{\bu}\Big[\Big|\Big\la\w_i^{\ot p}, \mF_p(\bu)\Big\ra\Big|^4\Big]\,\E_{\bu}\Big[\Big|\Big\la \wh{\Phi}_p^\top\wh{\Phi}_p(\w_i^{\ot p}) - \w_i^{\ot p}, \mF_p(\bar{\bu})\Big\ra\Big|^4\Big]} \,|(\w_i)_a(\w_i)_b|^2.
    \end{align*}
    Using \Cref{lem:moment-bound-query-tensor-after-projection}, since $\|\wh{\Phi}_p^\top\wh{\Phi}_p(\w_i^{\ot p})\|_F\leq \|\w_i^{\ot p}\|_F = 1$ according to \Cref{fact:pseudo-projection-tensor-facts}, we further obtain:
    \begin{align*}
        \E_\bu[|I_3-I_4|^2]&\leq 2k\sum_{i=1}^k \beta_i^2 \bigg(\sqrt{2^{O(p)}\|\wh{\Phi}_p^\top\wh{\Phi}_p(\w_i^{\ot p}) - \w_i^{\ot p}\|_F^4\times 2^{O(p)}} \\
        &\quad \quad \quad + \sqrt{2^{O(p)}\times 2^{O(p)}\|\wh{\Phi}_p^\top\wh{\Phi}_p(\w_i^{\ot p}) - \w_i^{\ot p}\|_F^4}\bigg) |(\w_i)_a(\w_i)_b|^2\\
        &\leq 4k\sum_{i=1}^k 2^{O(p)} \beta_i^2   \|\wh{\Phi}_p^\top\wh{\Phi}_p(\w_i^{\ot p}) - \w_i^{\ot p}\|_F^2 .
    \end{align*}
    As we have conditioned on the event $\evt_0$,
we know that for every $i\in[k]$, it holds $\|\wh{\Phi}_p^\top\wh{\Phi}_p(\w_i^{\ot p}) - \w_i^{\ot p}\|_F\leq p\sqrt{\zeta/\beta_i}$, where $\zeta = O({pdkV_{\sfS}}/{\sqrt{N_0\tau_0}})$ and $V_{\sfS} \leq k^2 2^{O(p)}/\wh{\sigma}_{2p+2}^2$. 
Thus, as $\beta_i\in(0,1)$, we finally get that conditioned on $\evt_0$, it holds
\begin{align*}
         \E_\bu[|I_3-I_4|^2]&\leq 4k\sum_{i=1}^k 2^{O(p)} \beta_i^2 p^2\frac{\zeta}{\beta_i} \leq k^2 p^2 2^{O(p)}\zeta \lesssim \frac{d k^5 p^3 2^{O(p)}}{\wh{\sigma}_{2p+2}^2 \sqrt{N_0\tau_0}}. \end{align*}
    Then, applying Markov's inequality, with probability at least $1 - \tau_\bu$ we have:
    \begin{align*}
        |I_3-I_4|\leq \sqrt{\frac{d k^5 p^3 2^{O(p)}}{\wh{\sigma}_{2p+2}^2 \tau_\bu\sqrt{N_0 \tau_0}}} \;.
    \end{align*}
    \end{proof}

    In summary, after a union bound on these events, we have that 
with probability at least $1 - \tau_{2p+2}$ over the fresh samples from $\wh{\D}_{\mathrm{new}}$ and with probability at least $1 - 3\tau_\bu$ over the random probing vector $\bu$, we have
    \begin{align*}
        |(\wh{\mM}_{2p+2}(\bu))_{ab} - (\mM_{2p+2}(\bu))_{ab}|&\leq |I_1|+|I_2|+|I_3-I_4|\\
        &\leq \frac{d k^2 2^{O(p)}}{|\wh{\sigma}_{2p+2}|\sqrt{\tau_{2p+2}\tau_\bu N_{0}}}  +  \sqrt{\frac{d k^5 p^3 2^{O(p)}}{\wh{\sigma}_{2p+2}^2 \tau_\bu\sqrt{N_0 \tau_0}}}. \end{align*}
    After rearranging the terms, we have:
    \begin{align*}
        |(\wh{\mM}_{2p+2}(\bu))_{ab} - (\mM_{2p+2}(\bu))_{ab}|\lesssim \frac{d k^3 p^2 2^{O(p)}}{|\wh{\sigma}_{2p+2}|(\tau_{2p+2}\tau_\bu\tau_0)^{1/2}}\frac{1}{N_0^{1/4}},
    \end{align*}
    with probability at least $(1 - \tau_{2p+2})(1-3\tau_\bu)\geq 1 - 3\tau_\bu - \tau_{2p+2}$ (by the independence between $\D$ and $\cU$), conditioned on the event $\evt_0$. Letting $\tau_0 = \tau_{2p+2} = \tau_\bu$ completes the proof of \Cref{lem:implicit-tensor-empirical-last-two-slices-close}.
    \end{proof}

By \Cref{lem:implicit-tensor-empirical-last-two-slices-close}, using a union bound over all $d^2$ pairs of indices $a,b$, and recalling from \Cref{fct:hermite-relu} that $|\wh{\sigma}_{2p+2}|\gtrsim 1/p^3$, conditioned on the event $\evt_0$, we have that  with probability at least $1 - 4d^2\tau_0$ over the fresh samples $\wh{\D}_{\mathrm{new}}(\cI)$ and the probing vector $\bu$, it holds:
    \begin{align*}
        \|\wh{\mM}_{2p+2}(\bu) - \mM_{2p+2}(\bu)\|_F\leq \frac{d\,\poly(k,1/\tau_0)2^{O(p)}}{N_0^{1/4}}. \end{align*}
    Then, using another union bound over all $N_\bu$ probing vectors from $\wh{\cU}$, we have that with probability at least $1 - d^2 N_\bu \tau_0$ (conditioned on $\evt_0$), it holds 
    \begin{align}\label{eq:max-u-matrix-close}
        \max_{\bu\in\wh{\cU}}\|\wh{\mM}_{2p+2}(\bu) - \mM_{2p+2}(\bu)\|_F\leq {d^2\,\poly(k,p,1/\tau_0)2^{O(p)}}/{N_0^{1/4}}.
    \end{align}
    Recall that $\evt_0$ happens with probability at least $1 - \tau_0$, thus, \eqref{eq:max-u-matrix-close} holds with probability at least $1 - (d^2 N_\bu + 1)\tau_0$ over all $N = (p+1)N_0$ samples and the probing vectors set $\wh{\cU}$.
    Recall that we have chosen $\tau_0 \lesssim \delta/(d^2 N_\bu)$, $p \lesssim \log(k/\lambda)$, $N_\bu = \poly(k/\lambda)\log(1/\delta)$ in \nameref{alg:implicit-tensor-approximate-matrix}. Hence, with $N_0 = \poly(d,k,N_\bu/\delta) 2^{O(p)} = \poly(d,k/\lambda,1/\delta)$, \eqref{eq:max-u-matrix-close} yields that with probability at least $1 - \delta/2$, 
    \begin{equation*}
        \max_{\bu\in\wh{\cU}}\|\wh{\mM}_{2p+2}(\bu) - \mM_{2p+2}(\bu)\|_2\leq \max_{\bu\in\wh{\cU}}\|\wh{\mM}_{2p+2}(\bu) - \mM_{2p+2}(\bu)\|_F \leq\frac{d^2\,\poly(k,p,1/\tau_0)2^{O(p)}}{N_0^{1/4}}\leq\sqrt{\xi}\lambda,
    \end{equation*}
    completing the proof of \Cref{prop:estimate-matrix-via-implicit-tensor}. 
\end{proof}

\subsection{Proof of the Main Result of the Section: 
Polynomial Complexity $\init$}\label{sec:proof-poly-initialization}
We are ready to prove \Cref{prop:poly-time-initialization-via-implicit-tensor}, which is restated below.

\ImplicitTensorPolyTimeInitialization*

\begin{proof}[Proof of \Cref{prop:poly-time-initialization-via-implicit-tensor}]
\nameref{alg:implicit-tensor-approximate-matrix} uses a total of $(p+1)N_0 = \poly(d,k/\lambda,1/\delta)$ samples. 
According to \Cref{thm:implict tensor adapted}, 
line~\ref{line:obtain-recursive pseudo-projection} takes $\poly(d,N_0)$ time. 
Further, by \Cref{cor:efficiently-computing-tensor-computation}, 
lines~\ref{line:projected-monent} and~\ref{line:projected-filter} (algorithms \nameref{alg:CompressEmpiricalMomentTensor} and \nameref{alg:CompressFilterTensor}) take runtime only $\poly(d,k,p)$. Thus, we conclude that the runtime of \nameref{alg:implicit-tensor-approximate-matrix} is $\poly(dN_0)$.

It remains to argue the correctness of \nameref{alg:implicit-tensor-approximate-matrix}.
Recall that by \Cref{lem:angle-separation-random}, fixing any $i^*\in[k]$ such that $\beta_{i^*}\geq\lambda$, the ``separation'' and ``correlation'' events 
    \begin{align*}
    \evt_{i^*}^{\mathrm{Sep}}\cap\evt_{i^*}^{\mathrm{Cor}} = \{\bu: (1+\xi)|\w_j\cdot\bu|\leq |\w_{i^*}\cdot\bu|,\,\forall j\in \cA_{i^*}(\xi)\}\cap\{\bu: |\w_{i^*}\cdot\bu|\geq (Z_p(\cU))^{1/(2p)}\}
\end{align*}
    happen with probability at least $\poly(\lambda/k)$ 
    over the choice of the random probing vector $\bu = \bp + \im \bq$, $\bp,\bq\sim\calN(\vec 0,\mI_d)$.
Since the number of probing vectors $N_\bu$ is chosen to be $N_\bu = \poly(k/\lambda)\log(k/\delta)$,
by the Chernoff bound, we know that drawing a set $\wh{\cU}$ of $N_\bu$ independent probing vectors $\bu$, with probability at least $1 - \delta/2$, it holds that: 
\begin{equation}\label{eq:matrix-approx-0}
    \text{for every $i^*\in[k]$ satisfying $\beta_{i^*}\geq \lambda$, }\exists\bu\in\wh{\cU},\,s.t.\,\bu\in\evt_{i^*}^{\mathrm{Sep}}\cap\evt_{i^*}^{\mathrm{Cor}}\;.
\end{equation}
On the other hand, by \Cref{prop:estimate-matrix-via-implicit-tensor} we know that under the choice of parameters specified in \nameref{alg:implicit-tensor-approximate-matrix},
    with probability at least $1 - \delta/2$ over all the samples from $\D$ 
    and the random vectors $\bu\in\wh{\cU}$, it holds
    \begin{align}\label{eq:matrix-approx-1}
        \max_{\bu\in\wh{\cU}}\|\wh{\mM}_{2p+2}(\bu) - \mM_{2p+2}(\bu)\|_2\leq \sqrt{\xi}\lambda.
    \end{align}
    Thus, the joint probability that these two events, \eqref{eq:matrix-approx-0} and \eqref{eq:matrix-approx-1} hold at the same time is at least $1 - \delta$. 
    Invoking \Cref{prop:low-degree-separates-direction}, we obtain that with probability at least 
    $1 - \delta$, for every $i^*\in[k]$ satisfying $\beta_{i^*}\geq\lambda$, there exists a vector 
    $\br_{i^*}$ in the list returned by \Cref{alg:implicit-tensor-approximate-matrix} that 
    satisfies $|\br_{i^*}\cdot\w_{i^*}|\geq 0.99$. Note that the size of the list is 
    $2N_\bu = \poly(k/\lambda)\log(1/\delta)$. 
    This completes the proof of \Cref{prop:poly-time-initialization-via-implicit-tensor}. 
\end{proof}

\section{Putting Everything Together: Proof of~\Cref{thm:main-informal}} \label{sec:main-alg}

By putting everything together, we obtain the following theorem, 
which is the formal version of \Cref{thm:main-informal}.

\begin{theorem}\label{thm:ultimate-main}
Let $d,k\in\mathbb{Z}_+$ and $\eps>0$. 
Let $(\x,y)\sim\D$ be a distribution supported on
$\R^d\times\R$ such that $\x\sim\calN(\vec 0,\mI_d)$ and
$y=f(\x)$ for an unknown $f\in\mathcal{C}_{d,k}$.
\nameref{alg:main}, using subroutine \nameref{alg:implicit-tensor-approximate-matrix} with parameters $\lambda = \poly(\eps/k)$, $\delta = \poly(\eps/k)$, has the following guarantee:
    with probability at least $2/3$, it outputs a hypothesis $\wh{y}(\x)$ satisfying $\E_{(\x,y)\sim\D}[(y - \wh{y}(\x))^2]\leq \eps$.
    \nameref{alg:main} uses $\poly(d,k,1/\eps)$ samples from $\D$, and runs in $\poly(d,k,1/\eps)$ time. 
\end{theorem}
\begin{proof}[Proof of \Cref{thm:ultimate-main}]
By \Cref{prop:optimization}, using a Coarse Direction Recovery procedure ($\init$) procedure 
satisfying \Cref{assum:initialization-algorithm} with parameters $\lambda = \poly(\eps/k)$ 
and $\delta = \poly(\eps/k)$, the algorithm \nameref{alg:main} outputs 
a hypothesis $\wh{y}(\x)$ with error $\Exy[(y - \wh{y}(\x))^2]\leq \eps$ with probability at least $2/3$. 
In addition, \nameref{alg:main} uses 
$$N = \widetilde{O}(d\,{\poly(k/\eps)} \, + \,L_\init^3(N_\init + L_\init^3)\poly(k/\eps))$$
samples from $\D$ and runs in time {$\poly(d, N, T_\init)$}, where $N_\init$, $T_\init$ and $L_\init$ are the sample complexity, runtime and the output list length of the $\init$ procedure respectively.

    On the other hand, in \Cref{prop:poly-time-initialization-via-implicit-tensor}, 
    we proved that \nameref{alg:implicit-tensor-approximate-matrix} is a $\init$ procedure 
    satisfying \Cref{assum:initialization-algorithm}, with sample complexity 
    $N_\init = \poly(d,k,1/\lambda,1/\delta)$, runtime 
    $T_\init = \poly(d,k,1/\lambda,1/\delta)$, 
    and list size $L_\init\leq \poly(k/\lambda)\log(1/\delta)$.

Thus, combining the results from \Cref{prop:optimization} 
and \Cref{prop:poly-time-initialization-via-implicit-tensor}, 
we obtain that our algorithm \nameref{alg:main}, 
using a $\init$ procedure \nameref{alg:implicit-tensor-approximate-matrix} 
with parameters $\lambda = \poly(\eps/k)$ and $\delta = \poly(\eps/k)$, 
with probability at least $2/3$ outputs 
a hypothesis $\wh{y}(\x)$ with $L_2^2$ error 
$\Exy[(y - \wh{y}(\x))^2]\leq \eps$.
In addition, considering the sample complexity, runtime, 
and the output list length of the $\init$ procedure 
\nameref{alg:implicit-tensor-approximate-matrix} 
under the choice of parameters $\lambda = \poly(\eps/k)$ and $\delta = \poly(\eps/k)$, 
it follows that \nameref{alg:main} uses $\poly(d,k,1/\eps)$ samples from $\D$, 
and runs in $\poly(d,k,1/\eps)$ time. 
This completes the proof of \Cref{thm:ultimate-main}.
\end{proof}

\appendix

\bibliographystyle{alphaabbr}
\bibliography{mydb}

\end{document}